\documentclass{article}
\usepackage{graphicx} % Required for inserting images
\usepackage{amsmath}
\usepackage{amsthm}
\usepackage{geometry}[margin=1inch]
\usepackage{physics}
\usepackage{authblk}
\usepackage{comment}
\usepackage{quantikz}
\usepackage{xurl}
\usepackage{cite}
\usepackage{hyperref}
\usepackage{subcaption}
\newcommand{\ver}{\nu}
\newcommand{\opqsp }{OP-QSP }

\newtheorem*{lemma*}{Lemma}
\newtheorem*{proposition*}{Proposition}

\newcommand{\qbinom}[3][]{\genfrac{[}{]}{0pt}{}{#2}{#3}_{#1}}

\theoremstyle{plain}
\newtheorem{theorem}{Theorem}
\newtheorem{definition}[theorem]{Definition}
\newtheorem{example}[theorem]{Example}

\newtheorem{remark}[theorem]{Remark}

\newtheorem{lemma}[theorem]{Lemma}
\newtheorem{corollary}[theorem]{Corollary}
\newtheorem{proposition}[theorem]{Proposition}

\usepackage{amsfonts}

\title{Analytical Angle-Finding and Series Expansions for Quantum Signal Processing via Orthogonal Polynomial Theory}
\author[1]{Pierre-Antoine Bernard\thanks{bernardpierreantoine@outlook.com}}
\author[1,2]{Nathan Wiebe}

\affil[1]{Department of Computer Science, University of Toronto, Toronto, ON, Canada}
\affil[2]{Pacific Northwest National Laboratory, Richland, WA, USA}

\date{}

\begin{document}
\maketitle

\begin{abstract}
    Quantum signal processing is a powerful framework in quantum algorithms, playing a central role in Hamiltonian simulation and related applications. The sequence of polynomials implemented at each step of this protocol provides a polynomial basis for block-encoding any polynomial of a unitary. We characterize the achievable polynomial bases in terms of their orthogonality or biorthogonality with respect to a linear functional admitting an integral representation. Explicit expressions for the quantum signal processing angles are derived for families of polynomial sequences, including Hermite, Jacobi, and Rogers–Szegő polynomials. We show that $2n+2$ rotation angles are required to encode a sequence of polynomials in these classes up to degree $n$. We use this result to show that an $\epsilon$-approximation of a smooth function $f$ can be block-encoded using $O(\log(1/\epsilon))$ gates via its Hermite series expansion. The connections established with the theory of orthogonal and biorthogonal polynomials lead to a new method for solving the quantum signal processing angle-finding problem, yielding explicit expressions for the angles. They also provide a complete characterization of the polynomials achievable by $\mathrm{SU}(1,1)$-QSP in terms of their roots. Biorthogonality properties are shown to hold in the bivariate QSP setting, yielding a set of necessary conditions for achievable polynomials.

\end{abstract}

\section{Introduction}

Quantum signal processing has revolutionized quantum algorithm design. It has led to optimal quantum simulation methods \cite{low2017optimal,low2019hamiltonian,liu2025toward}, efficient matrix inversion algorithms \cite{chakraborty2019power,lin2020optimal}, quantum eigenvalue and singular value transformation techniques \cite{gilyen2019quantum}, improved quantum control methods \cite{low2016methodology}, and new approaches to quantum sensing \cite{marrero2026encoded,dong2022beyond}. Quantum signal processing achieves this by implementing a polynomial transformation within two-dimensional irreducible subspaces of a unitary matrix, which it realizes through a sequence of interspersed applications of the unitary and rotations with angles tailored to the desired polynomial transformation \cite{low2017optimal,low2016methodology}. Since much of quantum algorithm design focuses on controlling two-dimensional subspaces, this method has acted as a disruptive force in the field.

%Quantum signal processing has revolutionized our approaches to quantum algorithm design.  Specifically, it has led to optimal quantum simulation methods, efficient matrix inversion algorithms, quantum eigenvalue and singular value transformation algorithms, improved quantum control methods and new approaches to quantum sensing.  Quantum signal processing achieves this by providing a way to apply a polynomial transformation within a two-dimensional irreducible subspace of a larger unitary matrix by constructing a sequence of interspersed calls to the unitary and rotations with angles tailored specifically for the desired polynomial transformation.  As much of quantum algorithm design focuses on controlling two-dimensional spaces, quantum signal processing has acted as a disruptive force in quantum algorithms.

A substantial issue with quantum signal processing is that numerical methods are often needed to find a particular polynomial transformation.  This means that, apart from Chebyshev polynomials which use trivial angle settings, we cannot analyze specifically how a given quantum signal processing algorithm works.  This issue is further confounded by the limitations of quantum signal processing, including its inability to transform non-unitary matrices and the challenges it faces generalizing to multivariate cases \cite{rossi2022multivariable,nemeth2023variants,gomes2024multivariable,laneve2025multivariate,rossi2025modular}.

%The key power of quantum signal processing derives from its ability to transform block encodings. Block-encodings are a potent unifying framework for the design and analysis of quantum algorithms. 

The key power of quantum signal processing can be understood through the unifying framework of block encodings. By introducing ancillary qubits and using post-selection, this framework enables the implementation of non-unitary linear transformations within an overall unitary quantum circuit. It underlies several cornerstone algorithms, including quantum linear systems algorithms \cite{harrow2009quantum, childs2017quantum, gilyen2019quantum}, quantum simulations~\cite{childs2012hamiltonian, low2017optimal}, differential equations~\cite{berry2017quantum,fang2023time,an2023linear}. 

The task of block-encoding polynomial transformations $P(U)$ of a given unitary $U$ has attracted substantial attention, as it constitutes a key subroutine in Hamiltonian simulation algorithms \cite{berry2015hamiltonian, low2017optimal}. A simple approach to this problem is provided by the Linear Combination of Unitaries (LCU) primitive~\cite{childs2012hamiltonian,berry2014exponential}, which enables block encoding of any matrix given its decomposition as a sum of unitaries, using two oracles that respectively encode the unitaries and the coefficients of the linear combination. Taking these unitaries to be monomials $U^i$ in $U$, this primitive allows one to block-encode any $P(U)$ given access to the coefficients in its monomial expansion.  This approach is highly flexible, but can require a large number of qubits to implement the polynomial transformation.

Quantum signal processing (QSP) provides a second approach. This protocol is inherently iterative: controlled applications of the target unitary are interleaved with single-qubit rotation gates acting on an ancillary qubit. This results in a block encoding of a polynomial $P(U)$, whose structure is fully determined by the sequence of rotation angles chosen at each step. Determining a sequence of angles that implements a desired target polynomial is a key preprocessing step and has been the subject of substantial research (see, e.g., \cite{haah2019product,chao2020finding,dong2021efficient,motlagh2024generalized}). While this preprocessing constitutes a limitation of the approach, QSP has the substantial advantage of requiring a constant number of ancillary qubits.  This is especially significant in the near term and early fault tolerant regime where qubits are expected to be expensive. 

This work gives an intermediate approach based on interpreting the QSP protocol as a procedure that block-encodes a polynomial basis. Since the QSP algorithm is iterative, after $i$ steps it produces a block encoding of a polynomial $P_i(U)$ of degree $i$. We thus propose an approach that uses a controlled version of this construction as part of an LCU-based protocol, so that any polynomial $P(U)$ can be block-encoded given access to its decomposition in the basis $\{P_i(U)\}_{i=0}^n$.

%A prominent technique built on the idea of block encoding is the Quantum Signal Processing (QSP) protocol, whose objective is to implement polynomial transformations of a given unitary, assuming controlled access to that unitary. The protocol is inherently iterative: controlled applications of the target unitary are interleaved with single-qubit rotation gates acting on an ancillary qubit. The specific polynomial realized by the circuit is entirely determined by the sequence of rotation angles chosen at each step. Determining a sequence of angles that implements a desired target polynomial constitutes a crucial preprocessing step and has been the subject of substantial research (see, for instance, \cite{dong2021efficient,motlagh2024generalized}). We refer to this task as the \emph{QSP angle-finding problem}. For the standard formulation of generalized quantum signal processing, it has been shown that an appropriate sequence of angles exists for any suitably normalized polynomial \cite{motlagh2024generalized}. The corresponding existence proof relies on the Fejér--Riesz theorem, which poses significant challenges for extensions beyond the univariate setting. In particular, this reliance complicates generalizations to multivariate quantum signal processing, where the goal is to construct polynomial functions of multiple unitaries \cite{rossi2022multivariable,laneve2025multivariate,nemeth2023variants,gomes2024multivariable}.

\begin{lemma*}[Informal statement of Lemma \ref{prop:sed2}]
    There exists an algorithm that uses the sequence of polynomials $\{P_i(U)\}_{i=0}^n$ obtained as block encodings at each step of an $n$-steps quantum signal processing protocol (or one of its variants) as a polynomial basis for block encoding any polynomial of degree at most $n$.
\end{lemma*}

While it is known that any suitably rescaled polynomial can be implemented using generalized quantum signal processing \cite{motlagh2024generalized}, a characterization of the classes of polynomial sequences $\{P_i(U)\}_{i=0}^n$ that can be realized within the QSP framework has not yet been fully established. Since actual application of the algorithm in the previous lemma depends on this question, we address it and provide a characterization in terms of orthogonality and biorthogonality properties of polynomial sequences. In particular, the sequences generated by generalized QSP are shown to be special cases of Laurent biorthogonal polynomials, while those arising in ${\rm SU}(1,1)$-QSP (a variant of QSP introduced in \cite{rossi2023quantum}) correspond with orthogonal polynomials on the unit circle. We also introduce a new variant, referred to as OP-QSP, whose associated polynomial sequences correspond to orthogonal polynomials, including orthogonal polynomials on the real line. This is summarized in Table~\ref{tab:1}.

 \begin{table}[t]
\centering
\begin{tabular}{|c|c|}
\hline
\textbf{QSP Variants} & \textbf{Polynomial Sequences} \\ \hline
\opqsp  & Orthogonal polynomials \\ \hline
Generalized QSP & Laurent biorthogonal polynomials (LBPs) satisfying  Proposition~\ref{prop:idlbps}\\ \hline
{\rm SU}(1,1)-QSP & Orthogonal polynomials on the unit circle (OPUC) \\ \hline
\end{tabular}
\caption{Characterization of polynomial sequences implemented by generalized quantum signal processing and other variants (see Propositions \ref{prop:corroprl1}, \ref{prop:idlbps}, and \ref{prop:corroprl12})}
\end{table}\label{tab:1}

The connection with the theory of orthogonal and biorthogonal polynomials leads to several examples of well-known polynomial families that can be implemented using variants of QSP and for which the corresponding rotation angles can be derived explicitly. A notable example is given by the Hermite polynomials, which are orthogonal with respect to an integral inner product involving a Gaussian weight. They can therefore be used as the polynomial sequence in the previous lemma, allowing the complexity of block-encoding a function $f$ to be analyzed through its Hermite expansion, which leads to the following result.

\begin{proposition*}[Informal statement of Proposition \ref{prop:herms}]
   Let $f(z)$ be a function that is square-integrable with respect to a Gaussian weight centered at $z=1$, with variance $1/\gamma$ (for $\gamma>1$), and such that its $k$-th derivative, when integrated against this weight, grows at most like $\sqrt{k!}/k^{1+3\gamma^2}$. Then a block encoding of $f(U)$ can be achieved with precision $\epsilon$ using $O(\log(1/\epsilon))$ depth and gate complexity, and with a block encoding constant that scales with the precision as $O(\log(1/\epsilon)^{1+3\gamma^2})$.
\end{proposition*}

The connection between ${\rm SU}(1,1)$-QSP and orthogonal polynomials on the unit circle also leads to a new characterization of the polynomials achievable by this algorithm. This characterization is expressed in terms of the roots of the polynomials and follows from the zero theorem for orthogonal polynomials on the unit circle \cite{barry2019constant}. For each QSP variant, we also present an approach to the angle-finding problem based on the computation of the moments of an associated linear functional. Finally, the biorthogonality properties of the polynomial sequences arising in bivariate QSP are investigated and shown to impose constraints on the class of polynomials that can be implemented.

The structure of the paper is as follows. In Section~\ref{sec:1}, we describe in terms of polynomial bases the block encoding achieved by the LCU protocol with a select oracle based on a controlled version of a QSP protocol. In Section~\ref{sec:2}, we introduce the \opqsp  variant and show that the sequence of polynomials generated at each step forms a family of orthogonal polynomials. The associated angle-finding problem is then solved by deriving explicit expressions for the angles that implement a target polynomial in terms of the moments of the corresponding linear functional. In Section~\ref{sec:3}, we consider the standard circuit for generalized quantum signal processing introduced in \cite{motlagh2024generalized}. We show that the sequence of polynomials implemented by this circuit forms a family of Laurent biorthogonal polynomials, and we solve the associated angle-finding problem in terms of the corresponding linear functional. In Section~\ref{sec:4}, we consider ${\rm SU}(1,1)$-QSP as introduced in \cite{rossi2023quantum} and establish its connection with the theory of orthogonal polynomials on the unit circle. Finally, in Section~\ref{sec:5}, we study bivariate quantum signal processing and describe two distinct types of biorthogonality satisfied by the polynomial sequences it generates. The resulting constraints on the class of achievable bivariate polynomials are then discussed.

\section{LCU and Polynomial Bases}\label{sec:1}

Consider an iterative algorithm for block-encoding a polynomial in a unitary $U$ acting on a given Hilbert space $\mathcal{H}$. Assuming access to $U$, or to a controlled version thereof, the algorithm implements at each step a unitary $W_j$ acting on $\mathcal{H}$ together with $\ell$ ancillary qubits. We assume that after $i$ steps the construction yields an $(\alpha_i, \ell, 0)$-block encoding of a monic degree-$i$ polynomial $\hat{P}_i$, namely
\begin{equation}\label{eq:bepoly}
(\bra{0}^{\otimes \ell} \otimes I)
\left( \prod_{j=1}^i W_j \right)
(\ket{0}^{\otimes \ell} \otimes I)
= \frac{\hat{P}_i(U)}{\alpha_i}.
\end{equation}
Throughout, we denote monic polynomials (i.e., polynomials whose leading coefficient is equal to 
1) with a hat. The simplest example is obtained by taking $\ell = 0$ and $W_i = U$ for all $i$. In this case, we have $\hat{P}_i(U) = U^i$ and block encoding constants $\alpha_i = 1$. Another example is obtained by taking $\ell = 1$ and choosing $W_j$ to be the controlled-$U$ operation together with the ancillary qubit rotation $R(\theta,\phi)$ appearing in a single step of the generalized QSP protocol (see Figure \ref{fig:circuitGQSP}). In this setting, the sequence of polynomials $\{\hat{P}_i(U)\}_{i = 0}^n$ depends on the chosen sequence of QSP angles. In general, we refer to any algorithm with block structure \eqref{eq:bepoly} as a QSP or QSP-like protocol, in the sense that it realizes a polynomial transformation of a unitary operator. 

\begin{figure}[t]
    \centering
\begin{quantikz}[wire types = {q,q,n,q}]
   & \octrl{1}& \gate{R(\theta,\phi)} & \\ &\gate[3]{U}& & \\
  &  &  \vdots&\\
  && &  \\
\end{quantikz}
    \caption{Circuit associated to generalized QSP (one iteration)}
    \label{fig:circuitGQSP}
\end{figure}
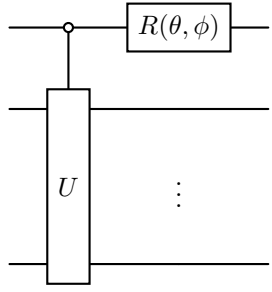

While such an algorithm only yields a block-encoding of the final polynomial $\hat{P}_n(U)$ in the sequence, the full family $\{\hat{P}_i(U)\}_{i=0}^n$ in fact spans the space of all polynomials of degree at most $n$. In other words, any polynomial $P(U)$ of degree $n$ admits a decomposition of the form
\begin{equation}\label{eq:neqc1}
P(U)
=
\sum_{i=0}^n v_i \frac{\hat{P}_i(U)}{\alpha_i}.
\end{equation}
Moreover, the expansion of functions in different polynomial bases is a well-established subject, which naturally motivates the development of a block-encoding method based on expansions such as \eqref{eq:neqc1}. Since this approach requires taking linear combinations of the polynomial block-encodings produced at each step of a QSP or QSP-like protocol, it can be implemented using the LCU primitive. The following lemma characterizes the block encoding obtained when a controlled version of a QSP or QSP-like algorithm is used as the select oracle within the LCU framework.

\begin{lemma}\label{prop:sed2}
  Let $U$ be a unitary acting on $\mathcal{H}$, and let $\{W_i\}_{i=1}^n$ be a sequence of unitaries acting on $(\mathbb{C}^2)^{\otimes \ell} \otimes \mathcal{H}$ with the block structure~\eqref{eq:bepoly}, where $\hat{P}_i$ is a monic polynomial of degree $i$ and $\alpha_i$ is a block encoding constant. Then an $(\|v\|_1, n +\ell , 0)$-block encoding of $P(U)$ with
\begin{equation}\label{eq:bepb}
P(U)
=
\sum_{i=0}^n v_i \frac{\hat{P}_i(U)}{\alpha_i},
\qquad
\|v\|_1 = |v_0| + |v_1| + \dots + |v_n|,
\end{equation}
is obtained via the LCU protocol with select oracle
$
U_{\mathrm{SEL}} = W_n^c W_{n-1}^c \cdots W_1^c,
$
defined in terms of the controlled versions $W_i^c$ of $W_i$ acting on $ \mathcal{H}_{anc}\otimes(\mathbb{C}^2)^{\otimes \ell} \otimes \mathcal{H}$ as
    \begin{equation*}
        W_i^c \ket{\underline{j}} \otimes \ket{\psi} = \left\{
	\begin{array}{ll}
		\ket{\underline{j}} \otimes W_i \ket{\psi}   & \mbox{if } j \geq i\\
		\ket{\underline{j}} \otimes \ket{\psi} & \mbox{if } j < i,
	\end{array}
\right.
    \end{equation*}
    with $\mathcal{H}_{anc} = \text{span}\{\, \ket{\underline{j}} \,|\, j = 0, 1, \dots, n\}$.
\end{lemma}
\begin{proof}
    This proposition follows from a simple application of the LCU lemma. Let $T$ be a matrix that admits a decomposition as a linear combination of unitary matrices $U_i$,
\begin{equation*}
T = \sum_{i=0}^n v_i U_i.
\end{equation*}
The linear combination of unitaries (LCU) primitive enables the construction of a block encoding of $T$ using preparation oracles together with a select oracle,
\begin{equation*}
    U_{\mathrm{SEL}} \ket{\underline{j}} \otimes \ket{\psi} = \ket{\underline{j}} \otimes U_i \ket{\psi} 
\end{equation*}
\begin{equation*}
    V_{\mathrm{prep}} \ket{\underline{0}}  = \frac{1}{\sqrt{\|v\|_1}} \sum_{i =0}^N \beta_i \ket{\underline{j}}, \quad  \bra{\underline{0}}\tilde{V}_{\mathrm{prep}}   = \frac{1}{\sqrt{\|v\|_1}} \sum_{i =0}^N \gamma_i \bra{\underline{j}}
\end{equation*}
with $|\beta_i| = |\gamma_i| = \sqrt{|v_i|}$ and $v_i = \beta_i \gamma_i$. The quantity $\|v\|_1 = |v_0| + \cdots + |v_n|$ denotes the $\ell_1$-norm of the vector with entries $v_i$. By the LCU lemma, the unitary obtained from the successive application of these oracles has the following block structure
\begin{equation*}
   ( \bra{\underline{0}} \otimes I ) (\tilde{V}_{\mathrm{prep}} \otimes I )U_{\mathrm{SEL}}(V_{\mathrm{prep}}\otimes I )( \ket{\underline{0}} \otimes I )   = \frac{T}{\|v\|_1}
\end{equation*}
For the select oracle $
U_{\mathrm{SEL}} = W_n^c W_{n-1}^c \cdots W_1^c
$, one obtains that $U_i = \prod_{j=1}^i W_j.$ Therefore, from the block structure in equation \eqref{eq:bepoly}, we obtain
\begin{equation*}
  ( \bra{\underline{0}} \otimes  \bra{0}^{\otimes \ell} \otimes I ) (\tilde{V}_{\mathrm{prep}} \otimes I )U_{\mathrm{SEL}}(V_{\mathrm{prep}}\otimes I )( \ket{\underline{0}} \otimes  \ket{0}^{\otimes \ell} \otimes I ) = \frac{1}{\|v\|_1}\sum_{i=0}^n v_i \frac{\hat{P}_i(U)}{\alpha_i}
\end{equation*}
The result follows from this, with the $n+\ell$ ancillary qubits for the block encoding arising from the $\ell$ ancilla qubits on which $W_j$ acts, together with the assumption of a simple unary encoding of $\ket{\underline{j}}$ using $n$ ancilla qubits.
\end{proof}
The previous proposition shows that the sequence of polynomials block-encoded at each step of a QSP protocol provides a basis for block-encoding any (properly rescaled) target polynomial $P(U)$. In the case $\ell = 0$ and $W_i = U$ for all $i$, the decomposition~\eqref{eq:bepb} of $P(U)$ is expressed in the monomial basis, with $\hat{P}_i(U) = U^i$. The associated select oracle where $\ket{\underline{i}}$ is identified with its unary encoding is presented in Figure \ref{fig:naiveLCUsel}.

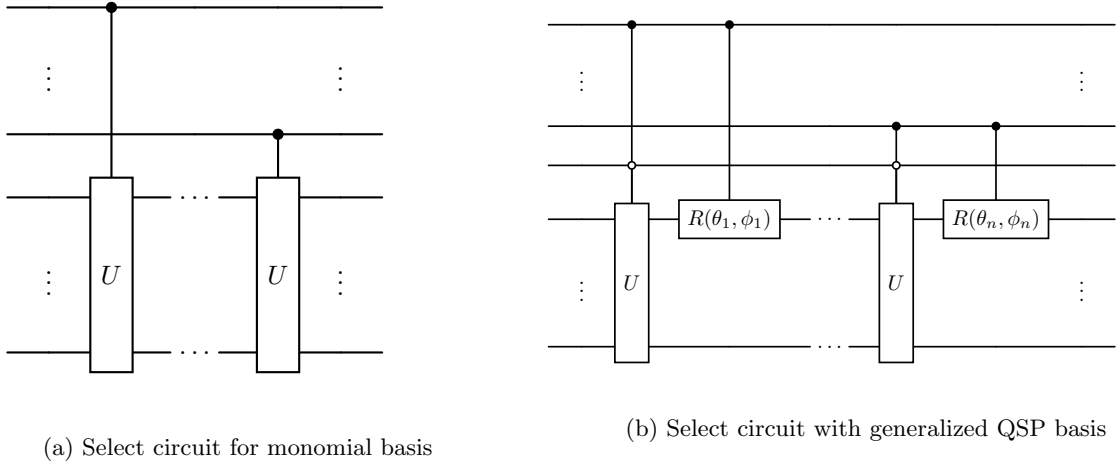
\begin{figure}[h]
    \centering
    % Left circuit: monomial basis, shifted right
    \begin{minipage}{0.45\textwidth}
        \centering
        \hspace*{-0.2\textwidth} % shift right by 5% of text width
        \begin{quantikz}[wire types={q,n,q,q,n,q}]
          & & \ctrl{3} &    &  & & \\
          & \vdots &    &  &  & \vdots & \\
          & & &   &  \ctrl{1} & & \\
          & & \gate[3]{U} & \ \ldots\  & \gate[3]{U} & & \\
          & \vdots & \dots &  & & \vdots &  \\
          & &  &  \ \ldots\  & & & \\
        \end{quantikz}
        \subcaption{Select circuit for monomial basis}
        \label{fig:naiveLCUsel}
    \end{minipage}
    \hfill
    % Right circuit: QSP basis scaled
    \begin{minipage}{0.45\textwidth}
        \centering
        \scalebox{0.8}{ % adjust scale factor as needed
        \hspace*{-0.2\textwidth}
        \begin{quantikz}[wire types={q,n,q,q,q,n,q}]
          & & \ctrl{4} & \ctrl{4} & & & & & \\
          & \vdots & & & & & & \vdots & \\
          & & & & & \ctrl{2} & \ctrl{2} & & \\
          & & \octrl{1} & & & \octrl{1} & & & \\
          & & \gate[3]{U} & \gate{R(\theta_1,\phi_1)} & \ \ldots\  & \gate[3]{U} & \gate{R(\theta_n,\phi_n)} & & \\
          & \vdots & & & & & & \vdots & \\
          & & & & \ \ldots\  & & & & \\
        \end{quantikz}
        }
        \subcaption{Select circuit with generalized QSP basis}
        \label{fig:LCUQSP}
    \end{minipage}
    \caption{Select circuits for LCU implementations of a polynomial block encoding: monomial basis (left) and generalized QSP based basis (right).}
\end{figure}

In the case based on the generalized QSP setting, the decomposition~\eqref{eq:bepb} of $P(U)$ is expressed with respect to a basis $\{\hat{P}_i(U)\}_{i=0}^n$ determined by the chosen sequence of QSP angles. The corresponding select oracle, where $\ket{\underline{i}}$ is identified with its unary encoding, is shown in Figure~\ref{fig:LCUQSP}.  If $|P(z)| \le 1$, it is known that the angles can be chosen such that $\hat{P}_n(U)/\alpha_n = P(U)$ \cite{motlagh2024generalized}, thereby recovering the standard QSP block encoding by setting $v_i = \delta_{i n}$. For arbitrary choices of angles, this approach to encoding $P(U)$ can be viewed as an intermediate strategy between an LCU method based on the monomial basis and the standard QSP construction, which requires a preprocessing step to determine the sequence of angles yielding $\hat{P}_n(U)/\alpha_n = P(U)$. Given a polynomial basis $\{\hat{P}_i(U)\}_{i=0}^n$ for which the corresponding QSP angles are known, one can test whether a target polynomial admits expansion coefficients $v_i$ leading to a sufficiently small block encoding constant $\|v\|_1$. If so, the circuit of Figure~\ref{fig:LCUQSP} can be used with these angles to implement the desired block encoding.\\

Proposition~\ref{prop:sed2} motivates a detailed analysis of the sequence of polynomials realized as block encodings at each step of QSP protocols. In particular, it naturally leads to the study of their orthogonality and biorthogonality relations, since such relations can be used to determine the coefficients $v_i$ in the decomposition~\eqref{eq:bepb} of a target polynomial $P(U)$. 

The next section introduces a QSP protocol for block encoding bases of orthogonal polynomials, including orthogonal polynomials on the real line. Orthogonality and biorthogonality properties for sequences arising in generalized QSP, continuous-variable QSP, and bivariate QSP protocols are discussed in the subsequent sections.

\section{QSP for orthogonal polynomials }\label{sec:2}

To illustrate the connection between QSP and orthogonality in polynomial sequences, we present in this section an example based on orthogonal polynomials (OP); see \cite{chihara2011introduction} for a reference on OP. The associated protocol is distinct from the standard and generalized QSP framework (which is covered in Section \ref{sec:3}) and will be referred to as \opqsp . %It is considerably less efficient than standard QSP, as it requires post-selection at each step, and should therefore be regarded primarily as an instructive example.

\subsection{Circuit and recurrence relation for \opqsp }

Let $U$ denote an arbitrary unitary acting on a finite-dimensional Hilbert space $\mathcal{H}$ and assume access to the doubly $0-$controlled version of $U$, $C^2(U)$, acting on $\mathbb{C}^2 \otimes \mathbb{C}^2 \otimes \mathcal{H}$. The circuit in Figure \ref{fig:circuitOPRL} block-encodes a non-unitary transformation on $\mathbb{C}^2 \otimes \mathcal{H}$ flagged by the top-most qubit being $0$.  The following lemma describes the transformation induced when $0$ is measured. 

\begin{figure}[t]
    \centering
    \scalebox{0.9}{
\begin{quantikz}[wire types = {q,q,q,n,q}]
    \lstick{$\ket{0}$}& \gate{X}& \gate[2]{R_{xy}(2(\tau- \pi))}&\gate{X} & \gate{R_y\left(2 \tan^{-1}(\cot(\omega)/\sqrt{2})\right)} & \gate[2]{R_{xy}\left(\frac{\pi}{2} \right)} & \octrl{1} &\gate{R_y(2\omega)} &  \qw \\
  & \gate{Z} & &  &  & & \octrl{1}  & &  \\
   & & & &  & &  \gate[3]{U} & & \\
 & \vdots& &  &  & \vdots & & \vdots&  \\
  &   & & &  &  & & & \\
\end{quantikz}}
    \caption{Circuit associated to \opqsp  (one iteration). $R_{xy}$ and $R_y$ denote respectively the $XY$ interaction and rotation around the $y$ axis gates as given in~\eqref{eq:RxyRy}.}
    \label{fig:circuitOPRL}
\end{figure}
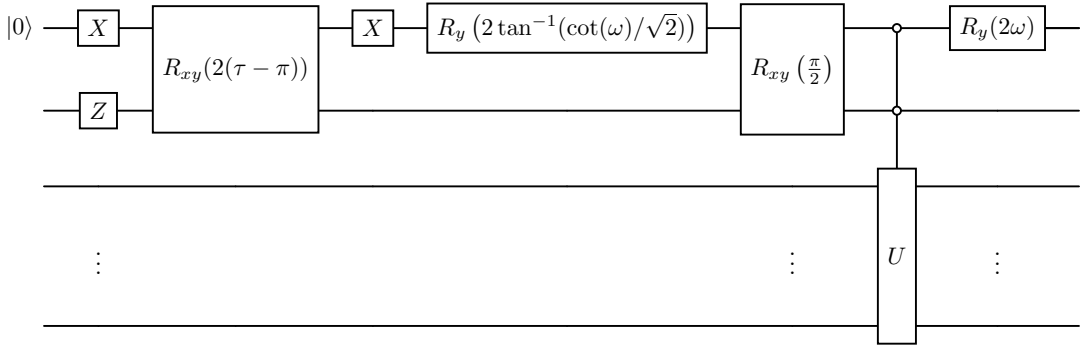

\begin{lemma}\label{lem:BE}
   Let $W$ be the unitary acting on $\mathbb{C}^2 \otimes \mathbb{C}^2 \otimes \mathcal{H}$ that implements the circuit in Figure \ref{fig:circuitOPRL}. %When viewed as a block matrix with respect to the first ancillary qubit, its $(0,0)$-block is given by
   The $W$ transformation enacts for any $\tau\in (-\pi/2,\pi/2)$ and $\omega\in (-\pi,\pi)\setminus \{0, \pm \frac{\pi}{2}\}$ an $(\alpha(\tau,\omega),1,0)$-block encoding of a matrix $T(\tau, \omega)$, i.e.
    \begin{equation*}
        (\bra{0}\otimes I)W(\tau,\omega) (\ket{0}\otimes I)  =\frac{T(\tau, \omega)}{\alpha(\tau, \omega)},
    \end{equation*}
    where
    \begin{equation}\label{eq:matOPRL}
     T(\tau,\omega) = \begin{pmatrix}
        U - \frac{\tan (\tau) + 1}{2} &  \frac{i \tan (\omega ) }{\sqrt{2}\cos (\tau )}  \\
\frac{i \tan (\tau )}{2}  \left(2 \tan (\omega )+ \cot (\omega )\right)- \frac{i\cot (\omega )}{2}  & 0  \end{pmatrix}, \quad \alpha(\tau, \omega) =    \frac{1}{\cos \tau }\sqrt{\frac{2 (1 + \sin^2 \omega)}{ \sin^2 2\omega }}.
\end{equation}
\end{lemma}
\begin{proof}
The gates $C^2(U)$, $R_{xy}$, and $R_y$ admit the matrix representation $C^2(U) = \text{diag}(U, 1,1,1)$ and
\begin{equation}
   R_{xy}(\theta) =  
\begin{pmatrix}
 1 & 0 & 0 & 0 \\
 0 & \cos \left(\frac{\theta }{2}\right) & -i \sin \left(\frac{\theta }{2}\right) & 0 \\
 0 & -i \sin \left(\frac{\theta }{2}\right) & \cos \left(\frac{\theta }{2}\right) & 0 \\
 0 & 0 & 0 & 1 \\
\end{pmatrix}, \quad R_y(\theta) =
\begin{pmatrix}
 \cos \left(\frac{\theta }{2}\right) & -\sin \left(\frac{\theta }{2}\right)  \\

 \sin \left(\frac{\theta }{2}\right)  & \cos \left(\frac{\theta }{2}\right) 
\end{pmatrix}.\label{eq:RxyRy}
\end{equation}
By linearity of quantum operations and standard matrix multiplication, the transformation induced by the last four gates of the circuit has the following structure in its first $2 \times 2$ block:
{\small
\begin{equation*}
    (R_{y}( 2 \omega) \otimes I) C^2(U) R_{xy}\left(\frac{\pi}{2}\right) (R_y\left(2 \tan^{-1}(\cot(\omega)/\sqrt{2})\right)\otimes I) = 
\begin{pmatrix}
 \frac{(2 U-1) \cos (\omega )}{\sqrt{2} \sqrt{\cot ^2(\omega )+2}} & \frac{i \sin (\omega )}{\sqrt{\cot ^2(\omega )+2}} & \cdot& \quad \cdot \\
 -\frac{i \cos (\omega ) \cot (\omega )}{\sqrt{2} \sqrt{\cot ^2(\omega )+2}} & 0 & \cdot & \quad \cdot\\
  \cdot & \cdot & \cdot& \quad \cdot\\
  \cdot & \cdot & \cdot& \quad \cdot\\
\end{pmatrix}
\end{equation*}}
In particular, the angles in the two $R_y$ rotations are related so that the second diagonal entry vanishes. The first three gates $(X \otimes I) R_{xy}(2\tau)(X \otimes Z)$ mix the first and last columns, while leaving the second and third columns unchanged. Composing the action of all gates, we find that the first block of the matrix representation of the circuit coincides with the r.h.s of equation \eqref{eq:matOPRL}, thereby completing the proof.
\end{proof}

This lemma allows one to determine how a state transforms under multiple iterations of the circuit. If the circuit is applied $i$ times, each with a distinct first ancillary qubit and angles $\tau_j$ and $\omega_j$, the resulting transformation exhibits a block structure given by the product of the matrices $T(\tau_j,\omega_j)$:
\begin{equation}\label{eq:bsoprl}
    (\bra{0}^{\otimes i} \otimes I ) W_i W_{i-1} \dots W_{1} (\ket{0}^{\otimes i} \otimes I ) = \frac{T(\tau_i,\omega_i)T(\tau_{i-1},\omega_{i-1}) \dots T(\tau_{1},\omega_{1})}{\alpha(\tau_i,\omega_i)\alpha(\tau_{i-1},\omega_{i-1}) \dots \alpha(\tau_{1},\omega_{1})}
\end{equation}
where $W_j$ denotes the unitary $W(\tau_j, \omega_j)$ acting such that the first ancillary qubit of the circuit corresponds to the $j$-th ancillary qubit of the overall register. Since the matrices $T$ depend linearly on $U$, their product yields a block encoding of a polynomial of degree $i$ in $U$,
\begin{equation}\label{eq:block encoding}
   \begin{bmatrix}
       \hat{P}_i(U) & \cdot \\
       Q_i(U) & \cdot
   \end{bmatrix} 
   = T(\tau_i,\omega_i,)\dots T(\tau_1,\omega_1).
\end{equation}
Therefore, the transformation associated with $i$ iterations of the circuit in Figure~\ref{fig:circuitOPRL} exhibits exactly the block structure described in equation~\eqref{eq:bepoly},with $\ell = i + 1$. The following proposition characterizes the sequence of polynomials realized as block encodings in terms of a recurrence relation. 
\begin{proposition}\label{prop:qsp-oprl}
     Let $\{\hat{P}_i(x)\mid i=0,1,\ldots,n\}$ be a sequence of monic polynomials. The following statements are equivalent:
     \begin{enumerate}
     
        \item The polynomials satisfy the following three-term recurrence relation, with $\deg \hat{P}_i=i$, $a_i^2,\, b_i \in \mathbb{R}$ and $a_i \neq 0 $,
        \begin{equation}\label{eq:3tr}
    \hat{P}_{i+1}(x) = (x-b_i)\hat{P}_i(x) - a_i^2 \hat{P}_{i-1}(x).
\end{equation}
     
         \item There exist angles $\{\omega_i  \}_{i=1}^n$ and $\{\tau_i \}_{i=1}^n$ such that $\tau_i\in (-\pi/2,\pi/2)$ and $\omega_i\in (-\pi,\pi)\setminus \{0, \pm \frac{\pi}{2}\}$ and such that, for each $i = 0, 1, \dots, n$, an $(\alpha_i, n + 1, 0)$-block encoding of $ \hat{P}_i(U)$ is realized after $i$ iterations of the circuit shown in Figure~\ref{fig:circuitOPRL}, with block encoding constants $\alpha_i$ given by
         \begin{equation}\label{eq:beconst}
    \alpha_i =  \prod_{j = 1}^i \alpha(\tau_j, \omega_j) =\prod_{j = 1}^i \left( \frac{1}{\cos \tau_j}\sqrt{\frac{2(1 + \sin^2 \omega_j)}{ \sin^2 (2\omega_j)}}\right).
\end{equation}
     \end{enumerate} 
\end{proposition}
\begin{proof}
    From equation~\eqref{eq:block encoding} and the structure of the matrix $T(\tau, \omega)$, it follows that the sequence of polynomials block-encoded at each iteration of the circuit shown in Figure~\ref{fig:circuitOPRL} satisfies
    \small
    \begin{align}
    \begin{bmatrix}
        \hat{P}_{i}(U)& \cdot \\
        Q_{i}(U)&\cdot 
    \end{bmatrix} = T(\tau_i, \omega_i) \begin{bmatrix}
        \hat{P}_{i-1}(U)& \cdot \\
        Q_{i-1}(U)&\cdot 
    \end{bmatrix}=\begin{bmatrix}
(U - \frac{\tan \tau_i + 1}{2}) \hat{P}_{i-1}(U)+   i \frac{\tan \omega_i}{\sqrt{2}\cos{\tau_i}}Q_{i-1}(U) & \cdot\\  \left(\frac{i \tan (\tau_i )}{2}  \left(2 \tan (\omega_i )+ \cot (\omega_i )\right)- \frac{i\cot (\omega_i )}{2}\right)\hat{P}_{i-1}(U) & \cdot
    \end{bmatrix}\label{eq:recur}.
\end{align}
\normalsize
This implies that $Q_i$ and $\hat{P}_{i-1}$ are linearly dependent,
\begin{equation}\label{eq:QiXPR}
    Q_i(U) = \left(\frac{i \tan (\tau_i )}{2}  \left(2 \tan (\omega_i )+ \cot (\omega_i )\right)- \frac{i\cot (\omega_i )}{2}\right) \hat{P}_{i-1}(U).
\end{equation}
Substituting \eqref{eq:QiXPR} into~\eqref{eq:recur} leads to the three term recurrence relation \eqref{eq:3tr} for $\hat{P}_i$ (now seen generically as a function of $x$), where the coefficients $b_i$ and $a_i$ are determined by the angles:
\begin{equation}\label{eq:recangle}
    b_i = \frac{ \tan\tau_{i+1} +1 }{2}, \quad
    a_i^2 = \frac{\tan (\omega_{i+1} ) }{ \sqrt{2}\cos (\tau_{i+1} )}   \left(\tan (\tau_{i} ) \tan (\omega_{i} )+ \frac{(\tan (\tau_{i} ) - 1)\cot (\omega_{i} )}{2} \right).
\end{equation}
This shows that the second statement implies the first. To prove the converse, we note that the relations between the angles $\omega_i$ and $\tau_i$ and the recurrence coefficients $a_i$ and $b_i$ can be inverted. Specifically, the angles $\tau_i \in (-\pi/2, \pi/2)$ can be expressed in terms of the coefficients $b_i$, while the angles $\omega_i$ are determined via the recurrence relation:
\begin{equation}\label{eq:inv1}
    \tau_{i+1} = \tan^{-1}\left(2 b_i - 1\right),
\end{equation}
\begin{equation}\label{eq:inv2}
    \omega_{i+1} = \cot^{-1}\left(\frac{\sqrt{2b_i^2 - 2 b_i + 1}}{a_i^2}\left(\left(2 b_{i-1} - 1\right) \tan (\omega_{i} )+ {\left( b_{i-1} - 1\right)\cot (\omega_{i} )}\right)\right)
\end{equation}
This thus shows that the first statement implies the second.
\end{proof}
The relations~\eqref{eq:recangle}, \eqref{eq:inv1}, and \eqref{eq:inv2} between the recurrence coefficients and the angles will play an important role in the following. They provide a systematic method to determine the angles that generate a sequence of polynomials with a known recurrence relation. In the special cases $b_i=\tfrac{1}{2}$ or $b_i=1$, the recurrence relation simplifies substantially, and a closed-form expression can be obtained, depending on the first angle $\omega_1$ as a free parameter:
\begin{equation}\label{eq:inv3}
    \omega_i = \cot^{-1}\left(\cot \omega_1 \prod_{j=1}^{i-1} \frac{-1}{2^{3/2}a_j^2} \right) \quad (b_i = 1/2),
\end{equation}
\begin{equation}\label{eq:inv4}
    \omega_{i} = \left\{
	\begin{array}{ll}
		\cot^{-1}\left( \cot \omega_1 \prod_{j=1}^{\lfloor \frac{i}{2} \rfloor} \frac{a_{2j-1}^2}{ a_{2j}^2 }  \right)  & \mbox{if } i \text{ odd} \\
		  \cot^{-1}\left( \frac{1}{ a_1^2 \cot \omega_1} \prod_{j=1}^{\frac{i-2}{2}} \frac{a_{2j}^2 }{ a_{2j+1}^2 }  \right) & \mbox{if } i \text{ even}
	\end{array}
\right.  \quad (b_i = 1).
\end{equation}

\noindent Note that for any angles $\tau$ and $\omega$, one has $\alpha(\tau,\omega) > 1.7$, so that the normalization constant in the previous proposition satisfies $\alpha_i > (1.7)^i$ and therefore grows exponentially with the degree. This behavior is expected, since the three-term recurrence relation \eqref{eq:3tr} can generate polynomials whose coefficients and whose values on the unit circle scale exponentially with the polynomial's degree, and the fact that the block encoding is embedded in a unitary implies that $|\hat{P}_i(U)/\alpha_i| \le 1$, so that $\alpha_i$ must compensate for this exponential growth. Using the same reasoning, this inequality can also be applied to derive a bound on the solutions of equation \eqref{eq:3tr}. This observation is formally stated in the following corollary.
\begin{corollary}
   Let $\{\hat{P}_i(x)\}_{i=0}^n$ be a sequence of monic polynomials satisfying the three-term recurrence relation \eqref{eq:3tr}, with $\deg \hat{P}_i = i$ and coefficients $a_i^2 \in \mathbb{R}\backslash\{0\}$ and $b_i \in \mathbb{R}$. Let $\alpha_i$ be the constant defined in \eqref{eq:beconst}. Then, for all $z \in \mathbb{T}$ and $i = 0,1,\dots,n$, we have 
   \begin{equation*}
       |\hat{P}_i(z)|< \alpha_i.
   \end{equation*}
\end{corollary}

The previous proposition and its derivation illustrates how the iterative nature of a QSP protocol implies a recurrence relation for the sequence $\{\hat{P}_i\}_{i=0}^n$ of polynomials implemented at each iteration of Figure \ref{fig:circuitOPRL}. The recurrence relation \eqref{eq:3tr} further allows a characterization of the sequences produced by the circuit in terms of an orthogonality relation. This topic is addressed in the next subsection.

\subsection{Relation between OP-QSP and orthogonal polynomials}

We now recall several definitions and theorems from the classical theory of orthogonal polynomials (OP) and orthogonal polynomials on the real line (OPRL). Let $\mathcal{L}$ be a complex-valued linear functional on the space of all polynomials, and let its moments be denoted by
\begin{equation*}
    c_i = \mathcal{L}[x^i].
\end{equation*}  

\begin{definition}
    A sequence of polynomials $\{P_0(x), P_1(x), \dots, P_i(x), \dots\}$ over the $\mathbb{R}$ is said to be \textit{orthogonal with respect to $\mathcal{L}$} if there exists constants $\chi_i\in \mathbb{R}$ such that
\begin{equation*}
    \mathcal{L}[P_i(x) P_j(x)]  = \chi_i \delta_{ij}.
\end{equation*} 
\end{definition}
A key concept to classify the linear functionals is the determinant of the matrix of moments:
\begin{equation*}
   h_{i} := \begin{vmatrix}
    c_0 & c_{1} & \dots & c_{i}\\
    c_{1} & c_{2} & \dots& c_{i+1}\\
    \vdots & \vdots & \ddots &  \vdots \\
    c_{i} & c_{i+1} & \dots & c_{2i}
    \end{vmatrix}.
\end{equation*}
The following characterization of the linear functional $\mathcal{L}$ is essential for the subsequent theorems on orthogonal polynomials.
\begin{definition}
    A linear functional $\mathcal{L}$ is said to be \textit{positive-definite} if $c_i\in \mathbb{R}$ and $h_i > 0$ for all $i$ and \textit{quasi-definite} if $h_i \neq 0$ for all $i$. 
    \end{definition}
    Theorem 3.1 in Chapter 2 of \cite{chihara2011introduction} states that all positive-definite functionals admit an integral representation with respect to a distribution function $\mu(x)$ defined on the real line,
\begin{equation}\label{eq:intrep}
    \mathcal{L}[f(x)] = \int_{\mathbb{R}}f(x) d\mu(x).
\end{equation}
 Furthermore, Theorem 3.1 in Chapter 1 of \cite{chihara2011introduction} states that quasi-definite functionals admit a sequence of orthogonal polynomials. The following theorem describes the recurrence relation they satisfy.

\begin{theorem}\label{theo:oprl}
    (reformulation of Theorem 4.1 in Chapter 1 of \cite{chihara2011introduction}) Let $\{\hat{P}_i(x)\mid i=0,1,\ldots\}$ be a sequence of monic orthogonal polynomials with respect to a quasi-definite linear functional $\mathcal{L}$. Then there exist non-zero coefficients $a_i$ and $b_i$ such that the polynomials satisfy the recurrence relation \eqref{eq:3tr}. Moreover, if $\mathcal{L}$ is positive-definite then the coefficients $a_i$ and $b_i$ are real.
\end{theorem}
In the context of \opqsp  and in view of Proposition~\ref{prop:qsp-oprl}, the preceding theorem implies that sequences of polynomials orthogonal with respect to a quasi-definite linear functional provide explicit examples of polynomial sequences that can be implemented by the protocol. Moreover, one can show that any sequence produced by this protocol is orthogonal with respect to a linear functional. Indeed, the celebrated Favard theorem states that the converse of Theorem~\ref{theo:oprl} also holds: any sequence of polynomials satisfying a three-term recurrence relation of the form~\eqref{eq:3tr} is orthogonal with respect to a linear functional.
\begin{theorem}
    (Favard's Theorem) Let $\{\hat{P}_i(x)\mid i=-1,0,1,\ldots\}$ be the sequence of monic polynomials defined as the solution of the recurrence relation 
    \begin{equation*}
    \hat{P}_{i+1}(x) = (x-b_i)\hat{P}_i(x) - a_i^2 \hat{P}_{i-1}(x),
\end{equation*}
    with initial conditions $\hat{P}_{-1}(x) = 0$ and $\hat{P}_0(x) = 1$. Then there exists a linear functional $\mathcal{L}$ and constants $\chi_i\in \mathbb{R}$ such that
\begin{equation*}
    \mathcal{L}[1] = a_0^2, \quad \mathcal{L}[\hat{P}_i(x) \hat{P}_j(x)] =\chi_i \delta_{ij}.
\end{equation*}
Moreover, $\mathcal{L}$ is quasi-definite if and only if $a_i \neq 0$, and it is positive definite if and only if $a_i \in \mathbb{R}\backslash\{0\}$ and $b_i \in \mathbb{R}$.
\end{theorem}

Thus, this yields a characterization of the polynomial sequences produced by iterating the circuit shown in Figure~\ref{fig:circuitOPRL}: they are precisely the sequences of polynomials that are orthogonal with respect to a quasi-definite linear functional. This includes  the case of all OPRL, which corresponds to the case of a positive-definite linear functional $\mathcal{L}$.

\begin{proposition}\label{prop:corroprl1}
    Let $\{\hat{P}_i(x)\mid i=0,1,\ldots\}$ be a sequence of monic polynomials. There exist angles $\{\omega_i \}_{i=1}^n$  and $\{\tau_i \}_{i=1}^n$  and constants $\alpha_i$ such that $\tau_i\in (-\pi/2,\pi/2)$ and $\omega_i\in (-\pi,\pi)\setminus \{0, \pm \frac{\pi}{2}\}$ and such that, for each $i = 0, 1, \dots, n$, an $(\alpha_i, n+1, 0)$-block encoding of $ \hat{P}_i(U)$ is realized after $i$ iterations of the circuit shown in Figure~\ref{fig:circuitOPRL} if and only if $\deg \hat{P}_i=i$ and the polynomials are orthogonal with respect to a quasi-definite linear functional $\mathcal{L}$.
\end{proposition}
\begin{proof}
    This result follows directly from the three-term recurrence in Proposition~\ref{prop:qsp-oprl} and Favard's Theorem.
\end{proof}
\noindent We now provide some examples based on well-known families of orthogonal polynomials. 
 \begin{example}
        Consider the monic polynomials $\hat{P}_i(x) = (2\gamma)^{-i} U_i(\gamma x - \kappa)$, defined in terms of the Chebyshev polynomials of the second kind $U_i$ and real parameters $\gamma$ and $\kappa$. Using the orthogonality relation of the Chebyshev polynomials, we obtain
        \begin{equation*}
            \mathcal{L}[\hat{P}_n(x)\hat{P}_m(x)] = \int_{\frac{\kappa -1}{\gamma}}^{\frac{\kappa +1}{\gamma}} \hat{P}_n(x)\hat{P}_m(x) \sqrt{1-(\gamma x- \kappa)^2}  dx = \frac{\pi }{(2\gamma)^{2n + 1}}\delta_{mn}.
        \end{equation*}
        Using the known three-term recurrence relation for $U_i$, we obtain the following recurrence relation for the polynomials $\hat{P}_i(x)$:
        \begin{equation*}
            \hat{P}_{i+1}(x) = \left(x -\frac{\kappa}{\gamma}\right) \hat{P}_i(x) - \frac{1}{4 \gamma^2} \hat{P}_{i-1}(x).
        \end{equation*}
       The angles $\tau_i$ and $\omega_i$ to implement a block encoding of $ \hat{P}_i(U)/\alpha_i$ are then obtained by solving equation~\eqref{eq:recangle} with $b_i = \kappa\gamma^{-1}$ and $a_i = (2\gamma)^{-1}$. For $2 \kappa < \gamma \leq \kappa$ or $\kappa \leq \gamma \leq 2\kappa$, one finds a solution corresponding to fixed angles $\tau_i = \tau$ and $\omega_i = \omega$ where
        \begin{equation*}
            \tau = \tan^{-1}\left(\frac{2\kappa}{\gamma} - 1\right), \quad  \omega = \tan^{-1}\left( \pm \sqrt{ \frac{\kappa-\gamma}{\gamma - 2 \kappa}- \frac{\mathrm{sign}(\gamma)}{4(  \gamma- 2\kappa)\sqrt{2\kappa^2 - 2\kappa \gamma + \gamma^2}}}\right).
        \end{equation*}
    Using these angles as input in equation \eqref{eq:beconst} yields the following normalization constant that grows exponentially with $i$:
        \begin{equation*}
    \alpha_i  = \left( \frac{1}{\cos \tau}\sqrt{\frac{2(1 + \sin^2 \omega)}{ \sin^2 (2\omega)}}\right)^i,
\end{equation*}
    \end{example}
    
\begin{example}\label{ex:twocheb}
Consider the sequence of monic polynomials defined by $\hat{P}_0(x)=1$, $\hat{P}_1(x)=x-1$, and for $i \geq 2$, 
\begin{equation}
    \hat{P}_{i}(x) = \frac{1}{(2\gamma)^i}U_i(\gamma(x-1)) + \left(\frac{1}{4 \gamma^2} - a^2\right) \frac{1}{(2\gamma)^{i-2}}U_{i-2}(\gamma(x-1)),
\end{equation}
where $a$ and $\gamma$ are fixed parameters and $U_i$ denotes the Chebyshev polynomials of the second kind. From the three-term recurrence relation satisfied by $U_i$, one can verify that they obey the recurrence
    \begin{equation*}
        \hat{P}_{i+1}(x) = (x-1)\hat{P}_{i}(x) - a_i^2\hat{P}_{i-1}(x), \quad a_i = \left\{
	\begin{array}{ll}
		a  & \mbox{if } i = 1 \\
		\frac{1}{2\gamma} & \mbox{if } i > 1.
	\end{array}
\right.
    \end{equation*}
    Using the identity $U_n\left( \frac{z}{2} + \frac{1}{2z}\right) = \frac{z^{n+1} -z^{-n-1}}{z -z^{-1}}$, one can show that these polynomials satisfy the orthogonality relation
    \begin{equation*}
        \mathcal{L}[\hat{P}_n(x)\hat{P}_m(x)] = \frac{\delta_{nm}}{(2 \gamma)^{2n}}
    \end{equation*}
where the linear functional $\mathcal{L}$ is defined by the contour integral,
\begin{equation*}
   \mathcal{L}[f] = \frac{-1}{4 \pi i}\int_{C}  f\left(\frac{z + z^{-1}}{2\gamma} + 1\right)\frac{(1- z^2)^2}{z(z^2 + \lambda)(1+ \lambda z^2 )} dz, \quad \lambda = 1 - 4 \gamma^2 a^2.
\end{equation*}
Here, $C$ is a contour enclosing $z=0$ and $z=\pm i\sqrt{\lambda}$, while excluding $z=\pm \frac{1}{i\sqrt{\lambda}}$. A proof of this orthogonality relation is provided in Appendix~\ref{sec:app1}. The angles $\omega_i$ and $\tau_i$ required to implement a block encoding of $ \hat{P}_i(U)/\alpha_i$ are obtained from Equations \eqref{eq:inv1} and \eqref{eq:inv4}, using $a_1=a$ and $a_i=\frac{1}{2\gamma}$ for $i>1$. This yields
\begin{equation*}
    \tau_i = \frac{\pi}{4},\quad \cot \omega_{2i + 1} = \frac{1}{a^2 \cot \omega_1}, \quad  \cot \omega_{2i} =  \frac{a^2 \cot \omega_1}{4 \gamma^2} .
\end{equation*}
  Using these angles as input in equation \eqref{eq:beconst}, we find that the associated block encoding constant $\alpha_n$ grows exponentially with $n$ and is given by
  \small
\begin{equation*}
    \alpha_{n} = \left(\cot^2\omega_1 + \frac{2}{\cot^2 \omega_1} + 3\right)^{\frac{1}{2}}\left(\frac{1}{a^4 \cot^2 \omega_1} + {2a^4 \cot^2 \omega_1} + 3\right)^{\frac{1}{2}\lfloor\frac{n}{2} \rfloor} \left(\frac{32 \gamma^4}{ a^4 \cot^2 \omega_1} + \frac{a^4 \cot^2 \omega_1}{16 \gamma^4 } + 3\right)^{\frac{1}{2}\lfloor\frac{n-1}{2} \rfloor}
\end{equation*}

\end{example}

\begin{example}\label{ex:herm}
    Consider the monic polynomials $\hat{P}_n(x) = \gamma^{-n} H_n\big(\gamma (x-1)\big)$, defined in terms of the Hermite polynomials and a real parameter $\gamma$. Using the orthogonality and three-term recurrence relations of the Hermite polynomials, one finds that  they satisfy
 \begin{equation}\label{eq:orher}
    \mathcal{L}[\hat{P}_n(x)\hat{P}_m(x)] = \int_{\mathbb{R}} \hat{P}_n(x)\hat{P}_m(x) e^{-\gamma^2(x-1)^2/2} \gamma dx =  \frac{\sqrt{2\pi}  n! \delta_{nm}}{\gamma^{2n}}
\end{equation}
and
\begin{equation*}
    \hat{P}_{n+1}( x) = (x-1) \hat{P}_n(x) - n \gamma^{-2} \hat{P}_{n-1}(x)
\end{equation*}
The angles $\tau_i$ and $\omega_i$ required to implement a block encoding of this polynomial sequence are obtained by solving equations~\eqref{eq:inv1} and~\eqref{eq:inv4}, with $b_n = 1$ and $a_n^2 = n \gamma^{-2}$. For $\cot \omega_1 = \gamma \sqrt{\frac{2}{\pi}}$, this yields $\tau_n = \pi/4$ and
\begin{equation*}
    \cot \omega_n = \left\{
	\begin{array}{ll}
		 \sqrt{\frac{2}{\pi}}\frac{\gamma}{2^{n-1}} \binom{n-1}{\frac{n-1}{2}}  & \mbox{if } n \text{ odd} \\
		 \sqrt{\frac{\pi}{2}}\frac{ \gamma 2^{n-2}}{ n-1} \binom{n-2}{\frac{n-2}{2}}^{-1} & \mbox{if } n \text{ odd}
	\end{array}
\right.
\end{equation*}
In particular, it follows from Stirling's approximation that as $n$ grows
\begin{equation*}
    \cot\omega_n = \frac{\gamma}{\sqrt{n}}\left(1 - \frac{1}{4n} + O(1/n^2)\right)
\end{equation*} From these angles and equation \eqref{eq:beconst}, we find that the block encoding constants $\alpha_n$ grow with $n$ as
\begin{equation}\label{eq:asyher}
    \alpha_n = \frac{  \sqrt{n!} }{\gamma^{n}} n^{\frac{1}{2} + 3\gamma^2 + O(1/n)}
\end{equation}

\end{example}

\begin{example}
    The monic Jacobi polynomials $\hat{P}_i^{(\lambda,\beta)}$ are known to satisfy the orthogonality relation
    \begin{equation}
       \mathcal{L}[\hat{P}^{(\lambda,\beta)}_n(x)\hat{P}^{(\lambda,\beta)}_m(x)] = \int_{-1}^1(1-x)^\lambda (1-x)^\beta \hat{P}_m^{(\lambda,\beta)}(x) \hat{P}_n^{(\lambda,\beta)}(x) dx = \chi_n \delta_{mn}
    \end{equation}
    with 
    \begin{equation}
        \chi_n = \frac{2^{2n + \lambda + \beta + 1} \Gamma(n+\lambda+1)\Gamma(n+\beta+1)\Gamma(n+\lambda + \beta +1) n!}{\Gamma(2n+\lambda+ \beta+2)\Gamma(2n+\lambda+\beta+1)},
    \end{equation}
    and the following three-term recurrence relation,
\begin{equation*}
\begin{split}
      \hat{P}_{i+1}^{(\lambda,\beta)}(x) &= \left(x - \frac{\beta^2 -\lambda^2}{(2i + \lambda + \beta) (2i + \lambda + \beta+2)}\right) \hat{P}_i^{(\lambda,\beta)}  \\ &- \frac{4 i (i+\lambda)(i+\beta)(i+\lambda + \beta)}{(2i + \lambda + \beta)^2(2i + \lambda + \beta+1)(2i + \lambda + \beta-1)}\hat{P}_{i-1}^{(\lambda,\beta)}(x).
\end{split}
\end{equation*}
The angles $\omega_i$ and $\tau_i$ to implement $ \hat{P}_{i}^{(\lambda,\beta)}(U)$ can thus be derived from equations \eqref{eq:inv1} and \eqref{eq:inv2} using
\begin{equation}
    b_i = \frac{\beta^2 -\lambda^2}{(2i + \lambda + \beta) (2i + \lambda + \beta+2)}, \quad a_i^2 = \frac{4 i (i+\lambda)(i+\beta)(i+\lambda + \beta)}{(2i + \lambda + \beta)^2(2i + \lambda + \beta+1)(2i + \lambda + \beta-1)}.
\end{equation}
In the case $\lambda = \beta$, we can consider the shifted polynomials $ \hat{P}_n^{(\lambda,\lambda)}(x-1/2)$, which satisfy a three-term recurrence relation with 
\begin{equation*}
    b_i = 1/2 , \quad a_i^2 = \frac{ i(i+2 \lambda)}{4(i + \lambda +1/2)(i + \lambda-1/2)}.
\end{equation*}
Using equations~\eqref{eq:inv1} and~\eqref{eq:inv4}, one finds that a block encoding of $ P_n^{(\lambda,\lambda)}(U-1/2)/\alpha_n$ is achieved with
\begin{equation*}
    \tau_i = 0, \quad \omega_{i+1} = \tan^{-1}\left(\frac{(-1)^i}{2^{i/2}} \prod_{j =1}^i \frac{j (j+2\lambda)}{(j +  \lambda +1/2)(j +  \lambda -1/2)}\right).
\end{equation*}
Expressing the product appearing in the expression for $\omega_i$ in terms of Gamma functions and using the asymptotic expansion for ratios of Gamma functions, we obtain for large $i$
\begin{equation*}
    \tan \omega_{i+1} = \frac{(-1)^i}{2^{i/2}} \left(\frac{\Gamma(\lambda + 3/2)\Gamma(\lambda + 1/2)}{\Gamma(2\lambda + 1) } + O(1/i^2)\right).
\end{equation*}
Consequently,
\begin{equation*}
   \alpha(\tau_{i+1},\omega_{i+1}) = 2^{(i-1)/2} \left(\frac{\Gamma(2\lambda + 1) }{\Gamma(\lambda + 3/2)\Gamma(\lambda + 1/2)} + O(1/i^2)\right),
\end{equation*}
and for large $n$, the normalization grows as $\alpha_n = \prod_{i =0}^n |\alpha(\tau_i,\omega_i)| \sim 2^{n^{2}/4}$.

\end{example}

\subsection{Angle-finding problem for \opqsp }

The previous subsection provided examples of polynomial sequences for which the \opqsp  angles can be derived exactly. Beyond such cases, a central aspect of QSP is to determine a sequence of angles that implements a block encoding of the form \eqref{eq:bsoprl} and \eqref{eq:block encoding} for a target polynomial $P_n(U)$. A natural approach is to invert the matrices $T(\tau,\omega)$ and iteratively fix the angles so that their action reduces the degree of $P_n$. Such procedure is commonly referred to as the \textit{layer-stripping} method.  Here, we present an alternative derivation of the \opqsp  angles for polynomials with real roots, based on the construction of a linear functional $\mathcal{L}$ associated with the target polynomial.

This approach is based on Proposition \ref{prop:qsp-oprl}, which characterizes the recurrence relation satisfied by the polynomial sequences implemented by OP-QSP. It relies on the identification of the recurrence coefficients $a_i$ and $b_i$ associated with a target polynomial $P_n$, as stated in the following proposition.

\begin{proposition}\label{prop:afp1}
   Let $\hat{P}_n$ be a monic polynomial of degree $n$ with distinct real roots $\{x_{m,n}\}_{m=1}^n$. Suppose that the linear functional $\mathcal{L}$, whose moments are given by $c_i = \mathcal{L}[x^i] = \sum_{m=1}^{n} x_{m,n}^i$, 
is quasi-definite. Then $\hat{P}_n$ is the $n$th polynomial in the sequence of orthogonal polynomials defined by the recurrence relation
    \begin{equation}\label{eq:recinprop}
    \hat{P}_{i+1}(x) = (x-b_i)\hat{P}_i(x) - a_i^2 \hat{P}_{i-1}(x),
\end{equation}
where $\hat{P}_{-1} =0$, $\hat{P}_0 = 1$, 
\begin{equation}\label{eq:aibi}
    a_i^2 = \frac{h_{i}h_{i-2}}{h_{i-1}^2}, \quad b_i =  \left( \frac{h_{i,i-1}}{h_{i-1} } -  \frac{h_{i+1,i}}{h_{i} }\right),
\end{equation}
and where $h_{i,k}$ denotes the determinant:
\begin{equation*}
   h_{i,k} := (-1)^{i+k}\begin{vmatrix}
    c_0 & c_{1} & \dots& c_{k-1} &  c_{k+1}& \dots & c_{i}\\
    c_{1} & c_{2} & \dots& c_{k} &  c_{k+2}& \dots & c_{i+1}\\
    \vdots & \vdots & \ddots & \vdots & \vdots & \ddots & \vdots \\
    c_{i-1} & c_{i} & \dots& c_{i+k-2} &  c_{i-k}& \dots & c_{2i-1}
    \end{vmatrix}.
\end{equation*}
\end{proposition}

\begin{proof}
    We seek to identify a sequence of polynomials $\{\hat{P}_i\}_{i=0}^n$ such that $\deg \hat{P}_i = i$ for each $i$, and which are orthogonal with respect to a linear functional $\mathcal{L}$, with the degree-$n$ polynomial coinciding with the target polynomial $\hat{P}_n$. To this end, we introduce a linear functional defined so that, for any polynomial $q(x)$ of degree at most $2n-1$, it satisfies
\begin{equation*}
   \mathcal{L}[q(x)]  = \sum_{m = 1}^{n} q(x_{m,n}),
\end{equation*}
where $x_{m,n}$ is the $m$-th root of the degree-$n$ target polynomial $\hat{P}_n(x)$. Given this construction, it is clear that $\mathcal{L}[x^i \hat{P}_n(x)] = 0$ for $i = 0, 1, \dots, n-1$. The associated moments $c_i$ with $i \leq 2n -1$ are real and given by
\begin{equation*}
    c_i = \mathcal{L}[x^i]  = \sum_{m = 1}^{n} x_{m,n}^i.
\end{equation*}
In the case where the form is quasi-definite ($h_i \neq 0$), one can express the sequence of polynomials orthogonal with respect to $\mathcal{L}$ in terms of these moments. The orthogonality condition $\mathcal{L}[\hat{P}_i(x)\hat{P}_j(x)] =0 $ for $j<i$ is equivalent to the condition that $\mathcal{L}[x^j \hat{P}_i(x)] = 0$ holds for all $j < i$. The coefficients $d_k$ of $\hat{P}_i(x) = d_0 + d_1 x + \cdots + d_i x^i$ therefore satisfy a system of linear equations, namely
\begin{equation*}
   0 =  \mathcal{L}[x^j P_i(x)] = \sum_{k=0}^i  \mathcal{L}[x^{j+k} d_k] =  \sum_{k=0}^i c_{j+k}d_k, \quad j = 0,1,\dots i-1.
\end{equation*}
The quasi-definiteness condition ensures the uniqueness of the solution for the coefficients $d_k$, subject to the monic normalization $d_i = 1$. Moreover, the solution can be expressed in terms of minors of a matrix whose entries are the moments $c_k$. The polynomial solution to $\mathcal{L}[x^j \hat{P}_i(x)] = 0$ for $j < i$ admits the form

\begin{equation}\label{eq:poly}
    \hat{P}_i(x) = \frac{1}{\ h_{i-1}}\begin{vmatrix}
    c_0 & c_{1} & \dots & c_{i}\\
    c_{1} & c_{2} & \dots & c_{i+1}\\
    \vdots & \vdots & \ddots & \vdots\\
    c_{i-1} & c_{i} & \dots & c_{2i-1}\\
        1 & x & \dots & x^{i}
    \end{vmatrix} = \sum_{k = 0}^i \frac{ h_{i,k}}{h_{i-1}} x^k,
\end{equation}
where $h_{i,k}$ denotes the determinant:
\begin{equation*}
   h_{i,k} := (-1)^{i+k}\begin{vmatrix}
    c_0 & c_{1} & \dots& c_{k-1} &  c_{k+1}& \dots & c_{i}\\
    c_{1} & c_{2} & \dots& c_{k} &  c_{k+2}& \dots & c_{i+1}\\
    \vdots & \vdots & \ddots & \vdots & \vdots & \ddots & \vdots \\
    c_{i-1} & c_{i} & \dots& c_{i+k-2} &  c_{i-k}& \dots & c_{2i-1}
    \end{vmatrix}.
\end{equation*}
By construction, these polynomials satisfy the orthogonality relation $
\mathcal{L}[\hat{P}_i(x)\hat{P}_j(x)] \propto \delta_{ij}$, with $\hat{P}_n$ coinciding with the target polynomial. We have thus identified a sequence of monic orthogonal polynomials with respect to the linear functional $\mathcal{L}$, whose degree-$n$ element is the target polynomial. By Theorem \ref{theo:oprl}, these polynomials also satisfy a three-term recurrence relation of the form \eqref{eq:recinprop}. The coefficients $a_i$ and $b_i$ can be determined by substituting equation \eqref{eq:poly} into equation \eqref{eq:recinprop} and comparing the leading and constant terms on both sides, from which we obtain equation \eqref{eq:aibi}.
\end{proof}

This result enables the derivation of an explicit formula for the angles in terms of the moments $c_i$, which determine the coefficients $h_i$ and $h_{i,k}$, and consequently $a_i$ and $b_i$. The solution to the OP-QSP angle-finding problem is given in the following proposition.

\begin{proposition}\label{prop:afp1v2}
    Let $\hat{P}_n$ be a monic polynomial of degree $n$ with distinct real roots $\{x_{m,n}\}_{m=1}^n$. Suppose that the linear functional $\mathcal{L}$, whose moments are given by $ c_i = \mathcal{L}[x^i] = \sum_{m=1}^{n} x_{m,n}^i$, is quasi-definite. Then OP-QSP yields a $(\alpha_n, n+1, 0)$-block encoding of $\hat{P}_n$ with angle sequences $\{\tau_i\}_{i=1}^n$ and $\{\omega_i\}_{i=1}^n$ given by
\begin{equation*}
    \tau_{i+1} = \tan^{-1}\left(2 b_i - 1\right),
\end{equation*}
and the recurrence
\begin{equation*}
    \omega_{i+1} = \cot^{-1}\left(\frac{\sqrt{2b_i^2 - 2 b_i + 1}}{a_i^2}\left(\left(2 b_{i-1} - 1\right) \tan (\omega_{i} )+ {\left( b_{i-1} - 1\right)\cot (\omega_{i} )}\right)\right)
\end{equation*}
where $a_i$ and $b_i$ are the coefficients given in equation \eqref{eq:aibi}
\end{proposition}
\begin{proof}
    This follows from Proposition \ref{prop:qsp-oprl}, which states that the sequence of polynomials implemented by OP-QSP is precisely the one satisfying the recurrence relation \eqref{eq:recinprop}, and Proposition \ref{prop:afp1}, which identifies the recurrence coefficients associated with a given target polynomial $\hat{P}_n$. The relation between the angles and these recurrence coefficients was derived in equations \eqref{eq:inv1} and \eqref{eq:inv2} in the proof of Proposition \ref{prop:qsp-oprl}.
\end{proof}

\begin{remark}
   The choice of the linear functional $\mathcal{L}$ in terms of the roots of $P_n$ in Propositions \ref{prop:afp1} and \ref{prop:afp1v2} is an application of the Gauss–Christoffel quadrature formula, which provides a method to evaluate $\mathcal{L}$ associated with a sequence of OPRL in terms of their roots.  
For any polynomial $q(x)$ of degree at most $2n-1$, we have that
\begin{equation*}
   \mathcal{L}[q(x)]  = \sum_{m = 1}^{n} q(x_{m,n}) \, \lambda_{m,n},
\end{equation*}
where $x_{m,n}$ is the $m$-th root of the degree-$n$ polynomial $P_n(x)$ in the OPRL sequence, and $\lambda_{m,n} > 0$ are the \textit{Cotes numbers}, given in terms of the derivative $P_n'$ and $P_{n-1}$ by
\begin{equation*}
    \lambda_{m,n} =  \frac{a_1^2 a_2^2 \dots a_{n-1}^2}{P_n'(x_{m,n}) \, P_{n-1}(x_{m,n})}.
\end{equation*}
Since a particular target polynomial is not fixed until the final $n$-th step, the polynomial $P_{n-1}$ in the \opqsp  setting is arbitrary. Consequently, the linear functional $\mathcal{L}$ associated with a target polynomial $P_n$ is not unique, and the corresponding Cotes numbers can be chosen as any positive values. In our case, we simply set $\lambda_{m,n} = 1$.
\end{remark}

%Note that the circuit of Figure~\ref{fig:circuitOPRL} can generate more than just sequences of OPRL, as it also admits cases where $a_n^2 \le 0$. However, we omit the analysis of this situation here. It is also worth noting that this circuit is an example of a circuit associated to OPRL, but is in no way unique. 

    While the \opqsp  approach to quantum signal processing is probabilistic (requiring a measurement outcome of $0$ on the ancillary qubit at each step, or $n$ such measurements at the end) it has the advantage of eliminating the need to determine the companion polynomial $Q_n$ that appears in the standard QSP protocol. It also provides a block encoding scheme for sequences of well-known classical orthogonal polynomials. Next, we discuss the functions that can be block-encoded when this scheme is used as a select oracle in an LCU circuit.

\subsection{Expressibility of LCU with \opqsp }

The previous subsections characterized and provided several examples of the sequence of polynomials block-encoded by the \opqsp  circuit. We now describe the block encoding achievable when the circuit is used as the select oracle of an LCU construction, as discussed in Section~\ref{sec:1}.

\begin{proposition}\label{prop:3.3orpl}
Let $f$ be a function that admits a degree-$n$ polynomial $P$ approximating it within $\epsilon$ on the unit circle $\mathbb{T}$,
\begin{equation*}
    |f(z) - P(z)| \le \epsilon, \quad z \in \mathbb{T}.
\end{equation*}
Let $\{\hat{P}_i(z)\}_{i=0}^n$ be a family of monic orthogonal polynomials, and let $\alpha_i$ denote the associated block encoding constants defined in~\eqref{eq:beconst}. Let $v_i$ be the coefficients of $P$ in this basis, i.e.
\begin{equation*}
    P(z) = \sum_{i=0}^n \frac{v_i}{\alpha_i} \hat{P}_i(z), 
    \qquad 
    \|v\|_1 = \sum_{i=0}^n |v_i|.
\end{equation*}
Then, an $(\|v\|_1, 2n + 1, \epsilon)$-block encoding of $f(U)$ can be implemented with $O(n)$ depth and gate complexity, including $n$ applications of controlled-$U$.
\end{proposition}

\begin{proof}
    In the previous subsections, we showed that the sequence of unitaries obtained from repeated applications of the circuit in Figure~\ref{fig:circuitOPRL} admits the block structure described in equations~\eqref{eq:bsoprl} and~\eqref{eq:block encoding}. This matches the general form~\eqref{eq:bepoly} with $\ell = n +1$, where $\{\hat{P}_i\}_{i=0}^n$ is a sequence of monic orthogonal polynomials, as established in Proposition~\ref{prop:corroprl1}. The result then follows directly from Lemma~\ref{prop:sed2}, with a complexity corresponding to $n$ applications of the circuit in Figure~\ref{fig:circuitOPRL} used as the select oracle in the LCU construction. This entails $n$ applications of doubly-controlled $U$, each of which can be implemented using a controlled-$U$ together with a Toffoli gate.
\end{proof}

This proposition can be used to derive more precise results once a specific family of monic orthogonal polynomials is fixed. We now consider the case where this family is given by the Hermite polynomials, as discussed in Example~\ref{ex:herm}.

\begin{proposition}\label{prop:herms}
    Let $f$ be a function and let $\gamma \in \mathbb{R}$ be a fixed parameter such that $\gamma > 1$ and
\begin{equation}\label{eq:sqi}
   \mathcal{L}[f^2] =  \int_{\mathbb{R}} f^2(x)\, e^{-\frac{\gamma^2}{2}(x-1)^2}\, \gamma\, dx  < +\infty.
\end{equation}
Assume moreover that there exists a constant $M>0$ such that, for all $k \ge 0$,
\begin{equation}\label{eq:sqi2}
  \left| \mathcal{L}\!\left[\dv[k]{f}{x}\right] \right|
  =
  \left|\int_{\mathbb{R}} \dv[k]{f(x)}{x}\, e^{-\frac{\gamma^2}{2}(x-1)^2}\, \gamma\, dx \right|
  \le M \frac{\sqrt{k!}}{k^{1+ 3 \gamma^2}}.
\end{equation} Then, an $(\|v\|_1, 2n+1, \epsilon)$-block encoding of $f(U)$ can be implemented with $O(n)$ depth and gate complexity, which includes $n$ applications of controlled-$U$, where
\begin{equation*}
    n = O(\log(1/\epsilon)), \qquad \|v\|_1 = O\big(\log(1/\epsilon)^{1 + 3 \gamma^2}\big).
\end{equation*}
\end{proposition}

\begin{proof}
   The Hermite polynomials form a complete basis for the space of square-integrable functions with respect to a Gaussian weight. It follows from \eqref{eq:sqi} that there exist coefficients $v_i$ such that
\begin{equation*}
f(x) = \sum_{k=0}^\infty \frac{v_k}{\alpha_k} \hat{P}_k(x),
\end{equation*}
where $\hat{P}_k(x) = \gamma^{-k} H_k\big(\gamma (x-1)\big)$ are the polynomials introduced in Example~\ref{ex:herm}, and $\alpha_k$ denotes the associated block encoding constants. Using the orthogonality relation for the Hermite polynomials \eqref{eq:orher}, we obtain the following expression for the coefficients $v_k$,
\begin{equation*}
v_k = \mathcal{L}[f(x)\hat{P}_k(x)] \, \frac{\alpha_k \gamma^{2k}}{\sqrt{2\pi}\, k!}.
\end{equation*}
We therefore seek an upper bound for $\mathcal{L}[f(x)\hat{P}_k(x)]$. Using Rodrigues’ formula for the Hermite polynomials, we have
\begin{equation*}
\hat{P}_k(x)\, e^{-\gamma^2(x-1)^2/2}
= \frac{1}{\gamma^{2k}} \frac{d^k}{dx^k} \left(e^{-\gamma^2(x-1)^2/2}\right).
\end{equation*}
Hence, using this relation, integration by parts and equation \eqref{eq:sqi2}, we obtain
\begin{equation*}
\begin{aligned}
\left|\mathcal{L}[f(x)\hat{P}_k(x)]\right|
&= \left|\int_{\mathbb{R}} f(x)\hat{P}_k(x) e^{-\gamma^2(x-1)^2/2}\, \gamma\, dx \right| \\
&= \left|\int_{\mathbb{R}} \gamma^{-2k} f(x) \frac{d^k}{dx^k}\!\left(e^{-\gamma^2(x-1)^2/2}\right) \gamma\, dx \right| \\
&= \left|\int_{\mathbb{R}} \gamma^{-2k} \frac{d^k f(x)}{dx^k} e^{-\gamma^2(x-1)^2/2} \gamma\, dx \right| \\
&  \leq M \frac{\sqrt{k!}  }{k^{1 + 3 \gamma^2}}\gamma^{-2k}.
\end{aligned}
\end{equation*}
Using this inequality together with the asymptotic behavior of $\alpha_k$ for Hermite polynomials given in \eqref{eq:asyher}, we obtain that there exists a constant $C$ such that
\begin{equation}\label{eq:bdvkher}
       |v_k| \leq   \frac{C}{\gamma^k}.
    \end{equation}
Next, since $\left|\hat{P}_k(z)/\alpha_k\right| \leq 1$ for all $z \in \mathbb{T}$, it follows that
\begin{equation*}
\max_{z \in \mathbb{T}} \left|f(z) - \sum_{k=0}^n v_k \frac{\hat{P}_k(z)}{\alpha_k}\right|
= \max_{z \in \mathbb{T}} \left|\sum_{k=n+1}^\infty v_k \frac{\hat{P}_k(z)}{\alpha_k}\right|
\leq \sum_{k=n+1}^\infty |v_k|.
\end{equation*}
In particular, applying the bound \eqref{eq:bdvkher} to the right-hand side, we deduce that there exists a constant $C'$ such that
\begin{equation*}
\max_{z \in \mathbb{T}} \left|f(z) - \sum_{k=0}^n v_k \frac{\hat{P}_k(z)}{\alpha_k}\right|
\leq  \frac{C'}{\gamma^n}.
\end{equation*}
Therefore, any function $f$ satisfying the conditions of the proposition can be approximated on the unit circle within error $\epsilon$ using $O(\log(1/\epsilon))$ terms of its Hermite expansion. To bound $\|v\|_1 = |v_0| + \dots + |v_n|$, we note that the Cauchy-Schwartz inequality implies
\begin{equation*}
    \|v\|_1 = \sum_{i =0}^n |v_1| \leq \sqrt{\sum_{k=0}^n v_k^2 \frac{\sqrt{2\pi}k!}{\alpha_k^2 \gamma^{2k}} }\sqrt{\sum_{k=0}^n \frac{\alpha_k^2 \gamma^{2k}}{\sqrt{2\pi}k!} }
\end{equation*}
Using the orthogonality of the Hermite polynomials together with the expansion of $f$ in this basis, we can bound the first term in terms of $\mathcal{L}[f^2]$:
\begin{equation*}
    \sum_{k=0}^n v_k^2 \frac{\sqrt{2\pi}\,k!}{\alpha_k^2 \gamma^{2k}}
    \leq \sum_{k=0}^\infty v_k^2 \frac{\sqrt{2\pi}\,k!}{\alpha_k^2 \gamma^{2k}}
    = \mathcal{L}[f^2].
\end{equation*}
Using the asymptotic behavior of $\alpha_k$ in equation \eqref{eq:asyher}, we find that there exists a constant $D$ independent of $f$ and $n$ such that
\begin{equation*}
    \sum_{k=0}^n \frac{\alpha_k^2 \gamma^{2k}}{\sqrt{2\pi}k!} \leq  \sum_{k =0}^n k^{1 + 6 \gamma^2 + O(1/k)} \leq D\, n^{2 + 6 \gamma^2}.
\end{equation*}
Combining these results, we obtain
\begin{equation*}
    \|v\|_1 \leq D^{1/2} \, \sqrt{\mathcal{L}[f^2]} \, n^{1 + 3\gamma^2}.
\end{equation*}
The desired result then follows from this bound together with an application of Proposition~\ref{prop:3.3orpl}, using the truncated Hermite series of $f$ with $n$ terms.
\end{proof}

\section{Generalized QSP and Laurent biorthogonal polynomials}\label{sec:3}

In this section, we consider the sequences of polynomials generated by generalized QSP, as introduced in \cite{motlagh2024generalized}. In contrast to the sequences produced by OP-QSP, these sequences do not form an orthogonal polynomial basis. Instead, we show that they form a biorthogonal basis and can be identified as Laurent biorthogonal polynomials (LBPs). We apply this result to the associated angle-finding problem.

\subsection{Circuit and recurrence relation for generalized QSP}

Let $U$ denote an arbitrary unitary acting on an Hilbert space $\mathcal{H}$ and assume access to its $0$-controlled version acting $\mathbb{C}^2 \otimes \mathcal{H}$. Generalized QSP achieves a block-encoding of a polynomial in $U$ by repeated application of the circuit illustrated in Figure \ref{fig:circuitGQSP}, which consists of a $0$-controlled application of $U$ followed by a rotation $R(\theta,\phi)$ on an ancillary qubit. The gate $R(\theta,\phi)$ corresponds to the following $U(2)$ rotation of the ancillary qubit:
\begin{equation*}
    R(\theta,\phi) = \begin{pmatrix}
        e^{i \phi} \cos \theta &  e^{i \phi} \sin \theta  \\
        \sin \theta & -\cos \theta
    \end{pmatrix} \otimes I.
\end{equation*}
This is sometimes referred to in the literature as an ${\rm SU}(2)$ rotation, as it differs from one only by a global phase $e^{i\phi/2}$, which does not affect the outcome of the protocol. The unitary transformation induced by one iteration of the circuit is given by,
\begin{equation}\label{eq:transfgqsp}
    T(\theta,\phi, U) = \begin{pmatrix}
        e^{i \phi} \cos \theta\, U &  e^{i \phi} \sin \theta  \\
        \sin \theta \, U & -\cos \theta
    \end{pmatrix}.
\end{equation}
Repeated $i$ times with different angles $\theta$ and $\phi$, the resulting transformation is given by
\begin{equation}\label{eq:begqsp}
   \begin{pmatrix}
       P_i(U) & \cdot \quad \\
       Q_i(U) & \cdot \quad 
   \end{pmatrix} 
   = T(\theta_i,\phi_i,U)\dots T(\theta_1,\phi
   _1, U),
\end{equation}
where $P_i(U)$ and $Q_i(U)$ are polynomials of degree $i$ in $U$. From the structure of the matrix $ T(\theta,\phi, U)$, we derive the following recurrence relations for the sequences of polynomials $P_i$ and $Q_i$, viewed as polynomials in a complex variable $z \in \mathbb{C}$:
 \begin{equation}\label{eq:rec1}
    P_{i}(z) = e^{i \phi_i} \cos \theta_i z P_{i-1}(z) + e^{i \phi_i} \sin \theta_i  Q_{i-1}(z),
\end{equation}
\begin{equation}\label{eq:rec2}
    Q_{i}(z) =\sin \theta_i z P_{i-1}(z) -\cos \theta_i  Q_{i-1}(z).
\end{equation}

From \eqref{eq:begqsp}, it follows that generalized QSP provides an instance of the class of algorithms discussed in Section~\ref{sec:1}, yielding a block encoding of the form \eqref{eq:bepoly} with $\ell = 1$, $W_j = T(\theta_j, \phi_j, U)$, and $P_i(U) = \hat{P}_i(U)/\alpha_i$.  Lemma~\eqref{prop:sed2} therefore applies, so that generalized QSP can be used as the select oracle in an LCU circuit to implement a block encoding of linear combinations of the polynomial basis $\{P_i(U)\}_{i=0}^n$. We now turn to the characterization of this basis.

\subsection{Relation between generalized QSP and Laurent biorthogonal polynomials}

The following proposition provides a Favard-type theorem characterizing the biorthogonality of the finite sequences of polynomials obtained in generalized QSP, that is, those defined by the recurrence relations \eqref{eq:rec1} and \eqref{eq:rec2}. Its proof is given in Appendix~\ref{sec:appA} and follows from inverting the recurrence relation and applying the residue theorem.

\begin{proposition}\label{prop:form}
    Let $\{P_i(z)\}_{i=0}^n$ and $\{Q_i(z)\}_{i=0}^n$ be finite sequences of polynomials satisfying the recurrence relations  \eqref{eq:rec1} and \eqref{eq:rec2}, with $\text{deg}\, P_i = \text{deg}\, Q_i = i$ and $\theta_i \neq 0 \mod \pi $. Then there exists a linear functional $\mathcal{L}$ and constants $\chi_i$ such that we have the following biorthogonal relations:
    \begin{equation}\label{eq:condform}
        \mathcal{L}[\tilde{Q}_i(z)P_j(z) ] = \chi_i \delta_{ij}
    \end{equation}
    where $\tilde{Q}_i(z) = z^{-i} Q_i(z)$. Furthermore, the linear functional $\mathcal{L}$ has a contour integral representation,
    \begin{equation}\label{eq:linfun}
         \mathcal{L}[f(z)] = \int_\Gamma \frac{f(z)}{z P_n(z)\tilde{Q}_n(z)} dz,
    \end{equation}
    where $\Gamma$ denotes a closed contour whose interior contains $z=0$ and all roots of $P_n(z)$, but none of the roots of $Q_n(z)$.
\end{proposition}

\begin{corollary}
    The sequences of polynomials $\{P_i(z)\}_{i=0}^n$ and $\{Q_i(z)\}_{i=0}^n$ implemented at each step of the generalized quantum signal processing algorithm form a biorthogonal basis; that is, they satisfy equation \eqref{eq:condform} for some constants $\chi_i$.
\end{corollary}

Sequences of polynomials $P_i$ and Laurent polynomials $\tilde{Q}_i$ that form a biorthogonal basis are referred to as Laurent biorthogonal polynomials \cite{hendriksen1991biorthogonal, zhedanov1998classical}. To complete the characterization of the polynomial sequences implementable by generalized QSP, we require the converse result. We therefore derive a necessary and sufficient condition for the converse, that is, for a sequence of LBPs to be implementable via generalized QSP. This relies on results that hold for linear functionals $\mathcal{L}$ with specific properties.

\begin{definition}
    A linear functional $\mathcal{L}$ is said to be \textit{biorthogonal quasi-definite of order $n$} if $h_{i} \neq 0$ for $i = 0, 1,\dots n$ where
\begin{equation*}
    h_{i} = \begin{vmatrix}
    c_0 & c_{1} & \dots & c_{i-1}\\
    c_{-1} & c_{0} & \dots & c_{1}\\
    \vdots & \vdots & \ddots & \vdots\\
    c_{-i+1} & c_{-i+2} & \dots & c_{0}
    \end{vmatrix}.
\end{equation*}
\end{definition}
In particular, we have the following lemma on the existence of Laurent biorthogonal polynomials with respect to biorthogonal quasi-definite linear functionals.
\begin{lemma}
   Let $\mathcal{L}$ be a linear functional that is biorthogonal quasi-definite of order $n$. 
Then there exist sequences of polynomials $\{P_i(z)\}_{i=0}^n$ and $\{Q_i(z)\}_{i=0}^n$, unique up to normalization, 
with $\deg P_i = \deg Q_i = i$, such that
\begin{equation}\label{eq:linfun2}
\mathcal{L}\big[P_i(z)\,\tilde{Q}_j(z)\big] = \chi_i \,\delta_{ij}, \qquad \chi_i \neq 0,
\end{equation}
for all $0 \le i,j \le n$, with $\tilde{Q}_i(z) = z^{-i} Q_i(z)$
\end{lemma}
\begin{proof}
    Assuming $\deg P_i = i$ and $\deg Q_i = i$, a linear independence argument shows that the biorthogonality condition in equation~\eqref{eq:linfun2} is equivalent to the following set of relations:
\begin{equation}\label{eq:co1}
        \mathcal{L}[z^{-j}P_i(z)]  = \eta_i \delta_{ij}, \quad 
        \mathcal{L}[z^{-j}Q_i(z)] = \upsilon_i \delta_{j0},
    \end{equation}
for some constant non-zero parameters $\eta_i$ and $\upsilon_i$. These equations can be reformulated as a linear system for the coefficients of $P_i$ and $Q_i$.
The biorthogonal quasi-definite condition then guarantee that this system admits a unique solution, up to a choice of normalization, which establishes the result.
\end{proof}

The preceding definition and lemma provide conditions for the existence of Laurent biorthogonal polynomials associated with a given linear functional.
The following proposition describes the recurrence relations satisfied by these sequences of polynomials.

\begin{proposition}\label{prop:rec_from_biortho}
   Let $\mathcal{L}$ be a linear functional that is biorthogonal quasi-definite of order $n$. Let $\{P_i(z)\}_{i=0}^n$ and $\{Q_i(z)\}_{i=0}^n$ be the finite sequences of polynomials satisfying $\deg P_i = \deg Q_i = i$ and
\begin{equation}\label{eq:biortho_condi}
        \mathcal{L}[\tilde{Q}_i(z)P_j(z) ] = \chi_i \delta_{ij}
    \end{equation}
    where $\tilde{Q}_i(z) = z^{-i} Q_i(z)$. Then, there exists coefficients $\zeta_i$, $\beta_i$, $\gamma_i$ and $\delta_i$ such that 
\begin{equation}\label{eq:rec1alt}
    P_{i}(z) = \zeta_i z P_{i-1}(z) + \beta_i Q_{i-1}(z), \quad
    Q_{i}(z) = \gamma_i z P_{i-1}(z) + \delta_i Q_{i-1}(z).
\end{equation}
\end{proposition}
\noindent The proof of this proposition is given in Appendix~\ref{sec:appA}. It follows from noting that one can choose $\zeta_i$ and $\gamma_i$ such that $P_{i}(z) - \zeta_i z P_{i-1}(z)$ and $Q_{i}(z) - \gamma_i z P_{i-1}(z)$ are polynomials of degree $i-1$. We then deduce that these polynomials are proportional to $Q_{i-1}$, since they satisfy the same biorthogonality relations with respect to the monomials $\{z^{-j}\, |\, j = 0,1,\dots,i-2 \}$. 

The recurrence relations in the previous proposition are similar to that satisfied by sequences associated with generalized QSP, but they are in fact more general, since there are no constraints requiring the coefficients $\zeta_i$, $\beta_i$, $\gamma_i$, and $\delta_i$ to be trigonometric functions as in equations \eqref{eq:rec1} and \eqref{eq:rec2}. This is expected, as polynomial sequences satisfying equation \eqref{eq:biortho_condi} are LBPs and form a broader class than those generated by generalized QSP. By comparing \eqref{eq:rec1} and \eqref{eq:rec2}, which describe general LBPs, with the recurrence relation \eqref{eq:rec1alt} associated with generalized QSP, we observe that the LBPs arising in the latter correspond to the case in which there exist real angles $\theta_i$ and $\phi_i$ such that
\begin{equation}\label{eq:contr1}
   \frac{\gamma_{i} \zeta_{i-1} }{\delta_i \gamma_{i-1}} = -e^{i \phi_{i-1}}\frac{\tan \theta_i}{\tan \theta_{i-1}}, \qquad \frac{\beta_i\gamma_i}{\zeta_i\delta_i}  = -\tan^2 \theta_i.
\end{equation}
This provides a method to identify which families of Laurent biorthogonal polynomials (LBPs) can be realized using generalized QSP. In practice, however, this approach may not be convenient, as it is expressed in terms of the recurrence relation \eqref{eq:rec1alt}, which does not correspond to the standard form commonly used for LBPs.
In the literature, these polynomials are typically studied in terms of the three-term recurrence relations they satisfy.

\begin{proposition} (reformulation of Proposition 1.1 in \cite{zhedanov1998classical})
  Let $\mathcal{L}$ be a linear functional that is biorthogonal quasi-definite of order $n$. Let $\{\hat{P}_i(z)\}_{i=0}^n$ and $\{Q_i(z)\}_{i=0}^n$ be the finite sequences of polynomials with $\hat{P}_i$ monic and $\deg \hat{P}_i = \deg Q_i = i$, satisfying
\begin{equation}\label{eq:biortho_condi}
        \mathcal{L}[\tilde{Q}_i(z)P_j(z) ] = \chi_i \delta_{ij}
    \end{equation}
    where $\tilde{Q}_i(z) = z^{-i} Q_i(z)$. Then there exist coefficients $d_i$ and $b_i$ such that the polynomials $\hat{P}_i$ satisfy the three-term recurrence relation
\begin{equation}\label{eq:reclaurent}
    \hat{P}_{i+1}(z) = (z-d_i)\hat{P}_i(z) -  z b_i \hat{P}_{i-1}(z).
\end{equation}
\end{proposition}
\begin{proof}
   This result, derived in \cite{zhedanov1998classical}, can also be obtained directly from equation~\eqref{eq:rec1alt}. Upon rescaling the polynomials $P_i(x)$ to their monic form
\begin{equation}
    \hat{P}_i(x) = \left(\prod_{j = 1}^{i}\zeta_i^{-1}\right) P_i(x), \quad \hat{Q}_i(x) = \gamma_{i}^{-1}\left(\prod_{j = 1}^{i-1}\zeta_i^{-1}\right) Q_i(x)
\end{equation}
and taking an appropriate linear combination of the two equations in \eqref{eq:rec1alt}, one obtains
\begin{equation*}
    \hat{P}_{i}(z) = z \hat{P}_{i-1}(z) + \frac{\beta_i \gamma_{i-1}}{\zeta_i \zeta_{i-1} }\hat{Q}_{i-1}(z),  \quad \frac{\delta_i}{\gamma_i} \hat{P}_{i}(z) = z \left(\frac{\delta_i}{\gamma_i} - \frac{\beta_i}{\zeta_i}\right) \hat{P}_{i-1}(z) + \frac{\beta_i}{\zeta_i} \hat{Q}_{i}(z).
\end{equation*}
By substituting one of the relations into the other, we then obtain the standard (monic) three-term recurrence relation \eqref{eq:reclaurent} for LBPs, with
\begin{equation}\label{eq:matchcoe}
    d_i = \frac{\beta_{i+1} \delta_{i}}{\zeta_{i+1} \beta_{i}}, \quad b_i = \frac{\beta_{i+1}}{\zeta_{i+1}}\left( \frac{\delta_{i}}{\beta_{i}}-\frac{\gamma_{i}}{\zeta_{i}}\right).
\end{equation}
\end{proof}

Rewriting the constraint \eqref{eq:contr1} in terms of the parameters $d_i$ and $b_i$ yields the following proposition, which characterizes the sequences of Laurent biorthogonal polynomials (LBPs) that can be implemented by generalized QSP in terms of their recurrence coefficients.

\begin{proposition}\label{prop:idlbps}
   Let $\{\hat{P}_i(z)\}_{i=0}^n$ be a finite sequence of monic Laurent biorthogonal polynomials with respect to a linear functional that is  biorthogonal quasi-definite of order $n$, satisfying $\deg \hat{P}_i = i$. Let $b_i$ and $d_i$ be the coefficients in the associated recurrence relation,
\begin{equation*}
    \hat{P}_{i+1}(z) = (z-d_i)\hat{P}_i(z) -  z b_i \hat{P}_{i-1}(z).
\end{equation*}   
Then there exist constants $\alpha_i$ such that the sequence $\{ {\hat{P}_i(U)}/{\alpha_i} \}_{i=0}^n$
can be realized as the block encoding produced at each iteration of a generalized QSP protocol if and only if there exist real angles $\theta_i$ and $\phi_i$ such that \begin{equation}\label{eq:co_laurent}
        d_i - b_i = \frac{\tan \theta_{i+1} \tan \theta_{i}}{e^{i \phi_{i}}}, \quad d_i = \frac{- \tan \theta_{i}}{e^{i \phi_{i}}\tan \theta_{i+1}}.
    \end{equation}
for all $i=0,\dots,n$. Furthermore, the constants $\alpha_j$ are given by
\begin{equation}\label{eq:beclbpsf}
    \alpha_j = \prod_{k=1}^{j}\frac{ e^{-i\phi_k} }{\cos \theta_k}.
\end{equation}
\end{proposition}
\begin{proof}
    Proposition \ref{prop:rec_from_biortho} shows that the LBPs satisfy the recurrence relations \eqref{eq:rec1alt}, which can be seen to coincide exactly with \eqref{eq:rec1} and \eqref{eq:rec2} for the polynomial sequences in generalized QSP, provided that equation \eqref{eq:contr1} holds. Using \eqref{eq:matchcoe}, which expresses the recurrence coefficients $b_i$ and $d_i$ in terms of those in \eqref{eq:rec1alt}, the condition in \eqref{eq:contr1} can then be rewritten directly in terms of $b_i$ and $d_i$. The result then follows by combining this with the scaling of the leading coefficient of $P_i = \hat{P}_i/\alpha_i$ in terms of the angles, which can be derived from \eqref{eq:rec1}.
\end{proof}

The sequences of Laurent biorthogonal polynomials that can be implemented by generalized QSP can equivalently be characterized by their final polynomial. The following proposition provides the necessary and sufficient conditions that the final polynomials $P_n$ and $Q_n$ must satisfy. Its proof is given in Appendix \ref{sec:appA}. It corresponds to Condition 2 in Theorem 3 of \cite{motlagh2024generalized}.

\begin{proposition}\label{prop:cond_rec_from_biortho}
  Let $\mathcal{L}$ be a linear functional that is biorthogonal quasi-definite of order $n$. Let $\{\hat{P}_i(z)\}_{i=0}^n$ and $\{\hat{Q}_i(z)\}_{i=0}^n$ be finite sequences of monic polynomials satisfying $\deg \hat{P}_i = \deg \hat{Q}_i = i$ and
    \begin{equation}\label{eq:biortho_condi2}
        \mathcal{L}[\tilde{Q}_i(z)\hat{P}_j(z) ]  = \chi_i \delta_{ij}
    \end{equation}
    where $\tilde{Q}_i(z) = z^{-i} \hat{Q}_i(z)$. Then there exist constants $\alpha_i$ and $\tilde{\alpha}_i$ such that $\{\hat{P}_i(U)/\alpha_i\}_{i=0}^n$ and $\{\hat{Q}_i(U)/\tilde{\alpha}_i \}_{i=0}^n$ can be realized as the block encoding \eqref{eq:begqsp} produced at each iteration of the generalized QSP protocol if and only there exists $\alpha_n$ and $\tilde{\alpha}_n$ such that
\begin{equation}\label{eq:unitok}
      |z| = 1 \qquad \Rightarrow \qquad  \left|\frac{\hat{P}_n(z)}{\alpha_n}\right|^2 + \left|\frac{\hat{Q}_n(z)}{\tilde{\alpha}_n}\right|^2 = 1.
    \end{equation}
\end{proposition}
\noindent This concludes the characterization of the polynomial sequences generated by generalized QSP, which are in correspondence with LBPs satisfying the additional condition~\eqref{eq:unitok}. We now provide an example based on a known family of LBPs.

\begin{example}
    Let us consider the fixed angle case, where $\theta_j = \theta$ and $\phi_j = \phi$.  We can inject these angles in equation \eqref{eq:co_laurent} and find the following constant recurrence coefficients:
\begin{equation}
    d_j = -e^{-i \phi}, \quad b_j =  -\frac{e^{-i \phi}}{\cos^2 \theta} .
\end{equation}
The LBPs associated to these recurrence coefficients were studied in \cite{barry2019constant}, where an explicit expression for the polynomial was derived:
\begin{equation}
    \hat{P}_n(z) = \sum_{k = 0}^n a_{nk} z^k, \qquad a_{nk} =   \sum_{j = 0}^k \binom{k}{j} \binom{n-j}{n-k-j} \frac{e^{-i(n-k) \phi}}{\cos^{2j} \theta}.
\end{equation}
\end{example}

\begin{remark}
    If $P_n(z)$ has $n$ distinct non-zero roots $z_{n,m}$, we note that the linear functional $\mathcal{L}$ from the Proposition \ref{prop:form} can then be represented, by applying the residue theorem, as
\begin{equation*}
        \mathcal{L}[f(z)] = \sum_{m = 1}^n \frac{ 2\pi i f(z_{n,m})}{z_{n,m} P_n'(z_{n,m})\tilde{Q}_n(z_{n,m})},
\end{equation*}
for any function $f$ such that $f(0)$ and $f(z_{n,m})$ are well defined.
\end{remark}

\begin{remark}
   One may replace the $U(2)$ elements $R(\theta_i, \phi_i)$ in the generalized QSP protocol by arbitrary $2 \times 2$ matrices through a block encoding approach involving additional ancillary qubits. In this case, the associated polynomial recurrence relation takes the general form of equation~\eqref{eq:rec1alt}, thereby lifting the constraint~\eqref{eq:unitok} on the companion polynomial $Q_n$ for a given target polynomial $P_n$. We omit the analysis of this approach, as requiring a block encoding at each iteration increases greatly the complexity and renders each step probabilistic rather than deterministic.
\end{remark}

We now leverage the connection between generalized QSP and LBPs to derive a closed-form expression for the angles that generate a given target polynomial, expressed in terms of the moments of the linear functional $\mathcal{L}$.

\subsection{Angle-finding problem for generalized QSP}

Proposition \ref{prop:form} states that the sequence of polynomials implemented by generalized QSP is biorthogonal with respect to a linear functional that is explicitly given in \eqref{eq:linfun} in terms of the QSP target polynomial $P_n$ and its partner $Q_n$. This correspondence can be used to determine the QSP angles associated with $P_n$ and $Q_n$ in terms of the moments of the associated linear functional $\mathcal{L}$, as stated in the following proposition.

\begin{proposition}
   Let $P_n$ and $Q_n$ be a pair of polynomials of degree $n$ satisfying $|P_n(z)|^2 + |Q_n(z)|^2 = 1$ for all $z \in \mathbb{T}$. Let $c_i = \mathcal{L}[z^i]$ denote the $i$th moment of the linear functional $\mathcal{L}$ defined by
    \begin{equation}\label{eq:linfuafp}
         \mathcal{L}[f(z)] = \int_\Gamma \frac{z^{n-1} f(z)}{P_n(z)Q_n(z)} dz,
    \end{equation}
   where $\Gamma$ denotes a closed contour whose interior contains $z=0$ and all roots of $P_n(z)$, but none of the roots of $Q_n(z)$. If $\mathcal{L}$ is biorthogonal quasi-definite of order $n$, then a $(1,1,0)$-block encoding of $P_n(U)$ is obtained via generalized QSP with angle sequences $\{\theta_i\}_{i=1}^n$ and $\{\phi_i\}_{i=1}^n$ determined by
    \begin{equation}\label{eq:parmomlbp}
    \tan^2 \theta_i = - \frac{h_{i}^+ h_{i}^-}{h_{i}^2}, \qquad e^{2 i \phi_{i}} = \frac{h_{i}^+h_{i+1}^-}{h_{i}^-h_{i+1}^+},
\end{equation}
where the coefficients $h_i$ and  $h_{i}^\pm$ are defined by
\begin{equation}\label{eq:detafpgqsp}
    h_{i} = \begin{vmatrix}
    c_0 & c_{1} & \dots & c_{i-1}\\
    c_{-1} & c_{0} & \dots & c_{1}\\
    \vdots & \vdots & \ddots & \vdots\\
    c_{-i+1} & c_{-i+2} & \dots & c_{0}
    \end{vmatrix}, \quad  h_{i}^\pm = \begin{vmatrix}
    c_{\pm 1} & c_{\pm 2} & \dots & c_{\pm i}\\
    c_{0} & c_{\pm 1} & \dots & c_{\pm i \mp 1}\\
    \vdots & \vdots & \ddots & \vdots\\
    c_{\pm i \mp 2} & c_{\pm i \mp 1} & \dots & c_{\pm 1}
    \end{vmatrix}.
\end{equation}
\end{proposition}

\begin{proof}
    We have a target polynomial $P_n(z) = \hat{P}_n(z)/\alpha_n$  to implement with generalized QSP, and its companion polynomial $Q_n(z)= \hat{Q}_n(z)/\tilde{\alpha}_n$ such that condition \eqref{eq:unitok} holds, i.e.
\begin{equation*}
    z \in \mathbb{T} \qquad \Rightarrow \qquad |P_n(z)|^2 + |Q_n(z)|^2 = \left|\frac{\hat{P}_n(z)}{\alpha_n}\right|^2 + \left|\frac{\hat{Q}_n(z)}{\tilde{\alpha}_n}\right|^2 = 1.
\end{equation*}
To determine the associated QSP angles, we can first construct sequences of Laurent biorthogonal polynomials whose final elements are proportional to $\hat{P}_n$ and $\hat{Q}_n$. Let $\{\hat{P}_i(z)\}_{i=0}^{n}$ and $\{\check{Q}_i(z)\}_{i=0}^{n}$ denote these sequences, normalized so that each $\hat{P}_i$ is monic and each $\check{Q}_i$ has constant term equal to $1$. These sequences must satisfy a biorthogonality relations with respect to a linear functional $\mathcal{L}$,
\begin{equation*}
    \mathcal{L}\!\left[\tilde{Q}_i(z)\hat{P}_j(z)\right] = \chi_i \delta_{ij},
\end{equation*}
where $\tilde{Q}_i(z) = z^{-i}\check{Q}_i(z)$. Expanding these polynomials in the monomial basis, this condition is equivalent to the following set of relations,
    \begin{equation}\label{eq:co1v}
        \mathcal{L}[z^{-j}\hat{P}_i(z)] = \chi_i \delta_{ij}, \quad 
        \mathcal{L}[z^{-j}\check{Q}_i(z)] = \chi_i \delta_{j0},
    \end{equation}
which hold for all $0 \le j \le i$. It is straightforward to verify that the case $i = n$ of \eqref{eq:co1v} holds, using the target polynomials $\hat{P}_n$ and $\hat{Q}_n$ together with the linear functional defined in \eqref{eq:linfuafp}. For $i \neq n$, these relations give rise to linear systems for the coefficients of the monomial expansions of the polynomials $\hat{P}_i$ and $\check{Q}_i$. A solution can be expressed in determinantal form as
\begin{equation}\label{eq:sol}
     \hat{P}_i(z) = \frac{1}{h_i}\begin{vmatrix}
    c_0 & c_{1} & \dots & c_{i}\\
    c_{-1} & c_{0} & \dots & c_{i-1}\\
    \vdots & \vdots & \ddots & \vdots\\
    c_{-i+1} & c_{-i+2} & \dots & c_{1}\\
        1 & z & \dots & z^{i}
    \end{vmatrix}, \qquad  \check{Q}_i(z) = \frac{1}{h_i}\begin{vmatrix}
    c_0 & c_{-1} & \dots & c_{-i}\\
    c_1 & c_{0} & \dots & c_{-i+1}\\
    \vdots & \vdots & \ddots & \vdots\\
    c_{i-1} & c_{i-2} & \dots & c_{-1}\\
    z^i & z^{i-1} & \dots & 1
    \end{vmatrix},
\end{equation}
with $\chi_i = h_{i+1}/h_i$. It is straightforward to verify that this is indeed a solution by substituting these expressions into \eqref{eq:co1v}, which, for instance, yields for $0 \leq j \leq i$
\small{
    \begin{equation*}
        \mathcal{L}[z^{-j}\hat{P}_i(z)] = \frac{1}{h_i}\begin{vmatrix}
    c_0 & c_{1} & \dots & c_{i}\\
    c_{-1} & c_{0} & \dots & c_{i-1}\\
    \vdots & \vdots & \ddots & \vdots\\
    c_{-i+1} & c_{-i+2} & \dots & c_{1}\\
        \mathcal{L}[z^{-j}] & \mathcal{L}[z^{1-j}] & \dots & \mathcal{L}[z^{i-j}]
    \end{vmatrix} = \frac{1}{h_i}\begin{vmatrix}
    c_0 & c_{1} & \dots & c_{i}\\
    c_{-1} & c_{0} & \dots & c_{i-1}\\
    \vdots & \vdots & \ddots & \vdots\\
    c_{-i+1} & c_{-i+2} & \dots & c_{1}\\
        c_{-j} & c_{1-j} & \dots & c_{i-j}
    \end{vmatrix} =\frac{h_{i+1}}{h_i}  \delta_{ij}.
    \end{equation*}
    }
    \normalsize
    From the determinant associated with the linear systems for the coefficients in the monomial expansions of the polynomials $\hat{P}_i$ and $\check{Q}_i$, we conclude that the solution \eqref{eq:sol} exists and is unique if and only if $h_i \neq 0$ for $i = 0,1,\dots,n$. In other words, this holds if and only if the functional $\mathcal{L}$ defined in equation \eqref{eq:linfuafp} is biorthogonal quasi-definite of order $n$. In this case, the uniqueness condition further ensures that the $i = n$ solution in \eqref{eq:sol} is proportional to the target polynomials $\hat{P}_n$ and $\hat{Q}_n$, as they satisfy \eqref{eq:co1v}.

From Proposition \ref{prop:rec_from_biortho}, the biorthogonality of the polynomials $\hat{P}_i$ and $\check{Q}_i$ implies that they satisfy a recurrence relations of the form \eqref{eq:rec1alt}. The associated coefficients $\zeta_i$, $\beta_i$, $\gamma_i$ and $\delta_i$ can be determined by substituting the explicit expressions \eqref{eq:sol} into these relations and matching the leading and constant terms of the polynomials on both sides. We obtain
\begin{equation*}
    \zeta_i = 1, \quad \beta_i = (-1)^i\frac{h_i^+}{h_i}, \quad \gamma_i = (-1)^i\frac{h_i^-}{h_i}, \quad \delta_i = 1,
\end{equation*}
where the coefficients $h_{i}^\pm$ are defined as in equation \eqref{eq:detafpgqsp}.  Since the target polynomials satisfy the condition in Proposition \eqref{prop:cond_rec_from_biortho}, these recurrence relations must coincide with those in \eqref{eq:rec1} and \eqref{eq:rec2} for generalized QSP when expressed in terms of rescaled polynomials $P_i(z) = \hat{P}_i(z)/\alpha_i$ and $Q_i(z) = \check{Q}_i(z)/\check{\alpha}_i$. In particular, equation \eqref{eq:contr1} holds and can be used to determine the angles $\theta_i$ and $\phi_i$. This yields the explicit expressions \eqref{eq:parmomlbp} for the angles $\theta_i$ and $\phi_i$ in terms of the determinants $h_i$ and $h_i^{\pm}$ (and hence in terms of the moments $c_i$). 
\end{proof}

\begin{example}
    Let us consider the case $\theta_j = 0^+$ and $\phi_j = a(j-1/2)$ with $a \in \mathbb{R}$. Injecting these angles in equation \eqref{eq:co_laurent}, we find that it leads to a block encoding of LBPs associated to the following recurrence coefficients:
\begin{equation}
    d_j = b_j = - e^{-ia(j-1/2)}.
\end{equation}
These polynomials have been the subject of investigation in relation to combinatorial tiling problems. Theorem 5.1 of \cite{kamioka2014laurent} states that the associated moments are
\begin{equation}
    c_k =\left\{
	\begin{array}{ll}
		\kappa N_{k-1}(t,q)  & \mbox{if } k \geq 1, \\
		\kappa   t^{2k-1}N_{-k}(t,q^{-1})& \mbox{if } k \leq 0,
	\end{array}
\right. 
\end{equation}
where $\kappa$ is an arbitrary constant,  $ N_{k-1}(t,q)$ denotes the $q$-Narayana polynomial and  $q = e^{-i a/2}$, $t = - e^{i a/2}$. This formula, together with equation \eqref{eq:sol}, provides an expression for the sequences of polynomials $P_i$ and $Q_i$ block-encoded by generalized QSP for these angles.

\end{example}

\section{Continuous variable QSP and orthogonality on the unit circle}\label{sec:4}

In this section, we consider the continuous-variable QSP (or ${\rm SU}(1,1)$-QSP)  model introduced in \cite{rossi2023quantum}, which is obtained by replacing the ${\rm SU}(2)$ rotations in generalized quantum signal processing with elements of ${\rm SU}(1,1)$. In the following, we refer to it as  ${\rm SU}(1,1)$-QSP. Here, we show that the analysis of the polynomial sequences it generates is  simpler than that of standard generalized QSP, as they are in one-to-one correspondence with orthogonal polynomials on the unit circle (OPUC).

\subsection{Circuit and recurrence for ${\rm SU}(1,1)$-QSP}

The circuit for ${\rm SU}(1,1)$-QSP is identical to that of generalized QSP, except that the ${\rm SU}(2)$ rotation $R(\theta,\phi)$ is replaced by its ${\rm SU}(1,1)$ analogue via the substitution $\theta \mapsto i\theta$. Since the resulting transformation is non-unitary, it can be implemented on a standard quantum device using block encoding, with an additional ancillary qubit introduced at each step (see Appendix \ref{app:newappsu11} for a detailed discussion). The corresponding circuit is shown in Figure~\ref{fig:circuitsu11QSP}. Denoting by \(W_\pm(\tilde{\theta},\phi,U)\) the unitary induced by this circuit, it admits the following block structure,
\begin{equation}\label{eq:bssu11qsp}
    \bra{0}W_\pm(\tilde{\theta}, \phi, U)\ket{0} = \begin{pmatrix}
        e^{i\phi} \cos^2 \frac{\tilde{\theta}}{2} U  &  \pm i e^{i \phi} \sin^2 \frac{\tilde{\theta}}{2} \\
       \pm i  \sin^2 \frac{\tilde{\theta}}{2} U &   -  \cos^2 \frac{\tilde{\theta}}{2} 
    \end{pmatrix} = \frac{T(i \theta, \phi, U)}{|\sinh \theta| + \cosh{\theta}},
\end{equation}
where $\theta$ and $\tilde{\theta}$ are related through $\tanh \theta = \pm \tan^2(\tilde{\theta}/2)$, and $T(i\theta,\phi,U)$ denotes the transformation implemented by generalized QSP in~\eqref{eq:transfgqsp} with $\theta \mapsto i\theta$. A possible physical realization of $T(i\theta,\phi,U)$ based on interferometry is also discussed in~\cite{rossi2023quantum}.

\begin{figure}[t]
    \centering
\begin{quantikz}[wire types = {q,q,q,n,q}]
\lstick{$\ket{0}$} & & &\gate{R_y(-\tilde{\theta})}  & \ctrl{1} & \gate{R_y(\tilde{\theta})}& & & \\ 
&\gate{R(0,\phi)} & \gate{Z}  & \gate{R_z(\mp \pi/2)}&  \targ{0}& \gate{R_z(\pm \pi/2)}& \octrl{1} && \\ 
& & &  && &\gate[3]{U} & & \\
& & &  & \vdots &  & & &\\
& &&  &  && & & \\
\end{quantikz}
    \caption{Circuit associated to ${\rm SU}(1,1)$-QSP (one iteration)}
    \label{fig:circuitsu11QSP}
\end{figure}
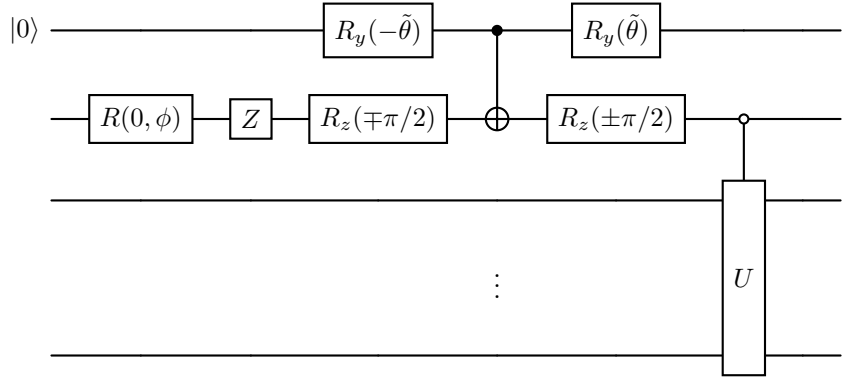

To identify the polynomial sequences generated by repeated application of this circuit, we apply the substitution $\theta \mapsto i\theta$ to the recurrences \eqref{eq:rec1} and \eqref{eq:rec2} that define the polynomial sequence in generalized QSP. We then find that the polynomials block-encoded at each iteration of ${\rm SU}(1,1)$-QSP satisfy the following recurrence relation
\begin{equation}\label{eq:rec1OPUC}
    P_{j}(z) = e^{i \phi_j} \cosh \theta_j z P_{j-1}(z) +  i e^{i \phi_j} \sinh \theta_j  Q_{j-1}(z),
\end{equation}
\begin{equation}\label{eq:rec2OPUC}
    Q_{j}(z) = i \sinh \theta_j z P_{j-1}(z) -\cosh \theta_{j}  Q_{j-1}(z).
\end{equation}
These relations can be identified with the Szegő recurrence, which are associated to sequences of OPUC. This topic is addressed in the following subsection.

\subsection{Relation between ${\rm SU}(1,1)$-QSP and orthogonal polynomials on the unit circle}

We now recall some definitions and theorems from the theory of orthogonal polynomials on the unit circle.

\begin{definition}
    Let $\mu$ be a measure on the unit circle $\mathbb{T}$. A sequence of polynomials $\{P_i(z)\}$ is said to be \textit{orthogonal on the unit circle with respect to $\mu$} if there exist strictly positive constants $\chi_i \in \mathbb{R}$ such that
\begin{equation}\label{eq:orthoUC}
    \mathcal{L}[\overline{{P}_i(z)}{P}_j(z)]  = \int_{\mathbb{T}} \overline{{P}_i(z)}{P}_j(z) d\mu(z) = \chi_i \delta_{ij}.
\end{equation}
\end{definition}

\noindent The theory of OPUC includes the following theorem, which characterizes the recurrence relations they satisfy.

\begin{theorem}(reformulation of  Theorem 1.5.2 of \cite{simon2005orthogonal}) Let $\{\hat{P}_i(z)\}$ be a sequence of monic orthogonal polynomials on the unit circle with respect to a measure $\mu$. Then, there exists a sequence of coefficients $\{\ver_j \,|\, j = 1,2,\dots\}$ with $|\ver_j| < 1$ such that the following recurrence relations hold:
\begin{equation}\label{eq:szego1}
     \hat{P}_{j}(z) = z\, \hat{P}_{j-1}(z) - \overline{\ver}_j\, \tilde{Q}_{j-1}(z),
\end{equation}
\begin{equation}\label{eq:szego2}
\tilde{Q}_{j}(z) = \tilde{Q}_{j-1}(z) - \ver_j\, z\, \hat{P}_{j-1}(z),
\end{equation}
where $\tilde{Q}_j(z) = z^j \overline{\hat{P}_j(1/\bar{z})}$.
\end{theorem}\label{theo:opuc1}

\noindent Analogously to Favard’s theorem for OPRL, the converse also holds: any sequence satisfying the previous recurrence relations forms a sequence OPUC.

\begin{theorem}(Verblunsky Theorem, reformulation of Theorem 1.7.11 of \cite{simon2005orthogonal}) Let $\{\hat{P}_i(z)\}$ be a sequence of monic polynomials satisfying the recurrence relations \eqref{eq:szego1} and \eqref{eq:szego2}, with $|\ver_j| < 1$ and $\tilde{Q}_j(z) = z^j \overline{\hat{P}_j(1/\bar{z})}$. Then, there exists a measure $\mu$ on the unit circle with respect to which $\{\hat{P}_i(z)\}$ forms a sequence of orthogonal polynomials on the unit circle.
\end{theorem}\label{theo:ver}

\noindent We can use these theorems to derive the following characterization of polynomial sequences generated by ${\rm SU}(1,1)$-QSP.

\begin{proposition}\label{prop:corroprl12}
    Let $\{\hat{P}_i(x)\}_{i=0}^n$ be a sequence of monic polynomials. There exist angles $\{\tilde{\theta}_i\}_{i=1}^n$, and $\{\phi_i\}_{i=1}^n$, and constants $\alpha_i$ such that, for each $i=0,1,\dots,n$, an $(\alpha_i,n+1,0)$-block encoding of $\hat{P}_i(U)$ is realized after $i$ iterations of the circuit shown in Figure~\ref{fig:circuitsu11QSP} if and only if $\deg \hat{P}_i = i$ and there exists a measure $\mu$ on the unit circle $\mathbb{T}$ along with coefficients $\chi_i > 0$ such that
    \begin{equation}\label{eq:orthoUC}
    \mathcal{L}[\overline{\hat{P}_i(z)}\hat{P}_j(z)]  = \int_{\mathbb{T}} \overline{\hat{P}_i(z)}\hat{P}_j(z) d\mu(z) = \chi_i \delta_{ij}.
\end{equation}
Furthermore, the norms $\chi_i$ and block encoding constants $\alpha_i$ are expressed in terms of the angles as
\begin{equation}\label{eq:becsu11}
    \chi_i = \prod_{k=1}^i \cosh^2 \theta_k, \qquad 
    \alpha_i = \prod_{k=1}^{i} {e^{i\phi_k + |\theta_k|}}{\cosh \theta_k},
\end{equation}
where $\tanh \theta_k = \pm \tan^2(\tilde{\theta}_k/2)$.
\end{proposition}
\begin{proof}
  First, using equation \eqref{eq:bssu11qsp}, we observe that $j$ iterations of the circuit in Figure~\ref{fig:circuitsu11QSP} yield the following block structure:
\begin{equation}\label{eq:bssu11qsp}
    (\bra{0}^{\otimes j} \otimes I )\, W_j W_{j-1} \cdots W_{1}\, (\ket{0}^{\otimes j} \otimes I )
    = \frac{T(i \theta_j, \phi_j, U)\, T(i \theta_{j-1}, \phi_{j-1}, U)\cdots T(i \theta_1, \phi_1, U)}{\prod_{k=1}^j \left(|\sinh \theta_k| + \cosh \theta_k\right)},
\end{equation}
where $W_k = W_\pm(\tilde{\theta}_k, \phi_k, U)$, $\tanh \theta_k = \pm \tan^2(\tilde{\theta}_k/2)$ and
\begin{equation}\label{eq:begqspsu11}
   \begin{pmatrix}
       P_j(U) & \cdot \quad \\
       Q_j(U) & \cdot \quad 
   \end{pmatrix} 
   = T(i\theta_j,\phi_j,U)\dots T(i\theta_1,\phi
   _1, U),
\end{equation}
The structure of the matrix $T$ implies the recurrence relations \eqref{eq:rec1OPUC} and \eqref{eq:rec2OPUC}. In terms of the rescaled polynomials $\hat{P}_j$ and $\tilde{Q}_j$, defined by
\begin{equation}\label{eq:nonosu11}
    \hat{P}_j(z) = \left(\prod_{k = 1}^{j} \frac{e^{-i \phi_k}}{\cosh\theta_k}\right) P_j(z), 
    \qquad 
    \tilde{Q}_j(z) = (-1)^j \left(\prod_{k = 1}^{j} \frac{1}{\cosh\theta_k}\right) Q_j(z),
\end{equation}
%and using the identification $\tilde{Q}_n(z) = z^n \overline{\hat{P}_n(1/\bar{z})}$,
these recurrence relations coincide exactly with the Szegő recurrence for OPUC given in \eqref{eq:szego1} and \eqref{eq:szego2}, where the coefficients $\ver_j$ are given by
\begin{equation}\label{eq:anglesOPUC}
    \ver_j = i\left(\prod_{k = 1}^{j} -e^{i \phi_k} \right)\tanh{\theta_j},
\end{equation}
and satisfy $|\ver_j| < 1$. From Theorem \ref{theo:opuc1}, we obtain that sequences of OPUC provide examples of polynomial sequences generated by this construction. Conversely, Theorem \ref{theo:ver} gives the reverse implication, and the main result follows. The block encoding constant is obtained from the normalization factors appearing in equations \eqref{eq:bssu11qsp} and \eqref{eq:nonosu11}. The norm follows from the formula $\chi_i = \prod_{k=1}^i (1 - |\ver_k|^2)^{-2}$, which expresses it in terms of the Verblunsky coefficients for OPUC (see Theorem 1.5.2 in \cite{simon2005orthogonal}).
\end{proof}

\begin{example}
    The simplest example of OPUC is given by the monomials $\hat{P}_n(z) = z^n$. These polynomials are orthogonal with respect to the uniform measure on the unit circle, i.e.
    \begin{equation}
        \mathcal{L}[\overline{\hat{P}_k(z)} \hat{P}_\ell(z)]=  \int_{0}^{2\pi}  \overline{\hat{P}_k(e^{i\theta})} \hat{P}_\ell(e^{i\theta}) \frac{d\theta}{2\pi} = \int_{0}^{2\pi} e^{-i\theta (k-\ell)} \frac{d\theta}{2\pi} = \delta_{ij}
    \end{equation}
    They are obtained from the Szegő recurrence with Verblunsky coefficients $\ver_j = 0$. Using equation \eqref{eq:anglesOPUC}, this corresponds, as expected, to angles $\phi_k = \theta_k = 0$. In particular, the ancillary qubits play no role in the ${\rm SU}(1,1)$-QSP circuit in this case. This is consistent with the fact that these polynomials can be implemented trivially, without ancillary qubits, by repeated application of $U$.
\end{example}

\begin{example} The Roger--Szegő $\Phi_i(z)$ polynomials are defined with respect to a complex parameter $q$ with $|q|< 1$. They are given by
\begin{equation}
    \Phi_i(z) = \sum_{j=0}^i (-1)^{i-j} \qbinom[q]{\,i\,}{j} q^{(i-j)/2} z^j,
\end{equation}
where $\qbinom[q]{\,i\,}{j}$ is the gaussian binomial coefficient
\begin{equation}
    \qbinom[q]{\,i\,}{j} = \frac{(1-q^i)(1-q^{i-1})\dots (1-q^{i-j+1})}{(1-q)(1-q^{2})\dots (1-q^{j})}
\end{equation}
These polynomials are known to have Verblunsky coefficients $\ver_i = (-1)^{i-1} q^{i/2}$, are orthogonal on the unit circle with respect to a so-called wrapped Gaussian measure (see Example 1.6.5 in \cite{simon2005orthogonal}), and are associated with moments $c_i = q^{i^2/2}$ \cite{simon2005orthogonal}. Let us denote $q = |q| e^{2 i \phi}$. Using equation \eqref{eq:anglesOPUC}, we find that the following angles in the ${\rm SU}(1,1)$-QSP protocol leads to a block encoding of the Roger-Szegő polynomials,
\begin{equation*}
    \theta_i = \tanh^{-1}\left(|q|^{i/2}\right), \qquad \phi_i = \left\{
	\begin{array}{ll}
		\phi  & \mbox{if } i \geq 2, \\
		\phi + \pi/2 & \mbox{if } i = 1.
	\end{array}
\right.
\end{equation*}
    The associated norms and block encoding constant is obtained from these angles, equation \eqref{eq:becsu11}, and standard identities for hyperbolic functions. They satisfy
    \begin{equation*}
        \chi_i = \prod_{k=1}^i \frac{1}{1 - |q|^k} \qquad |\alpha_i| = \prod_{k=1}^i \frac{1}{1 - |q|^{k/2}},
    \end{equation*}
    and converges to a constant as $i \to \infty$.
\end{example}

\subsection{Expressibility of ${\rm SU}(1,1)$-QSP}

Since ${\rm SU}(1,1)$-QSP admits the block structure of equation \eqref{eq:bepoly}, a controlled version of this circuit can be used as the select oracle in the LCU framework, so that Lemma \ref{prop:sed2} applies. This yields the following result on the block encoding constant achievable from a polynomial expansion of a function in terms of OPUC.

\begin{proposition}\label{prop:3.3opuc}
Let $f$ be a function that admits a degree-$n$ polynomial $P$ approximating it within $\epsilon$ on the unit circle $\mathbb{T}$,
\begin{equation*}
    |f(z) - P(z)| \le \epsilon, \quad z \in \mathbb{T}.
\end{equation*}
Let $\{\hat{P}_i(z)\}_{i=0}^n$ be a family of monic orthogonal polynomials on the unit circle, and let $\alpha_i$ denote the associated block encoding constants defined in~\eqref{eq:becsu11}. Let $v_i$ be the coefficients of $P$ in this basis, i.e.
\begin{equation*}
    P(z) = \sum_{i=0}^n \frac{v_i}{\alpha_i} \hat{P}_i(z), 
    \qquad 
    \|v\|_1 = \sum_{i=0}^n |v_i|.
\end{equation*}
Then, an $(\|v\|_1, 2n +1, \epsilon)$-block encoding of $f(U)$ can be implemented with $O(n)$ depth and gate complexity, including $n$ applications of controlled-$U$.
\end{proposition}

\begin{proof}
    We obtain from equations \eqref{eq:bssu11qsp} and \eqref{eq:begqspsu11} that the circuit associated with ${\rm SU}(1,1)$-QSP admits the block structure of \eqref{eq:bepoly} with $\ell = n+1$. From the Proposition \ref{prop:corroprl12}, it follows that for any sequence of OPUC there exists a corresponding sequence of angles such that each step of the ${\rm SU}(1,1)$-QSP circuit realizes the associated block encoding. The result then follows from Lemma \ref{prop:sed2}, which shows that any linear combination of these polynomials can be implemented using a controlled version of this circuit as the select oracle in an LCU protocol. The complexity is determined by that of the select oracle, which requires $n$ controlled copies of the circuit in Figure \ref{fig:circuitsu11QSP} in order to generate polynomials of degree up to $n$.
\end{proof}

The identification between polynomial sequences generated by ${\rm SU}(1,1)$-QSP and sequences of OPUC established in Proposition \ref{prop:corroprl12} allows one to transfer results established in the latter setting to quantum signal processing. For instance, we have the following result concerning the zeros of OPUC.

\begin{theorem} (Zeros Theorem, reformulation of Theorem 1.7.1 of \cite{simon2005orthogonal}) Let $P_n(z)$ be an OPUC polynomial. Then all zeros of $P_n(z)$ lie strictly inside the open unit disk,
\begin{equation*}
    P_n(z) = 0 \quad \Rightarrow \quad |z| < 1.
\end{equation*}
\end{theorem}\label{theo:zeopuc}

\noindent By this Zeros Theorem for OPUC, we obtain the following corollary.

\begin{corollary}
    Let $P(U)$ be a polynomial whose block encoding can be achieved by ${\rm SU}(1,1)$-QSP. If $P(z) = 0$ then  $|z| < 1$.
\end{corollary}

Moreover, it is known that the converse of Theorem \ref{theo:zeopuc} also holds, namely that any polynomial whose roots lie strictly inside the unit circle can be realized as an OPUC polynomial. This result provides an answer to the problem raised in \cite{rossi2023quantum} of characterizing the polynomials implementable by ${\rm SU}(1,1)$-QSP. In particular, it shows that this class consists precisely of polynomials whose roots lie inside the unit circle.

\begin{theorem} (reformulation  of Theorem 1.7.5 of \cite{simon2005orthogonal}) Let $\hat{P}_n(z)$ be a monic polynomial whose roots lie strictly inside the unit disk. Then $\hat{P}_n(z)$ is an orthogonal polynomial on the unit circle with respect to some measure $\mu$.
\end{theorem}\label{theo:zeopucconv}

\begin{corollary}
    Let $P(U)$ be a polynomial of degree $n$ whose roots lie inside the unit disk. Then there exists a nonzero scalar $\alpha$ such that an $(\alpha,n+1,0)$-block encoding of the polynomial $ P(U)$ can be implemented by ${\rm SU}(1,1)$-QSP using $n$ queries to controlled $U$.
\end{corollary}

In the following subsection, we address the problem of determining the angles required in ${\rm SU}(1,1)$-QSP to implement a given target polynomial whose roots lie inside the unit disk. We provide an explicit and constructive procedure for obtaining these angles, expressed in terms of the moments of a linear functional associated with the target polynomial.

\subsection{Angle-finding problem for ${\rm SU}(1,1)$-QSP}

The Bernstein–Szegő approximation provides a method to approximate the moments of a linear functional associated with a sequence of OPUC using a finite-degree polynomial $\hat{P}_n(z)$ from the sequence \cite{simon2005orthogonal}. The following proposition leverages this result to solve the angle-finding problem for ${\rm SU}(1,1)$-QSP. In particular, it expresses the corresponding QSP angles for a target polynomial $\hat{P}_n(z)$ in terms of these moments.

\begin{proposition}
  Let $\hat{P}_n(z)$ be a monic polynomial whose zeros all lie inside the unit circle. Let $c_k$ denote the following coefficients
    \begin{equation}\label{eq:ckprop1}
    c_k =  \mathcal{L}[z^k] = \int_{0}^{2\pi}  z^k d\mu(z)  = \frac{1}{2\pi } \int_{0}^{2\pi} \frac{e^{ik\theta} d\theta}{ |\hat{P}_n(e^{i\theta})|^2},
\end{equation}
where $k  = -n, -n+1, \dots, n$. If for every $j = 1,2,\dots, n$ we have 
\begin{equation}\label{eq:conpropsu11}
    \begin{vmatrix}
    c_0 & c_{1} & \dots & c_{j}\\
    c_{-1} & c_{0} & \dots & c_{j-1}\\
    \vdots & \vdots & \ddots & \vdots\\
    c_{-j} & c_{-j+1} & \dots & c_{0}
    \end{vmatrix} \neq 0
\end{equation}
Then a $(\alpha_n, n+1, 0)$-block encoding of $\hat{P}_n$ is obtained via ${\rm SU}(1,1)$-QSP, with angle sequences $\{\phi_j\}_{j=1}^n$ and $\{\theta_j\}_{j=1}^n$ determined by
    \begin{equation}\label{eq:anglesu11prop}
\phi_j = \pi + \text{arg}(\ver_j) -  \text{arg}({\ver_{j-1}}), \quad \theta_j = \tanh^{-1}|\ver_j|,
\end{equation}
where $\text{arg}(\ver_{0}) = \pi/2$ and for $j = 1,2,\dots, n$
\begin{equation}\label{eq:vermo}
    \bar{\ver}_j =  (-1)^{j+1}\begin{vmatrix}
    c_1 & c_{2} & \dots & c_{j}\\
    c_{0} & c_{1} & \dots & c_{j-1}\\
    \vdots & \vdots & \ddots & \vdots\\
    c_{-j+2} & c_{-j+1} & \dots & c_{1}
    \end{vmatrix}\begin{vmatrix}
    c_0 & c_{1} & \dots & c_{j-1}\\
    c_{-1} & c_{0} & \dots & c_{j-2}\\
    \vdots & \vdots & \ddots & \vdots\\
    c_{-j+1} & c_{-j+2} & \dots & c_{0}
    \end{vmatrix}^{-1}.
\end{equation}
\end{proposition}

\begin{proof}
    Based on Proposition \ref{prop:corroprl12}, block-encoding $\hat{P}_n(U)$ using ${\rm SU}(1,1)$-QSP amounts to finding a sequence of monic polynomials $\{\hat{P}_i(z)\}_{i=0}^n$ orthogonal on the unit circle, with $\deg \hat{P}_i = i$ and $\hat{P}_n$ being the target polynomial. Such a sequence is orthogonal with respect to a positive definite measure $d\mu(z)$, which can be explicitly described using the Bernstein–Szegő approximation \cite{simon2005orthogonal}. If $q(z)$ is a Laurent polynomial whose lowest and highest degrees are at most $n$, this approximation states that
\begin{equation}\label{eq:bern-app}
    \mathcal{L}[q(z)] = \int_{0}^{2\pi}  q(z) d\mu(z)  = \frac{1}{2\pi } \int_{0}^{2\pi} \frac{q(e^{i\theta}) d\theta}{ |\hat{P}_n(e^{i\theta})|^2}.
\end{equation}
In particular, the coefficients $c_k = \mathcal{L}[z^k]$ of equation \eqref{eq:ckprop1} are the moments of the linear functional $\mathcal{L}$ associated with the OPUC sequence whose $n$th orthogonal polynomial is given by $\hat{P}_n$.  To identify the full sequence of OPUC associated to $\mathcal{L}$, we note that the orthogonality relations 
\begin{equation*}
    \mathcal{L}[\overline{\hat{P}_i(z)}\hat{P}_j(z)] = \chi_i \delta_{ij},
\end{equation*}
are equivalent to the following conditions,
\begin{equation*}
    \mathcal{L}[z^{-j} \hat{P}_i(z)] = \chi_i \delta_{ij} , \quad 0 \leq j \leq i.
\end{equation*}
 These conditions lead to a linear system of equations for the coefficients in the monomial expansion of the polynomials $\hat{P}_i$. This system admits the following solution
\begin{equation}\label{eq:monicOPUC}
     \hat{P}_i(z) = \begin{vmatrix}
    c_0 & c_{1} & \dots & c_{i}\\
    c_{-1} & c_{0} & \dots & c_{i-1}\\
    \vdots & \vdots & \ddots & \vdots\\
    c_{-i+1} & c_{-i+2} & \dots & c_{1}\\
        1 & z & \dots & z^{i}
    \end{vmatrix} h_i^{-1}, \qquad h_i = \begin{vmatrix}
    c_0 & c_{1} & \dots & c_{i-1}\\
    c_{-1} & c_{0} & \dots & c_{i-2}\\
    \vdots & \vdots & \ddots & \vdots\\
    c_{-i+1} & c_{-i+1} & \dots & c_{0}
    \end{vmatrix},
\end{equation}
with $\chi_i = h_{i+1}/h_i$. To verify that this is indeed a solution, one observes that for $0 \leq j \leq i$
\begin{equation*}
    \mathcal{L}\Big[z^{-j}\begin{vmatrix}
    c_0 & c_{1} & \dots & c_{i}\\
    c_{-1} & c_{0} & \dots & c_{i-1}\\
    \vdots & \vdots & \ddots & \vdots\\
    c_{-i+1} & c_{-i+2} & \dots & c_{1}\\
        1 & z & \dots & z^{i}
    \end{vmatrix} \Big] = \begin{vmatrix}
    c_0 & c_{1} & \dots & c_{i}\\
    c_{-1} & c_{0} & \dots & c_{i-1}\\
    \vdots & \vdots & \ddots & \vdots\\
    c_{-i+1} & c_{-i+2} & \dots & c_{1}\\
        c_{-j}  &  c_{-j+1} & \dots &  c_{i-j}
    \end{vmatrix} = h_{i+1}\delta_{ij}.
\end{equation*}
From the properties of the linear system of equations for the coefficients in the monomial expansions of the polynomials $\hat{P}_i$, we find that this solution exists and is unique if and only if equation \eqref{eq:conpropsu11} holds. In that case, equation \eqref{eq:monicOPUC} provides an explicit expression for the sequence of OPUC associated with the linear functional $\mathcal{L}$, with $\hat{P}_n$ coinciding with the target polynomial. It follows from Theorem \ref{theo:opuc1} that these polynomials satisfy the Szegő recurrence relations. The coefficients $\ver_j$ appearing in these relations can be determined by substituting the expression \eqref{eq:monicOPUC} into the recurrence and matching the constant and leading order terms. This yields
\begin{equation}\label{eq:vermo}
    \bar{\ver}_j =  (-1)^{j+1}\begin{vmatrix}
    c_1 & c_{2} & \dots & c_{j}\\
    c_{0} & c_{1} & \dots & c_{j-1}\\
    \vdots & \vdots & \ddots & \vdots\\
    c_{-j+2} & c_{-j+1} & \dots & c_{1}
    \end{vmatrix}\begin{vmatrix}
    c_0 & c_{1} & \dots & c_{j-1}\\
    c_{-1} & c_{0} & \dots & c_{j-2}\\
    \vdots & \vdots & \ddots & \vdots\\
    c_{-j+1} & c_{-j+2} & \dots & c_{0}
    \end{vmatrix}^{-1}.
\end{equation}
In the derivation of Proposition \ref{prop:corroprl12}, it was shown that the Szegő recurrence relations are equivalent to the recurrence implemented by ${\rm SU}(1,1)$-QSP, with the coefficients $\ver_j$ related to the circuit angles via equation \eqref{eq:anglesOPUC}. Expressions for the angles $\theta_j$ and $\phi_j$ in terms of the moments $c_j$ then follow from relation \eqref{eq:vermo} together with \eqref{eq:anglesOPUC}. One finds equation \eqref{eq:anglesu11prop}
with the convention $\text{arg}(\ver_{0}) = \pi/2$, and the result follows.
\end{proof}

\section{Bivariate QSP and biorthogonality}\label{sec:5}

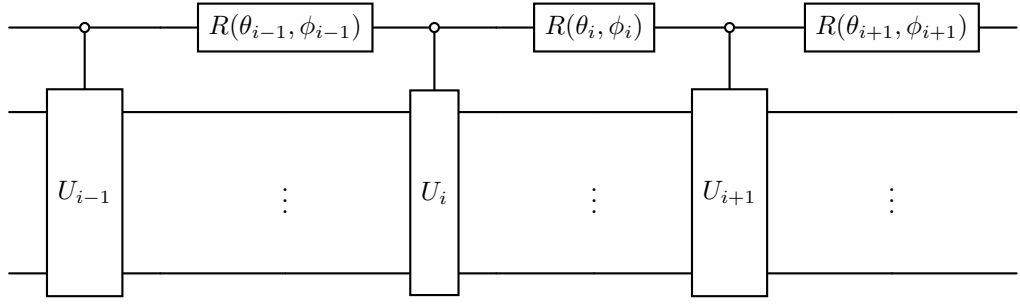
\begin{figure}[t]
    \centering
\begin{quantikz}[wire types = {q,q,n,q}]
    & \octrl{1} &  &\gate{R(\theta_{i-1},\phi_{i-1})} & \octrl{1} & &\gate{R(\theta_{i},\phi_{i})} & \octrl{1} & \gate{R(\theta_{i+1},\phi_{i+1})} &\\ &\gate[3]{U_{i-1}} &&&\gate[3]{U_{i}} &&&\gate[3]{U_{i+1}} & &\\
  & &  & \vdots &  &  & \vdots  &  &  \vdots &\\
  && & && & && &  \\
\end{quantikz}
    \caption{Standard circuit for generalized multivariate QSP (three iterations). The corresponding unitary is given by
$T(\theta_{i+1}, \phi_{i+1}, U_{i+1})\, T(\theta_i, \phi_i, U_i)\, T(\theta_{i-1}, \phi_{i-1}, U_{i-1})$, where $T$ denotes the unitary corresponding to a single iteration of generalized QSP, as defined in equation \eqref{eq:transfgqsp}. }
    \label{fig:circuitGQSPm}
\end{figure}

In this section, we consider the approach to bivariate QSP introduced in \cite{rossi2022multivariable}, which aims to implement a block encoding of a polynomial $P(U,V)$ in two commuting unitaries $U$ and $V$. We show that the resulting sequence of polynomials forms part of a biorthogonal basis. The class of bivariate polynomials that can be block-encoded using this approach is known to be limited and has yet to be fully characterized. As a contribution toward this goal, we derive necessary conditions on the admissible polynomials based on biorthogonality relations.

%\textcolor{blue}{NW: Explain multi-degree and other common concepts discussed below that the reader may not have seen.}

\subsection{Circuit and recurrence for bivariate QSP}

Let us assume access to the $0$-controlled versions of two commuting unitaries, $U$ and $V$. The standard approach to bivariate QSP reuses the univariate QSP circuit structure (see Figure \ref{fig:circuitGQSP}), with controlled applications of $U$ or $V$ at each iteration. A segment of the circuit consisting of three iterations is shown in Figure \ref{fig:circuitGQSPm}, where $U_i \in \{U,V\}$. In principle, this multivariate QSP approach can be extended to any number of commuting unitaries. However we restrict our attention to two unitaries, $U$ and $V$, since characterizing the set of polynomials it implements remains an open question even in the bivariate case.

For $n$ iterations with different angles $\theta_i$ and $\phi_i$,  the transformation implemented by the algorithm is given by
\begin{equation*}
   \begin{pmatrix}
       P_n(U,V) & \cdot \quad \\
       Q_n(U,V) & \cdot \quad 
   \end{pmatrix} 
   = T(\theta_n,\phi_n,U_n) \dots T(\theta_1,\phi
   _1, U_1),
\end{equation*}
where $T(\theta_i, \phi_i, U_i)$ denotes the unitary corresponding to a single iteration of generalized QSP, as defined in equation \eqref{eq:transfgqsp}. In particular, the polynomials implemented at steps $n$ and $n-1$ are related by
\begin{equation*}
    \begin{pmatrix}
       P_n(U,V) & \cdot \quad \\
       Q_n(U,V) & \cdot \quad 
   \end{pmatrix} 
   = T(\theta_n,\phi_n,U_n)  \begin{pmatrix}
       P_{n-1}(U,V) & \cdot \quad \\
       Q_{n-1}(U,V) & \cdot \quad 
   \end{pmatrix} 
\end{equation*}
From the structure of the matrix $T(\theta,\phi,U)$, we thus derive the following recurrence relations for the polynomial sequences $P_n$ and $Q_n$, now viewed as polynomials in the complex variables $z$ and $w$:
 \begin{equation}\label{eq:rec1biv}
    P_{n}(z,w) = e^{i \phi_n} \cos \theta_n z_n P_{n-1}(z,w) + e^{i \phi_n} \sin \theta_n  Q_{n-1}(z,w),
\end{equation}
\begin{equation}\label{eq:rec2biv}
    Q_{n}(z,w) =\sin \theta_n z_n P_{n-1}(z,w) -\cos \theta_n  Q_{n-1}(z,w).
\end{equation}
with $z_i \in\{z,w\}$. To discuss the properties of the polynomial solutions to these equations, it is useful to introduce a notion analogous to the degree in the univariate setting. For multivariate polynomials, this generalization is referred to as the multidegree and corresponds to a vector recording the highest exponent of each variable.

\begin{definition}
Let $P(z,w) = \sum_{i,j} r_{ij} z^i w^j$ be a polynomial in two variables $z$ and $w$. The multidegree of $P$ is defined as
\begin{equation*}
    \mathrm{multideg}\, P =
    \left(
    \max\{ i \mid r_{ij} \neq 0 \},
    \max\{ j \mid r_{ij} \neq 0 \}
    \right).
\end{equation*}
\end{definition}
\noindent As long as $\theta_i \notin \left\{ \frac{k\pi}{2} \right\}_{k \in \mathbb{Z}}$, the structure of the recurrence relations \eqref{eq:rec1biv} and \eqref{eq:rec2biv} imply that $P_n$ and $Q_n$ contain the monomial $z_1 z_2 \cdots z_n$, and thus
\begin{equation*}
    \mathrm{multideg}\, P_n = \mathrm{multideg}\, Q_n =  \mathrm{multideg}\, z_1 z_2\dots z_n.
\end{equation*}

\subsection{Biorthogonality in two variables}
For bivariate polynomials satisfying the recurrence relations \eqref{eq:rec1biv} and \eqref{eq:rec2biv}, we can derive two types of biorthogonality relations associated with different linear functionals. The first approach uses the map $z^i w^j \rightarrow z^{i + k j}$ to obtain a bijection between vector space spanned by monomials in one or two variables:
\begin{equation}
    \text{span}\{ z^i w^j \, |\, 0 \leq i \leq k-1, \, \, 0 \leq j \leq \ell - 1\} \cong \text{span}\{ z^i \, |\, 0 \leq i \leq k \ell -1\}
\end{equation}
Using this bijection and an argument parallel to that of Proposition~\ref{prop:form}, one can show that the sequence of polynomials arising in bivariate QSP satisfies a biorthogonality relation analogous to that in the univariate QSP setting.

\begin{proposition}\label{prop:formbiv1}
    Let $\{P_i(z,w)\}_{i=0}^n$ and $\{Q_i(z,w) \}_{i=0}^n$  be finite sequences of polynomials satisfying the recurrence relations  \eqref{eq:rec1biv} and \eqref{eq:rec2biv}, with $\mathrm{multideg}\, P_i = \mathrm{multideg}\, Q_i = (a_i,b_i)$  and $\theta_i \notin \left\{ \frac{k\pi}{2} \right\}_{k \in \mathbb{Z}}$. Then there exists a linear functional $\mathcal{L}$ such that we have the following biorthogonal relations:
    \begin{equation}\label{eq:condformbiv}
        \mathcal{L}[\tilde{Q}_i(z,w)P_j(z,w) ] \propto \delta_{ij}
    \end{equation}
    where $\tilde{Q}_i(z,w) = z^{-a_i} w^{-b_i}Q_i(z,w)$ and such that the $\mathcal{L}$ has a contour integral representation,
    \begin{equation}\label{eq:linfun3}
         \mathcal{L}[f(z,w)] = \int_\Gamma \frac{f(z,z^{a_n+1})}{z P_n(z,z^{a_n+1})\tilde{Q}_n(z,z^{a_n+1})} dz,
    \end{equation}
    where $\Gamma$ denotes a closed contour whose interior contains $z=0$ and all roots of $P_n(z,z^{a_n+1})$, but none of the roots of $Q_n(z,z^{a_n+1})$.
\end{proposition}

\noindent The proof of this proposition is laid out in Appendix \ref{sec:app2}. 

\begin{remark}
   In the linear functional $\mathcal{L}$ of Proposition~\ref{prop:formbiv1}, the variable $z$ plays a more prominent role than $w$. This stems from the choice of mapping $z^i w^j \mapsto z^{\,i + j(a_n + 1)}$. An analogous result can be derived with a contour integral in $w$ instead by using the alternative mapping $z^i w^j \mapsto w^{(b_n + 1)i + j}$.
\end{remark}

\begin{remark}
    In the previous sections, the angle-finding problem was essentially solved by establishing and exploiting the converse of a result analogous to Proposition \ref{prop:formbiv1}, namely that the orthogonality or biorthogonality of a polynomial sequence implies a QSP recurrence relation. This approach does not extend to the bivariate setting. Indeed, for a given $\mathcal{L}$ the set of bivariate polynomials $\{p_{ij} \, | \, i,j \in \mathbb{N}\}$, where $p_{ij}$ has multidegree $(i,j)$ and satisfies
    \begin{equation}
       \mathcal{L}[\,p_{ij}(z,w) z^{-k} w^{-\ell}]  \propto \delta_{ik} \delta_{j\ell},
    \end{equation}
    for all $ 0 \leq k \leq i $ and $ 0 \leq \ell \leq j $, obeys recurrence relations that generally exceed the scope of those implemented by bivariate QSP and given in \eqref{eq:rec1biv} and \eqref{eq:rec2biv}. This indicates that a more general circuit architecture is required to implement arbitrary bivariate polynomials.
\end{remark}

A second approach to the biorthogonality of polynomial sequences in bivariate QSP consists in defining linear functionals that treat one of the variables as a parameter. In contrast to Proposition \ref{prop:formbiv1} and equation \eqref{eq:linfun3}, where biorthogonality is derived by reducing to a univariate framework via the mapping $z^i w^j \mapsto z^{\,i + j(a_n + 1)}$, here we reduce to a one-variable setting by treating $w$ as a parameter and allowing the contour $\Gamma \to \Gamma_w$ to depend on this parameter. This leads to the following proposition:

\begin{proposition}\label{prop:formbiv2}
    Let $\{P_i(z,w)\}_{i=0}^n$ and $\{Q_i(z,w)\}_{i=0}^n$  be finite sequences of polynomials satisfying the recurrence relations  \eqref{eq:rec1biv} and \eqref{eq:rec2biv}, with $\mathrm{multideg}\, P_i = \mathrm{multideg}\, Q_i = (a_i,b_i)$ and $\theta_i \notin \left\{ \frac{k\pi}{2} \right\}_{k \in \mathbb{Z}}$. Then, for all $k \in \mathbb{N}$ such that $1 - a_{i+1} + a_i \le k < a_i$, we have
\begin{equation}\label{eq:bivbirel2}
\int_{\Gamma_w} \frac{P_i(z,w)\, z^{-k}}{z\, P_n(z,w)\, \tilde{Q}_n(z,w)} \, dz = 0, \quad
\int_{\Gamma_w} \frac{\tilde{Q}_i(z,w)\, z^{k}}{z\, P_n(z,w)\, \tilde{Q}_n(z,w)} \, dz = 0,
\end{equation}
where $\tilde{Q}_i(z,w) = z^{-a_i} w^{-b_i}Q_i(z,w)$ %and where the linear functional $\mathcal{L}_w$ acts on functions of $z$ and depends on the parameter $w$, and admits the following contour integral representation:
% \begin{equation}\label{eq:linfun}
%         \mathcal{L}_w[f(z,w)] = \int_{\Gamma_w} \frac{f(z,w)}{z P_n(z,w)\tilde{Q}_n(z,w)} dz.
 %   \end{equation}
 and $\Gamma_w$ denotes a closed contour whose interior contains $z=0$ and all roots of $P_n(z,w)$, but excludes all roots of $Q_n(z,w)$ for the fixed value of $w$.
\end{proposition}

\noindent The proof of this proposition is given in Appendix \ref{sec:app2}. 

Using equation \eqref{eq:bivbirel2}, one can derive biorthogonality relations between certain polynomials in the sequences $\{P_i(z,w)\}_{i=0}^n$ and $\{Q_i(z,w)\}_{i=0}^n$. Let $\mathcal{I} = \{\, i \in \mathbb{N} \mid i \leq n,\ z_{i+1} = z \,\}$, and let $\mathcal{L}_w$ be the following linear functional acting on functions of $z$, with $w$ regarded as a fixed parameter,
\begin{equation*}
    \mathcal{L}_w[f(z)] = \int_{\Gamma_w} \frac{f(z)}{z\, P_n(z,w)\, \tilde{Q}_n(z,w)} \, dz
\end{equation*}
It follows from the previous proposition, together with the monomial expansions of $P_i(z,w)$ and $Q_i(z,w)$, that for $i,j \in \mathcal{I}$ there exist functions $\chi_i(w)$ such that
\begin{equation}\label{eq:bibi1}
\mathcal{L}_w[\tilde{Q}_i(z,w)\,P_j(z,w)] = \delta_{ij}\,\chi_i(w).
\end{equation}
To see this, consider, for instance, the case $i > j$. By the definition $(a_\ell,b_\ell)=\mathrm{multideg}(z_1 z_2 \cdots z_\ell)$ and the properties of $\mathcal{I}$, we have $a_i > a_j$ and $a_{j+1}-a_j = 1$. Moreover, $P_j(z,w)$ admits the decomposition $ P_j(z,w)=\sum_{k=0}^{a_j} r_k(w)\,z^k$, where each $r_k(w)$ is a polynomial in $w$. Substituting this expression into the left-hand side of \eqref{eq:bibi1}, we may then use \eqref{eq:bivbirel2} to conclude that it vanishes, since $
\mathcal{L}_w[z^k \tilde{Q}_i(z,w)] = 0$ for all $0 \leq k \leq a_j < a_i$. A similar argument holds for $i < j$.

\subsection{Necessary conditions on bivariate QSP polynomials}

In addition to providing biorthogonality relations, the Proposition \ref{prop:formbiv2} can be interpreted as imposing constraints on the set of bivariate polynomials that are realizable by bivariate QSP. Let $c_i(w)$ and $d_i(z)$ denote the following functions  
\begin{equation}\label{eq:lastdef1}
c_i(w)  = \int_{\Gamma_w} \frac{z^{i-1}}{ \, P_n(z,w) \, \tilde{Q}_n(z,w)} \, dz, \quad d_i(z)  = \int_{\Gamma_z} \frac{w^{i-1}}{ \, P_n(z,w) \, \tilde{Q}_n(z,w)} \, dw. 
\end{equation}
where $\Gamma_w$ (resp.\ $\Gamma_z$) is a closed contour whose interior contains $z=0$ (resp.\ $w=0$) together with all the roots of $P_n(z,w)$, but excludes the roots of $Q_n(z,w)$ for fixed $w$ (resp.\ fixed $z$). The following proposition gives a necessary condition for a pair of polynomials $P_n(U,V)$ and $Q_n(U,V)$ to be implementable by bivariate QSP.

%Furthermore, let us denote by $C_{ij}(w)$ and $D_{ij}(z)$ the following determinants:
\begin{comment}
\begin{equation}\label{eq:lastdef2}
   C_{ij}(w) =  \begin{vmatrix}
c_0(w) & c_1(w) & \dots &  c_{j-1}(w) & c_{j+1}(w)& \dots& c_{i}(w) \\
c_{-1}(w) & c_0(w) & \dots & c_{j-2}(w) & c_{j}(w)& \dots& c_{i-1}(w) \\
\vdots & \vdots & \ddots & \vdots &\vdots &\dots & \vdots \\
c_{-i+1}(w) & c_{-i+2}(w) & \dots &c_{-i+j}(w) & c_{-i+j+2}(w)& \dots& c_1(w) 
\end{vmatrix},
\end{equation}
\begin{equation}\label{eq:lastdef3}
   D_{ij}(z) =  \begin{vmatrix}
d_0(z) & d_1(z) & \dots &  d_{j-1}(z) & d_{j+1}(z)& \dots& d_{i}(z) \\
d_{-1}(z) & d_0(z) & \dots & d_{j-2}(z) & d_{j}(z)& \dots& d_{i-1}(z) \\
\vdots & \vdots & \ddots & \vdots &\vdots &\dots & \vdots \\
d_{-i+1}(z) & d_{-i+2}(z) & \dots &d_{-i+j}(z) & d_{-i+j+2}(z)& \dots& d_1(z) 
\end{vmatrix},
\end{equation}
where $i = 0,1,\dots,n$ and $j = 0,1,\dots,i$. We have the following proposition, which characterizes their ratios when $P_n$ and $Q_n$ are realizable by bivariate QSP.
\end{comment}
\begin{proposition}\label{prop:condibiv}
    Let $P_n(U,V)$ and $Q_n(U,V)$ be a pair of bivariate polynomials realizable by bivariate QSP in $n$ steps, and let $(a_i,b_i)=\mathrm{multideg}( U_1 U_2 \cdots U_i)$. Let $c_{i}(w)$ and $d_{i}(z)$, with $i = 0,1,\dots,n$, denote the quantities associated with $P_n(U,V)$ and $Q_n(U,V)$, as defined in equations \eqref{eq:lastdef1}. Then, if $a_{i+1}-a_i=1$, the following systems admit vector solutions $u(w)$ and $v(w)$ such that each of their $a_i+1$ components is a polynomial in $w$ of degree at most $b_i$.
    
    \scriptsize
    \begin{equation}\label{eq:sysla1}
        \begin{pmatrix}
c_0(w) & c_1(w) & \dots & c_{a_i}(w) \\
c_{-1}(w) & c_0(w) & \dots & c_{a_i-1}(w) \\
\vdots & \vdots & \ddots & \vdots \\
c_{-a_i+1}(w) & c_{-a_i+2}(w) & \dots & c_1(w)
\end{pmatrix} u(w) = 0, \quad \begin{pmatrix}
c_0(w) & c_{-1}(w) & \dots & c_{-a_i}(w) \\
c_1(w) & c_0(w) & \dots & c_{-a_i+1}(w) \\
\vdots & \vdots & \ddots & \vdots \\
c_{a_i-1}(w) & c_{a_i-2}(w) & \dots & c_{-1}(w) 
\end{pmatrix} v(w) = 0
    \end{equation}
    \normalsize 
    Similarly, if $b_{i+1}-b_i=1$, the following systems admit vector solutions $u(z)$ and $v(z)$ such that each of their $b_i+1$ components is a polynomial in $z$ of degree at most $a_i$.
        \scriptsize
    \begin{equation}\label{eq:sysla2}
        \begin{pmatrix}
d_0(z) & d_1(z) & \dots & d_{b_i}(z) \\
d_{-1}(z) & d_0(z) & \dots & d_{b_i-1}(z) \\
\vdots & \vdots & \ddots & \vdots \\
d_{-b_i+1}(z) & d_{-b_i+2}(z) & \dots & d_1(z) 
\end{pmatrix} u(z) = 0, \quad \begin{pmatrix}
d_0(z) & d_{-1}(z) & \dots & d_{-b_i}(z) \\
d_1(z) & d_0(z) & \dots & d_{-b_i+1}(z) \\
\vdots & \vdots & \ddots & \vdots \\
d_{b_i-1}(z) & d_{b_i-2}(z) & \dots & d_{-1}(z)
\end{pmatrix} v(z) = 0
    \end{equation}
    \normalsize 

\end{proposition}
\begin{proof}
   We begin with the case $a_{i+1}-a_i = 1$. Let $P_i$ and $Q_i$ be the polynomials associated with the $i$-th step of the bivariate QSP protocol, so that they satisfy
$\mathrm{multideg}\, P_i = \mathrm{multideg}\, Q_i = (a_i,b_i)$.
Accordingly, they admit the decomposition
\begin{equation}\label{eq:lade}
    P_i(z,w) = \sum_{k=0}^{a_i} u_k(w)\, z^k,
    \qquad
    Q_i(z,w) = \sum_{k=0}^{a_i} v_k(w)\, z^k,
\end{equation}
where each $u_k(w)$ and $v_k(w)$ is a polynomial in $w$ of degree at most $b_i$. Now observe that Proposition \ref{prop:formbiv2} implies that these polynomials satisfy
\begin{equation*}
\int_{\Gamma_w} \frac{P_i(z,w)\, z^{-k}}{z\, P_n(z,w)\, \tilde{Q}_n(z,w)} \, dz = 0,
\qquad
\int_{\Gamma_w} \frac{Q_i(z,w)\, z^{k-a_i}}{z\, P_n(z,w)\, \tilde{Q}_n(z,w)} \, dz = 0,
\end{equation*}
for all $0 \leq k < a_i$.
Substituting the decomposition \eqref{eq:lade} into these equations, and denoting by $u(w)$ and $v(w)$ the vectors whose entries are the coefficients $u_k(w)$ and $v_k(w)$, respectively, we conclude that \eqref{eq:sysla1} must hold. The second case follows analogously by exploiting the symmetry between $z$ and $w$.
\end{proof}

Since Proposition \eqref{prop:condibiv} provides a necessary condition that polynomials implementable by bivariate QSP must satisfy, it can be used to identify pairs of polynomials that lie beyond the scope of bivariate QSP. The following example illustrates how this proposition can be applied in practice.

\begin{example}
   The $FRT = QSP$ conjecture, introduced in \cite{rossi2022multivariable}, stated that any pair of polynomials $P(z,w)$ and $Q(z,w)$ satisfying $|P(z,w)|^2 + |Q(z,w)|^2 = 1$ for $(z,w) \in \mathbb{T} \times \mathbb{T}$ can be implemented by bivariate QSP. This conjecture was subsequently disproven, and a first counterexample was provided in \cite{nemeth2023variants}. After a change of variables and a rescaling by a global phase, this counterexample is equivalent to
    \begin{equation*}
    \begin{split}
        P(z,w) &= \frac{6}{25} \sqrt{\frac{37}{493}}\Big(w^2 z^2+\left(\frac{114}{37}+\frac{56 i}{37}\right) \left(w^2+z^2\right)-\left(\frac{122}{37}+\frac{8 i}{37}\right) \left(w^2 z+z\right) \\
        & \quad +\left(\frac{362}{111}-\frac{248 i}{111}\right) \left(w z^2+w\right)+\left(\frac{692}{111}-\frac{719 i}{222}\right) w z+1\Big)
    \end{split}
    \end{equation*}
    \begin{equation*}
        \begin{split}
            Q(z,w) &=  \frac{6}{25} \sqrt{\frac{37}{493}} \Big( w^2 z^2+\left(\frac{56}{37}+\frac{114 i}{37}\right) \left(w^2-z^2\right)\\
            & \quad +\left(-\frac{122}{37}-\frac{66 i}{37}\right) \left(w^2-z\right)+\left(\frac{362}{111}-\frac{418 i}{111}\right) \left(w z^2-w\right)-1 \Big)
        \end{split}
    \end{equation*}
    It is straightforward to verify that $|P(z,w)|^2 + |Q(z,w)|^2 = 1$ holds. We now provide a proof that the pair $P(U,V)$ and $Q(U,V)$ cannot be implemented by bivariate QSP based on Proposition \ref{prop:condibiv}.  First, we note that these polynomials satisfy $\mathrm{multideg}\, P = \mathrm{multideg}\, Q = (2,2)$. Therefore, if the pair $P(U,V)$ and $Q(U,V)$ were implementable by bivariate QSP, four steps would be required, with two steps $i$ such that $U_i = U$ and two steps such that $U_i = V$. Consequently, there would exist a step $i$ such that $(a_i,b_i) = (1,b_i)$ and $(a_{i+1},b_{i+1}) = (2,b_i)$. This corresponds to the first case discussed in Proposition \ref{prop:condibiv}. 
    It then implies that there would exist a vector $u(w) = (u_0(w), u_1(w))^t$ whose entries are polynomials of degree $b_i$ in $w$ such that
    \begin{equation}\label{eq:bivcond2x2}
        \begin{pmatrix}
            c_0(w) & c_1(w)
        \end{pmatrix}\begin{pmatrix}
            u_0(w) \\
            u_1(w)
        \end{pmatrix} = 0.
    \end{equation}
    We can show that such a vector $u(w)$ does not exist. For a fixed value of $w$, the two roots $z_\pm(w)$ of $P(z,w)$ are given by
    \small
    \begin{equation*}
        z_\pm(w) = \frac{(12+120 i) (1+w^2)-(154+205 i) w\pm (2+i) \sqrt{1296 w^4-30600 i w^3+4633 w^2-30600 i w+1296}}{4 \left((3+18 i) w^2+(50+52 i) w+(-18+60 i)\right)}.
    \end{equation*}
    \normalsize
    These can be used to compute the functions $c_0(w)$ and $c_1(w)$ given by \eqref{eq:lastdef1}, with $P_n(z,w) = P(z,w)$ and $\tilde{Q}_n(z,w) = z^{-2}w^{-2} Q(z,w)$. We have
    \small
    \begin{equation*}
        c_0(w) = \int_{\Gamma_w} \frac{z w^2}{ \, P(z,w) \, Q(z,w)} \, dz = 2 \pi i w^2 \left(\frac{z_+(w)}{P'(z_+(w),w) Q(z_+(w), w)} +\frac{z_-(w)}{P'(z_-(w),w) Q(z_-(w), w)}  \right),
    \end{equation*}
    \begin{equation*}
        c_1(w) = \int_{\Gamma_w} \frac{z^{2}w^2}{ \, P(z,w) \, Q(z,w)} \, dz = 2 \pi i w^2 \left(\frac{z_+(w)^2}{P'(z_+(w),w) Q(z_+(w), w)} +\frac{z_-(w)^2}{P'(z_-(w),w) Q(z_-(w), w)}  \right),
    \end{equation*}
    \normalsize
    where $\Gamma_w$ is a closed contour whose interior contains $z=0$ and $z=z_\pm(w)$. The right-hand side of these expressions follows from the residue theorem. Injecting these function in equation \eqref{eq:bivcond2x2}, we find that any solution for $u(w)$ is proportional to
    \begin{equation}
        u(w) = \begin{pmatrix}
           -((3+3 i) w+i) \left(4 w^2+57 w-36\right) \\
            4 (w+(3-3 i)) \left(3 w^2+19 b+3\right).
        \end{pmatrix}
    \end{equation}
    In particular, it is proportional to a polynomial of degree $3$ in $w$. Since $\mathrm{multideg}\, P = \mathrm{multideg}\, Q = (2,2)$, we have $b_i \leq 2 < 3$. We thus conclude that no solution $u(w)$ satisfies the condition of Proposition \ref{prop:condibiv}, and hence that $P(U,V)$ and $Q(U,V)$ cannot be achieved by bivariate QSP.
\end{example}

\section{Outlook}

%Primitives for block-encoding polynomial transformations of a unitary have become a central tool in the design of quantum algorithms. One established approach, based on the LCU primitive, realizes this task by decomposing the target polynomial in the monomial basis. In contrast to the QSP approach, it has the advantage of requiring little preprocessing, provided that this decomposition is known. 
This paper introduces a framework for block-encoding functions of a unitary operator through their series expansion in different polynomial bases. In contrast to the standard LCU approach relying on monomial expansions, our method employs as a basis the sequence of polynomials block-encoded at each step of a QSP protocol (or one of its variants). In particular, we show that the sequence of polynomials generated during an $n$-step QSP protocol forms a polynomial basis that enables the block-encoding of any polynomial $P(U)$ of degree at most $n$ using $n$ controlled applications of $U$.

Different variants of QSP give rise to distinct classes of polynomial bases that can be exploited for such expansions. We provide a general characterization of these bases in terms of their orthogonality and biorthogonality properties. In the case of generalized QSP, we prove that a circuit consisting of $n$ controlled applications of $U$ and $n+1$ single-qubit rotations encodes a family of $n+1$ Laurent biorthogonal polynomials. For $\mathrm{SU}(1,1)$-QSP, we show that a circuit composed of $n$ controlled applications of $U$ generates a sequence of $n+1$ orthogonal polynomials on the unit circle. We further introduce a new variant, termed OP-QSP, which is capable of implementing an arbitrary sequence of orthogonal polynomials.

These results are then leveraged to derive exact and explicit QSP angle sequences for several classical families of orthogonal polynomials, including Chebyshev, Hermite, Jacobi, and Rogers-Szeg\H{o} polynomials. This yields a method for block-encoding smooth functions of a unitary through their Hermite series expansion, achieving gate complexity that scales with the target precision $\epsilon$ as $O(\log(1/\epsilon))$, including $ O(\log(1/\epsilon))$ controlled applications of $U$, while the resulting block-encoding normalization constant scales as $O(\mathrm{polylog}(1/\epsilon))$.

These connections with the theory of orthogonal and biorthogonal polynomials also yield new results for QSP protocols. By exploiting tools such as the Gauss–Christoffel formula and the Bernstein–Szegő approximation, the angle-finding problem is solved for several variants of QSP. In particular, we obtain explicit expressions for the angles in terms of the moments of a linear functional constructed from the target polynomial.

We further provide a characterization of the class of polynomial block-encodings achievable via $\mathrm{SU}(1,1)$-QSP. Specifically, a polynomial in a unitary $P(U)$ of degree $n$ can be implemented using $n$ controlled applications of $U$ if and only if all roots of $P$ lie inside the unit circle. 

Lastly, we show that the sequence of polynomials block-encoded at each step of a bivariate QSP protocol satisfies biorthogonality relations. These relations are then used to derive a set of necessary conditions for a polynomial to be achievable via bivariate QSP. We further demonstrate that these conditions are violated by a known counterexample to the FRT = QSP conjecture. This serves as a sanity test for our results and illustrates that our criteria provides a broadly applicable necessary condition that encapsulates current understanding of achievability for multivariate QSP sequences.

There are several directions in which we expect this work to be useful. First, it provides a systematic pathway for identifying new families of polynomials that admit analytic solutions to the QSP angle-finding problem. As illustrated by the examples presented, the QSP angles for known sequences of orthogonal or biorthogonal polynomials can be constructed directly from their recurrence coefficients.

Second, it provides a framework for analyzing the complexity of block-encoding a function of a unitary from its series expansion in different polynomial bases. More precisely, our algorithm enables one to determine how the approximation error and the resulting block-encoding normalization scale in terms of the expansion coefficients.

Third, it introduces a set of tools for studying the expressivity of bivariate QSP protocols. Future work should investigate the constraints we derive on the class of bivariate polynomials implementable via bivariate QSP, as well as develop efficient methods to verify these constraints for a given target polynomial.

Finally, the close interplay between orthogonality, recurrence relations, and the structure of QSP-like circuits suggests a natural route toward new circuit generalizations. In particular, designing architectures that implement recurrence relations associated with bivariate biorthogonal Laurent polynomials, bivariate orthogonal polynomials on the unit circle, or matrix-valued orthogonal polynomials appears to be a promising direction for future research.

\section*{Acknowledgments}
PAB acknowledges support from a CQIQC postdoctoral fellowship and a postdoctoral fellowship from the Fonds de Recherche du Québec -- Nature et Technologies (FRQNT).
NW acknowledge the support from DOE, Office of Science, National Quantum Information Science Research Centers, Co-design Center for Quantum Advantage (C2QA) under Contract No.~DE-SC0012704 (Basic Energy Sciences, PNNL FWP 76274) and Pacific Northwest National Laboratory's Quantum Algorithms and Architecture for Domain Science (QuAADS) Laboratory Directed Research and Development (LDRD) Initiative. 

\appendix
\section{Proof of orthogonality relation in Example 3.2} \label{sec:app1}

In Example \ref{ex:twocheb}, we consider a sequence of polynomials given by the sum of two Chebyshev polynomials of the second kind. We show that this sequence is orthogonal with respect to the following form:
\begin{equation*}
   \mathcal{L}[f] = \frac{-\lambda}{4 \pi i}\int_{C}  f\left(\frac{z + z^{-1}}{2\gamma} + 1\right)\frac{(1- z^2)^2}{z(z^2 + \lambda)(1+ \lambda z^2 )} dz, \quad \lambda = 1 - 4 \gamma^2 a^2
\end{equation*}
Using $U_n\left( \frac{z}{2} + \frac{1}{2z}\right) = \frac{z^{n+1} -z^{-n-1}}{z -z^{-1}}$ and the definition of the polynomials $\hat{P}_i(x)$ in Example \ref{ex:twocheb}, we first note that for all $i \in \mathbb{N}$
\begin{equation*}
\begin{split} \hat{P}_{i}\left(\frac{(z+z^{-1})}{2\gamma} + 1\right) %&= \frac{1}{(2\gamma)^i(1- \delta_{i0}\lambda)}\left( \frac{z^{i+1} -z^{-i-1} + \lambda z^{i-1} -\lambda z^{-i+1}}{z -z^{-1}} \right) \\ 
&= \frac{1}{(2\gamma)^i(1- \delta_{i0}\lambda)}\left( \frac{z^{i-1}(z^2 + \lambda) -z^{-i-1}(1 + \lambda z^2)}{z -z^{-1}} \right) 
\end{split}
\end{equation*}
Next, the product of two polynomials can be rewritten in the following form,
\begin{equation*}
\begin{split}
  &\left( \frac{z^{n-1}(z^2 + \lambda) -z^{-n-1}(1 + \lambda z^2)}{z -z^{-1}} \right)\left( \frac{z^{m-1}(z^2 + \lambda) -z^{-m-1}(1 + \lambda z^2)}{z -z^{-1}} \right) \\
  & \quad= \left( \frac{z^{n+m-2}(z^2 + \lambda)^2 +z^{-n-m-2}(1 + \lambda z^2)^2 - (z^{n-m-2} + z^{-n+m-2})(z^2 + \lambda)(1 + \lambda z^2)}{(z -z^{-1})^2} \right)  
\end{split}
\end{equation*}
Substituting this expression into the definition of the form $\mathcal{L}$, we obtain a sum of three integrals
\begin{equation*}
\begin{split}
    \mathcal{L}[\hat{P}_n(x)\hat{P}_m(x)] & = -\frac{\lambda}{4 \pi i (2\gamma)^{n+m}(1- \delta_{n0}\lambda)} \int_C \frac{z^{n+m-1}(z^2 + \lambda)}{(1+ \lambda z^2 )} dz \\
    & - \frac{\lambda}{4 \pi i (2\gamma)^{n+m}(1- \delta_{n0}\lambda)} \int_C \frac{ z^{-n-m-1}(1 + \lambda z^2)}{(z^2 + \lambda)} dz \\
    & + \frac{\lambda}{4 \pi i (2\gamma)^{n+m}(1- \delta_{n0}\lambda)} \int_C  (z^{n-m-1} + z^{-n+m-1}) dz
\end{split}
\end{equation*}
Analyzing each of these integrals, we find that they simplify to the following expressions
\begin{equation*}
    \int_C  (z^{n-m-1} + z^{-n+m-1}) dz = 4 \pi i \delta_{nm}, \quad    \int_C \frac{z^{n+m-1}(z^2 + \lambda)}{(1+ \lambda z^2 )} dz = 2 \pi i \lambda \delta_{n0}\delta_{m0}.
\end{equation*}

\begin{equation*}
\begin{split}
    \int_C \frac{ z^{-n-m-1}(1 + \lambda z^2)}{(z^2 + \lambda)} dz &= \int_C z^{-n-m-1}\frac{ (1 -\lambda^2) }{(z^2 + \lambda)} + \lambda \int_C z^{-n-m-1}dz \\
    & =  (1 -\lambda^2)\int_C \frac{z^{-n-m-1} }{(z^2 + \lambda)}dz + 2 \pi i \lambda \delta_{n0}\delta_{m0}
\end{split}
\end{equation*}
Now, recall that the contour $C$ encloses the points $z =0$, $z = \pm i \sqrt{\lambda }$, and thus we obtain
\begin{equation*}
\begin{split}
    \int_C \frac{z^{-k+1} }{(z^2 + \lambda)} dz &= \int_C z^{-k} \left(\frac{ 1}{(z +  i \sqrt{\lambda})} + \frac{ 1}{(z -  i \sqrt{\lambda})} \right) dz \\
    & = 2 \pi i (-i \sqrt{\lambda})^{-k} + 2\pi i (-1)^{k-1}(i \sqrt{\lambda})^{-k} + 2 \pi i (i \sqrt{\lambda})^{-k} + 2\pi i (-1)^{k-1}(-i \sqrt{\lambda})^{-k} \\
    & = 0 
\end{split}
\end{equation*}
Summing the contributions of these integrals, we obtain
\begin{equation*}
     \mathcal{L}[\hat{P}_n(x)\hat{P}_m(x)]  = \frac{ \delta_{nm}   }{ (2\gamma)^{n+m}}.
\end{equation*}

\section{Proofs for Section \ref{sec:3}} \label{sec:appA}

In this section, we present the proofs of Propositions~\ref{prop:form}, \ref{prop:rec_from_biortho}, and \ref{prop:cond_rec_from_biortho}. We begin by establishing several auxiliary lemmas that will be used throughout the proofs

\begin{lemma}\label{lem:roots}
     Let $\{P_i(z,w)\, |\, i = 0,1,\dots,n \}$ and $\{Q_i(z,w)\, |\, i = 0,1,\dots,n \}$ be finite sequences of polynomials satisfying the recurrence relations  \eqref{eq:rec1} and \eqref{eq:rec2}, with $\text{deg}\, P_i = \text{deg}\, Q_i = i$. Then $P_i$ and $Q_i$ have no common roots.
\end{lemma}
\begin{proof}
   By substituting $P_{i}(z)=Q_{i}(z)=0$ into the recurrence relations
\eqref{eq:rec1} and \eqref{eq:rec2}, we find that this forces
$P_{i-1}(z)=Q_{i-1}(z)=0$. Repeating this argument $i$ times yields
$P_0(z)=Q_0(z)=0$. Since $P_0$ and $Q_0$ are nonzero constants, this
is a contradiction. We therefore conclude that $P_i$ and $Q_i$ have distinct roots.

\end{proof}

\begin{lemma}\label{lem:roots2}
 Let $\{P_i(z)\, |\, i = 0,1,\dots,n \}$ and $\{Q_i(z)\, |\, i = 0,1,\dots,n \}$ be finite sequences of polynomials satisfying the recurrence relations  \eqref{eq:rec1} and \eqref{eq:rec2}. Then, $Q_n(0) = 0$ if and only $\exists\, i \in \{0,1,\dots n\}$ s.t. $\theta_i =0 \mod \pi $ . 
\end{lemma}
\begin{proof}
By iterating equation~\eqref{eq:rec2}, we obtain
\begin{equation*}
    Q_i(0) = (-1)^{i-1} \prod_{j=1}^{i-1} \cos\theta_j \, Q_0(0).
\end{equation*}
Moreover, $Q_0(0)\neq 0$ since $Q_0(z)$ is a nonzero constant. Therefore,
the left-hand side vanishes if and only if $\cos\theta_j=0$ for some
$j$. The result follows.

\end{proof}
\noindent Next, we give the proof of Proposition \ref{prop:form}
\begin{proof}[Proof of Proposition \ref{prop:form}.]
  By inverting the recurrence relations \eqref{eq:rec1} and \eqref{eq:rec2}, one finds 
\begin{equation*}
    \begin{pmatrix}
        P_{i-1}(z) \\
        Q_{i-1}(z) 
    \end{pmatrix} = 
T(\theta_i, \phi_i, z)^{-1}
\begin{pmatrix}
        P_{i}(z) \\
        Q_{i}(z) 
    \end{pmatrix},
\end{equation*}
with the inverse of $T(\theta_i, \phi_i,z)$ given by
\begin{equation*}
T(\theta_i, \phi_i,z)^{-1}= 
    \begin{pmatrix}
 {e^{-i \phi_i } \cos (\theta_i )}z^{-1} & \sin (\theta_i )z^{-1}  \\
 e^{-i \phi_i } \sin (\theta_i ) & -\cos (\theta_i ) 
\end{pmatrix}.
\end{equation*}
This allows to express $P_i(z)$ and $Q_i(z)$ in terms of the last polynomials in the sequences, i.e.
\begin{equation*}
    \begin{pmatrix}
        P_i(z) \\
        Q_i(z) 
    \end{pmatrix} = 
T(\theta_{i+1}, \phi_{i+1},z)^{-1}T(\theta_{i+2}, \phi_{i+2},z)^{-1}\dots T(\theta_{n}, \phi_{n},z)^{-1}
\begin{pmatrix}
        P_{n}(z) \\
        Q_{n}(z) 
    \end{pmatrix}.
\end{equation*}
Now, we observe that
\begin{equation*}
   T(\theta_{i+1}, \phi_{i+1},z)^{-1}T(\theta_{i+2}, \phi_{i+2},z)^{-1}\dots T(\theta_{n}, \phi_{n},z)^{-1} = \begin{pmatrix}
        z^{-1}A(z) &  z^{-1} B(z) \\
        C(z) & D(z)
    \end{pmatrix},
\end{equation*}
where $A(z)$, $B(z)$, $C(z)$ and $D(z)$ are polynomials in $z^{-1}$ of degree at most $n-i-1$. This yields the following relations:
\begin{equation}\label{eq:trick1}
    \frac{P_i(z)}{z P_n(z) \tilde{Q}_n(z)} = \frac{A(z) z^{n-2}}{ Q_n(z)} +  \frac{B(z)  z^{n-2}}{ P_n(z)}, \quad \frac{\tilde{Q}_i(z)}{z P_n(z) \tilde{Q}_n(z)} = \frac{C(z)z^{n-i-1}}{  Q_n(z)} +  \frac{D(z)z^{n-i-1}}{ P_n(z)},
\end{equation}
where $\tilde{Q}_n(z) = z^{-n} {Q}_n(z)$. Let $\Gamma$ denote a closed contour whose interior contains $z=0$ and all roots of $P_n(z)$, but none of the roots of $Q_n(z)$.  Lemmas \ref{lem:roots} and \ref{lem:roots2} guarantee the necessary condition for the contour $\Gamma$ to exist.  Evaluating the residue at $z=0$, we finds for $j \leq i$:
\begin{equation*}
    \int_\Gamma \frac{C(z) z^{n-i-1+j}}{Q_n(z)} dz = 0, \quad
 \quad  \int_\Gamma \frac{A(z) z^{n-2-j}}{Q_n(z)} dz = \left\{
	\begin{array}{ll}
		\lim_{z \rightarrow 0} \frac{A(z)z^{n-i-1}}{Q_n(z)}  & \mbox{if } j = i, \\
		0 & \mbox{if } j < i.
	\end{array}
\right.
\end{equation*}
In terms of the variable $u = 1/z$, the contour $\Gamma$ corresponds to a closed contour $\Gamma'$ whose interior contains $u = 0$ and all the roots of polynomial $u^n Q_n(1/u)$, but none of the roots of $u^n P_n(1/u)$. Further, $B(1/u)$ and $D(1/u)$ are polynomials in $u$ of degree at most $n-i-1$. Using this change of variable and evaluating the residue in $u = 0$, we find for $0 \leq j \leq i$
\begin{equation*}
    \int_\Gamma \frac{B(z) z^{n-2-j}}{P_n(z)} dz = - \int_{\Gamma'} \frac{B(1/u) u^{j}}{u^nP_n(1/u)} du =0,
\end{equation*}
and
\begin{equation*}
    \int_\Gamma \frac{D(z) z^{n-i+j-1}}{P_n(z)} dz  =  -\int_{\Gamma'} \frac{D(1/u) u^{i-j-1}}{u^nP_n(1/u)} dz = \left\{
	\begin{array}{ll}
		-\lim_{z \rightarrow \infty} \frac{D(z)z^{i+1}}{P_n(z)}  & \mbox{if } j = i, \\
		0 & \mbox{if } j < i.
	\end{array}
\right.
\end{equation*}
From the value of these integrals and equation \eqref{eq:trick1}, we derive
\begin{equation*}
   \mathcal{L}[P_i(z)z^{-j}]  = \int_\Gamma  \frac{P_i(z)z^{-j}}{z P_n(z) \tilde{Q}_n(z)}  = \left\{
	\begin{array}{ll}
		\lim_{z \rightarrow 0} \frac{A(z)z^{n-i-1}}{Q_n(z)}  & \mbox{if } j = i, \\
		0 & \mbox{if } j < i,
	\end{array}
\right.
\end{equation*}
and
\begin{equation*}
   \mathcal{L}[\tilde{Q}_i(z)z^{j}]  = \int_\Gamma  \frac{\tilde{Q}_i(z)z^{j}}{z P_n(z) \tilde{Q}_n(z)}  =  \left\{
	\begin{array}{ll}
		-\lim_{z \rightarrow \infty} \frac{D(z)z^{i+1}}{P_n(z)}  & \mbox{if } j = i, \\
		0 & \mbox{if } j < i.
	\end{array}
\right.
\end{equation*}
Since $P_i(z)$ and $\tilde{Q}_i(z)$ respectively belong to $\text{span}\{1, z, \dots, z^i\}$ and $\text{span}\{1, z^{-1}, \dots, z^{-i}\}$, the biorthogonality relations follow,
\begin{equation*}
    \mathcal{L}[\tilde{Q}_i(z)P_j(z) ] \propto \delta_{ij}.
\end{equation*}
\end{proof}
Before proving Proposition~\ref{prop:rec_from_biortho}, we establish the following lemma.
\begin{lemma}\label{lem:equivrels}
   Let $\{P_i(z)\, |\, i = 0,1,\dots,n \}$ and $\{Q_i(z)\, |\, i = 0,1,\dots,n \}$ be finite sequences of polynomials satisfying $\deg P_i = \deg Q_i = i$. They satisfy
    \begin{equation*}
        \mathcal{L}[\tilde{Q}_i(z)P_j(z) ] \propto \delta_{ij}
    \end{equation*}
    if and only if, for all $j \leq i$
    \begin{equation}\label{eq:co1}
        \mathcal{L}[z^{-j}P_i(z)] \propto \delta_{ij}, \quad 
        \mathcal{L}[z^{-j}Q_i(z)] \propto \delta_{j0}.
    \end{equation}
\end{lemma}
\begin{proof}
    Recall that $\tilde{Q}_i(z) = z^{-i}Q_i(z)$. The result follows from the biorthogonality relation and 
    \begin{equation*}
        \text{span}\{1, z, \dots, z^i\} = \text{span}\{P_0(z), P_1(z), \dots, P_i(z)\} = \text{span}\{Q_0(z), Q_1(z), \dots, Q_i(z)\}.
    \end{equation*}
\end{proof}
\noindent Next, we give the proof of Proposition \ref{prop:rec_from_biortho}
\begin{proof}[Proof of Proposition \ref{prop:rec_from_biortho} ]
   First, fix the coefficients $\zeta_i$ and $\gamma_i$ so that $P_{i}(z) - \zeta_i z P_{i-1}(z)$ and $Q_{i}(z) - \gamma_i z Q_{i-1}(z)$ are polynomials of degree $i-1$ in $z$. For $j \leq i-1$, we have
\begin{equation*}
\mathcal{L}[z^{-j}(P_{i}(z) - \zeta_i z P_{i-1}(z))] \propto \delta_{j0},
\qquad
\mathcal{L}[z^{-j}(Q_{i}(z) - \gamma_i z Q_{i-1}(z))] \propto \delta_{j0}.
\end{equation*}
This shows that the polynomials $P_{i}(z) - \zeta_i z P_{i-1}(z)$ and $Q_{i}(z) - \gamma_i z Q_{i-1}(z)$ satisfy the same relation \eqref{eq:co1} with respect to $\mathcal{L}$ as $Q_{i-1}$ established in Lemma \ref{lem:equivrels}. Since they are also of degree $i-1$, they must be proportional; otherwise, the uniqueness of the sequences satisfying the biorthogonality relation would be violated, and the form is assumed to be biorthogonal quasi-definite. The result then follows.
\end{proof}

\begin{lemma}
    Let $P_{i-1}$, $Q_{i-1}$, $P_{i}$, and $Q_{i}$ be polynomials satisfying equation \eqref{eq:rec1alt}. Suppose there exist coefficients $\kappa_1$ and $\kappa_2$ such that, for all $z$ on the unit circle,
\begin{equation}\label{eq:condunit1}
   \kappa_1|P_{i}(z)|^2 + \kappa_2|Q_{i}(z)|^2 = 1.
\end{equation}
This holds if and only if there exist coefficients $\tau_1$ and $\tau_2$ such that, for all $z$ on the unit circle,
\begin{equation}\label{eq:condunit2}
   \tau_1|P_{i-1}(z)|^2 + \tau_2|Q_{i-1}(z)|^2 = 1,
\end{equation}
and the recurrence coefficients satisfy
\begin{equation}\label{eq:condunit3}
    V
    \begin{pmatrix}
       \kappa_1 & 0 \\
        0 &\kappa_2
    \end{pmatrix}
    V^\dagger
    =
    \begin{pmatrix}
        \tau_1 & 0 \\
        0 & \tau_2
    \end{pmatrix},
    \qquad
    V =
    \begin{pmatrix}
        \zeta_i & \gamma_i \\
        \beta_i & \delta_i
    \end{pmatrix}.
\end{equation}
\end{lemma}
\begin{proof}
    On the unit circle, we have $\bar{z} = z^{-1}$, and the monomials $z^k$ with $k \in \mathbb{Z}$ are linearly independent. Substituting the recurrence relation \eqref{eq:rec1alt} into equation \eqref{eq:condunit1} and comparing the coefficients of $z^{i}$ and $z^{-i}$, we obtain relations, for instance
\begin{equation*}
\kappa_1 \zeta_i \bar{\beta}_i + \kappa_2 \gamma_i \bar{\delta}_i = 0.
\end{equation*}
Equations \eqref{eq:condunit2} and \eqref{eq:condunit3} then follow from these relations.
The converse direction is obtained straightforwardly by substituting \eqref{eq:rec1alt} into the left-hand side of equation \eqref{eq:condunit1} and using equations \eqref{eq:condunit2} and \eqref{eq:condunit3}.
\end{proof}
We conclude with the proof of Proposition \ref{prop:cond_rec_from_biortho}.
\begin{proof}
Since $\hat{P}_0(z)=1$ and $\hat{Q}_0(z)=1$ are constant, one can fix $\alpha_0$ and $\tilde{\alpha}_0$ such that
\begin{equation*}
\left|\frac{\hat{P}_0(z)}{\alpha_0}\right|^2 + \left|\frac{\hat{Q}_0(z)}{\tilde{\alpha}_0}\right|^2= 1.
\end{equation*}
Using the explicit expressions of the recurrence coefficients $\zeta_i$, $\beta_i$, $\gamma_i$, and $\delta_i$ in terms of trigonometric functions in the generalized QSP recurrence relations \eqref{eq:rec1} and \eqref{eq:rec2}, we find that they satisfy \eqref{eq:condunit3} with $\kappa_1 = \kappa_2 = \tau_1 = \tau_2 = 1$. It then follows from the previous lemma that for all $i = 0,1,\dots, n$ we have
\begin{equation*}
|P_i(z)|^2 + |Q_i(z)|^2 = 1.
\end{equation*}
Conversely, if
\begin{equation*}
|P_{n}(z)|^2 + |Q_{n}(z)|^2 = 1,
\end{equation*}
then, by the previous lemma and an appropriate rescaling of the polynomials, we obtain $\det(V) = 1$ and  $V V^\dagger = \mathbb{I}$, where
\begin{equation*}
V =
\begin{pmatrix}
\zeta_i & \gamma_i\\
\beta_i & \delta_i
\end{pmatrix}.
\end{equation*}
Since $V \in U(2)$, its entries can be parametrized in terms of trigonometric functions so as to match the coefficients appearing in the generalized QSP recurrence relations \eqref{eq:rec1} and \eqref{eq:rec2}. The result then follows.
\end{proof}

\section{${\rm SU}(1,1)$-QSP on standard quantum device}\label{app:newappsu11}

The ${\rm SU}(1,1)$-QSP protocol requires the implementation of a single qubit rotation with an imaginary angle, i.e.
\begin{equation*}
    R(i\theta,\phi) = \begin{pmatrix}
        e^{i \phi} \cosh \theta &  i e^{i \phi} \sinh \theta  \\
       i \sinh \theta & -\cosh \theta
    \end{pmatrix} \otimes I.
\end{equation*}
This can be factorized as  $R(i\theta, \phi) = R(i \theta,0) R(0,\phi)$, where $R(0,\phi)$ is unitary that can be implemented using a single qubit rotation, and $ R(i \theta,0)$ is the following non-unitary ${\rm SU}(1,1)$ rotation
\begin{equation*}
   R(i \theta,0) =  \begin{pmatrix}
        \cosh \theta &  i \sinh \theta  \\
       i \sinh \theta & -\cosh \theta
    \end{pmatrix} \otimes I .
\end{equation*}
To implement this non-unitary transformation, one can use a block encoding approach. In particular, the matrix can be decomposed as a sum of two unitaries, $R(i\theta,0) = \cosh \theta \, Z \otimes I + i \sinh \theta \, X \otimes I,
$ and thus admits a block encoding via an LCU construction. This construction uses the following select and preparation oracles,
\begin{equation*}
    U_{\mathrm{sel}} = (I \otimes R_z(\pm \pi/2))\, \mathrm{CNOT} \, (I \otimes R_z(\mp \pi/2) Z ) = \ket{0}\bra{0} \otimes Z \pm \ket{1}\bra{1} \otimes i X
\end{equation*}
\begin{equation*}
    V_{\mathrm{prep}} = R_y(-\tilde{\theta}), \quad V_{\mathrm{prep}}^\dagger = R_y(\tilde{\theta}).
\end{equation*}
Using $R_y(-\tilde{\theta}) \ket{0} = \cos\frac{\tilde{\theta}}{2} \ket{0} - \sin\frac{\tilde{\theta}}{2} \ket{0}$, one finds that the unitary encoding the action of the circuit represented in Figure \ref{fig:su11rot} has the following block structure
\begin{equation}
   (\bra{0} \otimes I)  (V^\dagger_{\mathrm{prep}} \otimes I )U_{\mathrm{SEL}}(V_{\mathrm{prep}}\otimes I )(\ket{0} \otimes I)  = \begin{pmatrix}
        \cos^2\frac{\tilde{\theta}}{2} &  \pm i\sin^2\frac{\tilde{\theta}}{2} \\
     \pm i\sin^2\frac{\tilde{\theta}}{2}  & -\cos^2\frac{\tilde{\theta}}{2} 
    \end{pmatrix} .
\end{equation}
 Using the identification $\tanh \theta = \pm \tan^2(\tilde{\theta}/2)$, one finds that
\begin{equation*}
    \begin{pmatrix}
        \cos^2\frac{\tilde{\theta}}{2} &  \pm i\sin^2\frac{\tilde{\theta}}{2} \\
     \pm i\sin^2\frac{\tilde{\theta}}{2}  & -\cos^2\frac{\tilde{\theta}}{2} 
    \end{pmatrix} = \frac{1}{|\sinh \theta| + \cosh \theta}\begin{pmatrix}
        \cosh \theta &  i \sinh \theta  \\
       i \sinh \theta & -\cosh \theta
    \end{pmatrix}
\end{equation*}
so that the circuit in Figure \ref{fig:su11rot} implements the ${\rm SU}(1,1)$ rotation corresponding to $R(i\theta,0)$.

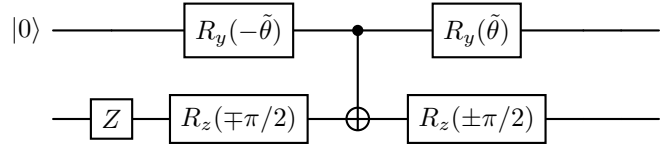
\begin{figure}[h]
    \centering
\begin{quantikz}[wire types = {q,q,q,n,q}]
\lstick{$\ket{0}$}  & &\gate{R_y(-\tilde{\theta})}  & \ctrl{1} & \gate{R_y(\tilde{\theta})}& & & \\ 
   & \gate{Z} & \gate{R_z(\mp \pi/2)}&  \targ{0}& \gate{R_z(\pm \pi/2)}&   && \\ 
\end{quantikz}
    \caption{Circuit associated to ${\rm SU}(1,1)$ rotations}
    \label{fig:su11rot}
\end{figure}

\section{Proofs for Section \ref{sec:4}}\label{sec:app2}

We now present the proof of Proposition~\ref{prop:formbiv1}. The argument is the bivariate analogue of the proof of Proposition~\ref{prop:form} in the univariate case.

\begin{proof}
First, we observe that
\begin{equation}
   T(\theta_{i+1},\phi_{i+1},z_{i+1})^{-1}T(\theta_{i+2}, \phi_{i+2},z_{i+2})^{-1}\dots T(\theta_{n}, \phi_{n}, z_{n})^{-1} = \begin{pmatrix}
        z_{i+1}^{-1}A(z,w) &  z_{i+1}^{-1} B(z,w) \\
        C(z,w) & D(z,w)
    \end{pmatrix},
\end{equation}
where $A(z,w)$, $B(z,w)$, $C(z,w)$ and $D(z,w)$ are polynomials in $(z^{-1},w^{-1})$ of multidegree at most $(a_n - a_{i+1}, b_n -b_{i+1} )$. Since we have
\begin{equation}
    \begin{pmatrix}
        P_i(z,w) \\
        Q_i(z,w) 
    \end{pmatrix} = 
T(\theta_{i+1}, \phi_{i+1},z_{i+1})^{-1}T(\theta_{i+2}, \phi_{i+2},z_{i+2})^{-1}\dots T(\theta_{n}, \phi_{n},z_{n})^{-1}
\begin{pmatrix}
        P_{n}(z,w) \\
        Q_{n}(z,w) 
    \end{pmatrix},
\end{equation}
we can derive the following relations:
\begin{equation}\label{eq:trick1biv}
    \frac{P_i(z,w)}{z P_n(z,w) \tilde{Q}_n(z,w)} = \frac{A(z,w) z^{a_n - 1}w^{b_n} }{ z_{i+1} Q_n(z,w)} +  \frac{B(z,w)  z^{a_n - 1}w^{b_n}}{ z_{i+1}P_n(z,w)},
\end{equation}
\begin{equation}\label{eq:trick2altbiv}
    \frac{\tilde{Q}_i(z)}{z P_n(z,w) \tilde{Q}_n(z,w)} = \frac{C(z,w)z^{a_n-a_i-1} w^{b_n-b_i}}{  Q_n(z,w)} +  \frac{D(z,w)z^{a_n-a_i-1} w^{b_n-b_i}}{ P_n(z,w)},
\end{equation}
 where $\tilde{Q}_i(z,w) = z^{-a_i} w^{-b_i}Q_i(z,w)$. Let $\Gamma$ denote a closed contour whose interior contains $z=0$ and all roots of $P_n(z,z^{a_n + 1})$, but none of the roots of $Q_n(z,z^{a_n + 1})$. The proofs of Lemmas~\ref{lem:roots} and~\ref{lem:roots2} can be adapted to establish the necessary condition for the existence of such a contour. Evaluating the residue at $z=0$, we finds for $k\leq a_i $ and $\ell \leq b_i$:
\begin{equation}
    \int_\Gamma \frac{C(z,z^{a_n + 1}) z^{a_n-a_i-1} z^{(a_n+1)(b_n-b_i)}}{Q_n(z,z^{a_n + 1})z^{-k} z^{-(a_n+1)\ell}} dz = 0, 
\end{equation}
\scriptsize{
\begin{equation}
\begin{split}
    \int_\Gamma \frac{A(z,z^{a_n + 1}) z^{a_n-1}z^{b_n (a_n+1)}}{z^{a_{i+1} - a_i} z^{(a_n+1)(b_{i+1} - b_i)}Q_n(z,z^{a_n + 1})z^{k} z^{(a_n+1)\ell}} &dz = \left\{
	\begin{array}{ll}
		\lim_{z \rightarrow 0} \frac{A(z,z^{a_n + 1}) z^{a_n-1}z^{b_n (a_n+1)}}{z^{a_{i+1} } z^{(a_n+1)b_{i+1}}Q_n(z,z^{a_n + 1})}  & \mbox{if } (k,\ell) = (a_i,b_i), \\
		0 & \mbox{otherwise}.
	\end{array}
\right.
\end{split}
\end{equation}
}
\normalsize
In terms of the variable $u = 1/z$, the contour $\Gamma$ corresponds to a closed contour $\Gamma'$ whose interior contains $u = 0$ and all the roots of polynomial $u^{a_n}u^{(a_n + 1)b_n} Q_n(1/u, 1/u^{a_n + 1})$, but none of the roots of $u^{a_n}u^{(a_n + 1)b_n} P_n(1/u, 1/u^{a_n + 1})$. Further, $B(1/u,1/u^{a_n + 1})$ and $D(1/u,1/u^{a_n + 1})$ are polynomials in $u$ of degree at most $a_n - a_{i+1} + (a_n + 1)(b_n - b_{i+1})  $. Using the change of variable $u = 1/z$ and evaluating the residue in $u = 0$, we find for $k\leq a_i $ and $\ell \leq b_i$
\begin{equation}
    \int_\Gamma \frac{B(z, z^{a_n + 1}) z^{a_n - 1-k}z^{(a_n + 1)(b_n -\ell)}}{z^{a_{i+1}  -a_i}z^{(a_n + 1)(b_{i+1}  -b_i)}P_n(z,z^{a_n + 1})} dz = - \int_{\Gamma'} \frac{B(1/u,1/u^{a_n + 1}) u^{k + a_{i+1} }u^{(a_n+1)(\ell + b_{i+1} )}}{u^{1+ a_n +  a_i}u^{(a_n+1)(b_n+  b_i)}P_n(1/u,1/u^{a_n + 1})} du =0,
\end{equation}
and
\small
\begin{equation}
\begin{split}
    \int_\Gamma \frac{D(z,z^{a_n + 1}) z^{a_n-a_i-1}z^{(a_n +1)(b_n-b_i)}}{z^{-k}z^{-(a_n+1)\ell}P_n(z,z^{a_n + 1})} dz  &=  -\int_{\Gamma'} \frac{D(1/u,1/u^{a_n + 1})u^{a_i}u^{(a_n +1)b_i}}{u^{k+a_n+1}u^{(a_n +1)(\ell +b_n)}P_n(1/u,1/u^{a_n + 1})} dz \\ 
    & = \left\{
	\begin{array}{ll}
		  - \lim_{u\rightarrow 0}\frac{D(1/u,1/u^{a_n + 1})}{u^{a_n}u^{(a_n +1)b_n}P_n(1/u,1/u^{a_n + 1})}& \mbox{if } (k,\ell) = (a_i,b_i), \\
		0 & \mbox{otherwise}.
	\end{array}
\right.
\end{split}
\end{equation}
\normalsize
From the value of these integrals and equations \eqref{eq:trick1biv} and \eqref{eq:trick2altbiv}, we derive for $k\leq a_i $ and $\ell \leq b_i$
\begin{equation}
   \mathcal{L}[P_i(z,w)z^{-k}w^{-\ell}]   = \left\{
	\begin{array}{ll}
		\lim_{z \rightarrow 0} \frac{A(z,z^{a_n + 1}) z^{a_n-1}z^{b_n (a_n+1)}}{z^{a_{i+1} } z^{(a_n+1)b_{i+1}}Q_n(z,z^{a_n + 1})}  & \mbox{if } (k,\ell) = (a_i,b_i), \\
		0 & \mbox{otherwise }
	\end{array}
\right.
\end{equation}
and
\begin{equation}
   \mathcal{L}[\tilde{Q}_i(z,w)z^{k}w^{\ell}]  =  \left\{
	\begin{array}{ll}
		 - \lim_{u\rightarrow 0}\frac{D(1/u,1/u^{a_n + 1})}{u^{a_n}u^{(a_n +1)b_n}P_n(1/u,1/u^{a_n + 1})} & \mbox{if } (k,\ell) = (a_i,b_i), \\
		0 & \mbox{otherwise}.
	\end{array}
\right.
\end{equation}
The result then follows by combining these relations with the monomial expansions of the polynomials $P_i(z,w)$ and $\tilde{Q}_i(z,w)$.
\end{proof}

\normalsize
 
 \noindent Next, we give the proof of Proposition \ref{prop:formbiv2}.

\begin{proof}
Once again, we have the relations \eqref{eq:trick1biv} and \eqref{eq:trick2altbiv}.
Let $\Gamma_w$ denote a closed contour whose interior contains $z=0$ and all the roots of $P_n(z,w)$, but none of the roots of $Q_n(z,w)$, where $w$ is regarded as a fixed parameter. The proofs of Lemmas~\ref{lem:roots} and~\ref{lem:roots2} can be adapted to establish the necessary condition for the existence of such a contour. Evaluating the residue at $z=0$, we finds for $k \leq a_i$:
\begin{equation}
\begin{split}
    \int_{\Gamma_\omega} \frac{A(z,w) z^{a_n }w^{b_n} z^{-k}}{ z_i Q_n(z,w)} dz  = 0
\end{split}
\end{equation}
and for $a_i - a_{i+1} \leq k \leq a_i$,
\begin{equation}
   \int_{\Gamma_\omega} \frac{C(z,w)z^{a_n-a_i} w^{b_n-b_i} z^k}{  Q_n(z,w)} dz = 0.
\end{equation}
Next, using the change of variable $u = 1/z$ we derive for $2 - a_{i+1} + a_i  \leq k$
\begin{equation}
    \int_{\Gamma_\omega} \frac{B(z,w) z^{a_n }w^{b_n} z^{-k}}{ z_i P_n(z,w)} dz = 0
\end{equation}
and for $ k  \leq a_i -2$
\begin{equation}
   \int_{\Gamma_\omega} \frac{D(z,w)z^{a_n-a_i} w^{b_n-b_i}z^k}{ P_n(z,w)} dz = 0.
\end{equation}
From these four integrals and the relations \eqref{eq:trick1biv} and \eqref{eq:trick2altbiv}, we obtain the following relations for $1 - a_{i+1} + a_i \le k \le a_i$, 
\begin{equation}
    \int_{\Gamma_\omega}\frac{P_i(z,w) z^{-k}}{ z P_n(z,w) \tilde{Q}_n(z,w)} dz = \delta_{a_i, k} f(w), \quad \int_{\Gamma_\omega}\frac{\tilde{Q}_i(z,w) z^{k}}{ z P_n(z,w) \tilde{Q}_n(z,w)} dz = \delta_{a_i, k} g(w).
\end{equation}
This concludes the demonstration.
\end{proof}

%\bibliographystyle{unsrt}
%\bibliography{ref.bib}

\begin{thebibliography}{10}

\bibitem{low2017optimal}
Guang~Hao Low and Isaac~L. Chuang.
\newblock {Optimal Hamiltonian simulation~by quantum signal processing}.
\newblock {\em Physical Review Letters}, 118(1):010501, 2017.
\newblock
  \href{https://doi.org/10.1103/PhysRevLett.118.010501}{doi.org/10.1103/PhysRevLett.118.010501}.

\bibitem{low2019hamiltonian}
Guang~Hao Low and Isaac~L Chuang.
\newblock {Hamiltonian simulation by qubitization}.
\newblock {\em Quantum}, 3:163, 2019.
\newblock \href{
  https://doi.org/10.22331/q-2019-07-12-163}{doi.org/10.22331/q-2019-07-12-163}.

\bibitem{liu2025toward}
Yuan Liu, John~M Martyn, Jasmine Sinanan-Singh, Kevin~C Smith, Steven~M Girvin,
  and Isaac~L Chuang.
\newblock {Toward Mixed Analog-Digital Quantum Signal Processing: Quantum AD/DA
  Conversion and the Fourier Transform}.
\newblock {\em IEEE Transactions on Signal Processing}, 2025.
\newblock
  \href{https://ieeexplore.ieee.org/abstract/document/11129874}{doi.org/10.1109/TSP.2025.3599462}.

\bibitem{chakraborty2019power}
Shantanav Chakraborty, Andr{\'a}s Gily{\'e}n, and Stacey Jeffery.
\newblock {The Power of Block-Encoded Matrix Powers: Improved Regression
  Techniques via Faster Hamiltonian Simulation}.
\newblock In {\em 46th International Colloquium on Automata, Languages, and
  Programming (ICALP 2019)}, pages 33--1. Schloss Dagstuhl--Leibniz-Zentrum
  f{\"u}r Informatik, 2019.
\newblock
  \href{https://drops.dagstuhl.de/entities/document/10.4230/LIPIcs.ICALP.2019.33}{doi.org/10.4230/LIPIcs.ICALP.2019.33}.

\bibitem{lin2020optimal}
Lin Lin and Yu~Tong.
\newblock {Optimal polynomial based quantum eigenstate filtering with
  application to solving quantum linear systems}.
\newblock {\em Quantum}, 4:361, 2020.
\newblock \href{
  https://doi.org/10.22331/q-2020-11-11-361}{doi.org/10.22331/q-2020-11-11-361}.

\bibitem{gilyen2019quantum}
Andr{\'a}s Gily{\'e}n, Yuan Su, Guang~Hao Low, and Nathan Wiebe.
\newblock {Quantum singular value transformation and beyond: exponential
  improvements for quantum matrix arithmetics}.
\newblock In {\em Proceedings of the 51st annual ACM SIGACT symposium on theory
  of computing}, pages 193--204, 2019.
\newblock
  \href{https://ieeexplore.ieee.org/document/7354428}{doi.org/10.1109/FOCS.2015.54}.

\bibitem{low2016methodology}
Guang~Hao Low, Theodore~J Yoder, and Isaac~L Chuang.
\newblock Methodology of resonant equiangular composite quantum gates.
\newblock {\em Physical Review X}, 6(4):041067, 2016.
\newblock
  \href{https://doi.org/10.1103/PhysRevX.6.041067}{doi.org/10.1103/PhysRevX.6.041067}.

\bibitem{marrero2026encoded}
Carlos~Ortiz Marrero, Rui~Jie Tang, and Nathan Wiebe.
\newblock {Encoded Quantum Signal Processing for Heisenberg-Limited Metrology}.
\newblock {\em arXiv preprint},
  \href{https://arxiv.org/abs/2603.22798}{arXiv:2603.22798}, 2026.

\bibitem{dong2022beyond}
Yulong Dong, Jonathan Gross, and Murphy~Yuezhen Niu.
\newblock {Beyond Heisenberg limit quantum metrology through quantum signal
  processing}.
\newblock {\em arXiv preprint},
  \href{https://arxiv.org/abs/2209.11207}{arXiv:2209.11207}, 2022.

\bibitem{rossi2022multivariable}
Zane~M. Rossi and Isaac~L. Chuang.
\newblock {Multivariable Quantum Signal Processing (M-QSP): Prophecies of the
  Two-Headed Oracle}.
\newblock {\em Quantum}, 6:811, 2022.
\newblock \href{
  https://doi.org/10.22331/q-2022-09-20-811}{doi.org/10.22331/q-2022-09-20-811}.

\bibitem{nemeth2023variants}
Balázs Németh, Blanka Kövér, Boglárka Kulcsár, Roland~Botond Miklósi,
  and András Gilyén.
\newblock {On Variants of Multivariate Quantum Signal Processing and Their
  Characterizations}.
\newblock {\em arXiv preprint},
  \href{https://arxiv.org/abs/2312.09072}{arXiv:2312.09072}, 2023.

\bibitem{gomes2024multivariable}
Niladri Gomes, Hokiat Lim, and Nathan Wiebe.
\newblock {Multivariable QSP and Bosonic Quantum Simulation Using Iterated
  Quantum Signal Processing}.
\newblock {\em arXiv preprint},
  \href{https://arxiv.org/abs/2408.03254}{arXiv:2408.03254}, 2024.

\bibitem{laneve2025multivariate}
Lorenzo Laneve and Stefan Wolf.
\newblock {On multivariate polynomials achievable with quantum signal
  processing}.
\newblock {\em Quantum}, 9:1641, 2025.
\newblock \href{
  https://doi.org/10.22331/q-2025-02-20-1641}{doi.org/10.22331/q-2025-02-20-1641}.

\bibitem{rossi2025modular}
Zane~M Rossi, Jack~L Ceroni, and Isaac~L Chuang.
\newblock {Modular quantum signal processing in many variables}.
\newblock {\em Quantum}, 9:1776, 2025.
\newblock
  \href{https://doi.org/10.22331/q-2025-06-18-1776}{doi.org/10.22331/q-2025-06-18-1776}.

\bibitem{harrow2009quantum}
Aram~W. Harrow, Avinatan Hassidim, and Seth Lloyd.
\newblock {Quantum Algorithm for Linear Systems of Equations}.
\newblock {\em Physical Review Letters}, 103(15):150502, 2009.
\newblock \href{
  https://doi.org/10.1103/PhysRevLett.103.150502}{doi.org/10.1103/PhysRevLett.103.150502}.

\bibitem{childs2017quantum}
Andrew~M Childs, Robin Kothari, and Rolando~D Somma.
\newblock {Quantum algorithm for systems of linear equations with exponentially
  improved dependence on precision}.
\newblock {\em SIAM Journal on Computing}, 46(6):1920--1950, 2017.
\newblock
  \href{https://doi.org/10.1137/16M1087072}{doi.org/10.1137/16M1087072}.

\bibitem{childs2012hamiltonian}
Andrew~M Childs and Nathan Wiebe.
\newblock {Hamiltonian Simulation Using Linear Combinations of Unitary
  Operations}.
\newblock {\em Quant. Inf. Comput.}, 12:0901--0924, 2012.
\newblock \href{ https://doi.org/10.48550/arXiv.1202.5822}{arXiv.1202.5822}.

\bibitem{berry2017quantum}
Dominic~W Berry, Andrew~M Childs, Aaron Ostrander, and Guoming Wang.
\newblock {Quantum algorithm for linear differential equations with
  exponentially improved dependence on precision}.
\newblock {\em Communications in Mathematical Physics}, 356(3):1057--1081,
  2017.
\newblock
  \href{https://doi.org/10.1007/s00220-017-3002-y}{doi.org/10.1007/s00220-017-3002-y}.

\bibitem{fang2023time}
Di~Fang, Lin Lin, and Yu~Tong.
\newblock {Time-marching based quantum solvers for time-dependent linear
  differential equations}.
\newblock {\em Quantum}, 7:955, 2023.
\newblock
  \href{https://doi.org/10.22331/q-2023-03-20-955}{doi.org/10.22331/q-2023-03-20-955}.

\bibitem{an2023linear}
Dong An, Jin-Peng Liu, and Lin Lin.
\newblock {Linear Combination of Hamiltonian Simulation for Nonunitary Dynamics
  with Optimal State Preparation Cost}.
\newblock {\em Physical Review Letters}, 131(15):150603, 2023.
\newblock
  \href{https://doi.org/10.1103/PhysRevLett.131.150603}{doi.org/10.1103/PhysRevLett.131.150603}.

\bibitem{berry2015hamiltonian}
Dominic~W. Berry, Andrew~M. Childs, and Robin Kothari.
\newblock {Hamiltonian simulation with nearly optimal dependence on all
  parameters}.
\newblock In {\em 2015 IEEE 56th annual symposium on foundations of computer
  science}, pages 792--809. IEEE, 2015.
\newblock
  \href{https://ieeexplore.ieee.org/document/7354428}{doi.org/10.1109/FOCS.2015.54}.

\bibitem{berry2014exponential}
Dominic~W Berry, Andrew~M Childs, Richard Cleve, Robin Kothari, and Rolando~D
  Somma.
\newblock Exponential improvement in precision for simulating sparse
  hamiltonians.
\newblock In {\em Proceedings of the forty-sixth annual ACM symposium on Theory
  of computing}, pages 283--292, 2014.
\newblock
  \href{https://doi.org/10.1145/2591796.259185}{doi.org/10.1145/2591796.259185}.

\bibitem{haah2019product}
Jeongwan Haah.
\newblock Product decomposition of periodic functions in quantum signal
  processing.
\newblock {\em Quantum}, 3:190, 2019.
\newblock \href{
  https://doi.org/10.22331/q-2019-10-07-190}{doi.org/10.22331/q-2019-10-07-190}.

\bibitem{chao2020finding}
Rui Chao, Dawei Ding, Andr{\'a}s Gily{\'e}n, Cupjin Huang, and Mario Szegedy.
\newblock Finding angles for quantum signal processing with machine precision.
\newblock {\em arXiv preprint arXiv:2003.02831},
  \href{https://arxiv.org/abs/2003.02831}{arXiv:2003.02831}, 2020.

\bibitem{dong2021efficient}
Yulong Dong, Xiang Meng, K~Birgitta Whaley, and Lin Lin.
\newblock {Efficient Phase-Factor Evaluation in Quantum Signal Processing}.
\newblock {\em Physical Review A}, 103(4):042419, 2021.
\newblock
  \href{https://doi.org/10.1103/PhysRevA.103.042419}{doi.org/10.1103/PhysRevA.103.042419}.

\bibitem{motlagh2024generalized}
Danial Motlagh and Nathan Wiebe.
\newblock {Generalized Quantum Signal Processing}.
\newblock {\em PRX Quantum}, 5(2):020368, 2024.
\newblock
  \href{https://doi.org/10.1103/PRXQuantum.5.020368}{doi.org/10.1103/PRXQuantum.5.020368}.

\bibitem{rossi2023quantum}
Zane~M. Rossi, Victor~M. Bastidas, William~J. Munro, and Isaac~L. Chuang.
\newblock {Quantum Signal Processing with Continuous Variables}.
\newblock {\em arXiv preprint},
  \href{https://arxiv.org/abs/2304.14383}{arXiv:2304.14383}, 2023.

\bibitem{barry2019constant}
Paul Barry.
\newblock {Constant Coefficient Laurent Biorthogonal Polynomials, Riordan
  Arrays and Moment Sequences}.
\newblock {\em arXiv preprint},
  \href{https://arxiv.org/abs/1906.06370}{arXiv:1906.06370}, 2019.

\bibitem{chihara2011introduction}
Theodore~S Chihara.
\newblock {\em {An introduction to orthogonal polynomials}}.
\newblock Courier Corporation, 2011.

\bibitem{hendriksen1991biorthogonal}
Erik Hendriksen and Olav Njåstad.
\newblock {Biorthogonal Laurent Polynomials with Biorthogonal Derivatives}.
\newblock {\em The Rocky Mountain Journal of Mathematics}, pages 301--317,
  1991.

\bibitem{zhedanov1998classical}
Alexei Zhedanov.
\newblock {The “Classical” Laurent Biorthogonal Polynomials}.
\newblock {\em Journal of Computational and Applied Mathematics},
  98(1):121--147, 1998.
\newblock
  \href{https://doi.org/10.1016/S0377-0427(98)00118-6}{doi.org/10.1016/S0377-0427(98)00118-6}.

\bibitem{kamioka2014laurent}
S~Kamioka.
\newblock {Laurent biorthogonal polynomials, q-Narayana polynomials and domino
  tilings of the Aztec diamonds}.
\newblock {\em Journal of Combinatorial Theory, Series A}, 123(1):14--29, 2014.
\newblock
  \href{https://www.sciencedirect.com/science/article/pii/S009731651300174X}{doi.org/10.1016/j.jcta.2013.11.002}.

\bibitem{simon2005orthogonal}
Barry Simon.
\newblock {\em {Orthogonal Polynomials on the Unit Circle}}.
\newblock American Mathematical Society, 2005.

\end{thebibliography}

\end{document}